\documentclass[12pt]{article}
\usepackage{geometry,amsmath, amsthm, amssymb, psfrag, hyperref, enumerate, bbm, nicefrac, comment, color, listings}
\usepackage{mathrsfs}
\usepackage{graphicx}

\usepackage{courier}
\makeatletter
\renewcommand\paragraph{\@startsection{paragraph}{4}{\z@}%
            {-2.5ex\@plus -1ex \@minus -.25ex}%
            {1.25ex \@plus .25ex}%
            {\normalfont\normalsize\bfseries}}
\makeatother
\newcommand{\R}{\ensuremath{\mathbb R}}
\renewcommand{\P}{\ensuremath{\mathbb P}}
\newcommand{\N}{\ensuremath{\mathbb N}}
\newcommand{\E}{\ensuremath{\mathbb E}}

\DeclareMathOperator*{\smallsum}{\textstyle\sum}
\DeclareMathOperator*{\smallprod}{\textstyle\prod}

\newtheorem{theorem}{Theorem}[section]

\newtheorem{proposition}[theorem]{Proposition}
\newtheorem{corollary}[theorem]{Corollary}
\newtheorem{lemma}[theorem]{Lemma}

\newtheorem{algo}[theorem]{Framework}

\newcommand\yesnumber{\refstepcounter{equation}\tag{\theequation}}

\title{Solving high-dimensional optimal stopping\\ problems using deep learning}

\author{
    Sebastian Becker$^1$,  
    Patrick Cheridito$^2$,\\
	Arnulf Jentzen$^3$,
    and Timo Welti$^4$
	\bigskip
	\\
	\small{$^1$ RiskLab, Department of Mathematics, ETH Z\"urich, Switzerland;}\\
	\small{email: \texttt{sebastian.becker\textcircled{\texttt{a}}math.ethz.ch}}
	\smallskip
	\\
	\small{$^2$ RiskLab, Department of Mathematics, ETH Z\"urich, Switzerland;}\\
	\small{email: \texttt{patrick.cheridito\textcircled{\texttt{a}}math.ethz.ch}}
	\smallskip
	\\
    \small{$^3$ Seminar for Applied Mathematics, Department of Mathematics, ETH Z\"urich, Switzerland;}\\
    \small{Faculty of Mathematics and Computer Science, University of M\"unster, Germany;}\\
    \small{email: \texttt{a.j\textcircled{\texttt{a}}uni-muenster.de}}
	\smallskip
	\\
    \small{$^4$ Seminar for Applied Mathematics, Department of Mathematics, ETH Z\"urich, Switzerland;}\\
	\small{D ONE Solutions AG, Z\"urich, Switzerland;}\\
    \small{email: \texttt{contact\textcircled{\texttt{a}}twelti.org}}
}

\usepackage{etoolbox}
\makeatletter
\patchcmd{\l@section}{1.0em}{0.85em}{}{}
\makeatother

\begin{document}

\maketitle

\begin{abstract}
Nowadays many financial derivatives, such as American or Bermudan options,
are of early exercise type. Often the pricing of early exercise options gives rise to high-dimensional 
optimal stopping problems, since the dimension corresponds to the number of underlying assets.
High-dimensional optimal stopping problems are, however, notoriously difficult to solve due to the well-known curse of dimensionality.
In this work, we propose an algorithm for solving such problems, which is based on deep learning and
computes, in the context of early exercise option pricing, both approximations of an optimal exercise strategy
and the price of the considered option. The proposed algorithm can also be applied to optimal stopping problems 
that arise in other areas where the underlying stochastic process can be efficiently simulated.
We present numerical results for a large number of example problems,
which include the pricing of many high-dimensional American and Bermudan options,
such as Bermudan max-call options in up to 5000~dimensions.
Most of the obtained results are compared to reference values computed by exploiting the specific problem design
or, where available, to reference values from the literature. These numerical results suggest that the proposed 
algorithm is highly effective in the case of many underlyings, in terms of both accuracy and speed.
\end{abstract}

\begin{center}
\begin{tabbing}
\>
option pricing,
optimal stopping,
curse of dimensionality,
deep learning
\end{tabbing}
\end{center}

\pagebreak
\tableofcontents

\section{Introduction}

Nowadays many financial derivatives, such as American or Bermudan options,
are of early exercise type. Contrary to European options, the holder of such an option has the right to 
exercise before the time of maturity. In models from mathematical finance
for the appropriate pricing of early exercise options, this aspect gives rise to optimal stopping problems.
The dimension of such optimal stopping problems can often be quite high
since it corresponds to the number of underlying assets.
Due to the curse of dimensionality (cf.\ Bellman~\cite{Bellman1957}),
high-dimensional optimal stopping problems are, however, notoriously difficult to solve.
Such optimal stopping problems can in nearly all cases not be solved explicitly
and it is an active topic of research to design and analyse approximation methods
which are capable of approximately solving possibly high-dimensional optimal stopping problems.
Many different approaches for numerically solving optimal stopping problems
and, in particular,
American and Bermudan option pricing problems
have been studied in the literature;
cf., e.g.,
\cite{Tilley1993,
DavisKaratzas1994,
BarraquandMartineau1995,
Carriere1996,
AitSahliaCarr1997,
BG97,
TsitsiklisVanRoy1999,
Andersen2000,
BroadieGlassermanHa2000,
LongstaffSchwartz2001,
TsitsiklisRoy2001,
Rogers2002,
Schweizer2002,
BallyPages2003b,
Garcia2003,
AndersenBroadie2004,
BouchardTouzi2004,
BroadieGlasserman2004,
Glasserman2004,
HaughKogan2004,
Egloff2005,
Firth2005,
GobetLemorWarin2005,
BenderKolodkoSchoenmakers2006,
KolodkoSchoenmakers2006,
PeskirShiryaev2006,
ChenGlasserman2007,
EgloffKohlerTodorovic2007,
Jamshidian2007,
BenderKolodkoSchoenmakers2008,
BS08,
BC08,
Kohler2008,
KohlerKryzakWalk2008,
LambertonLapeyre2008,
LordFangBervoetsOosterlee2008,
BelomestnyBendnerSchoenmakers2009,
Kohler2010,
KohlerKrzyzakTodorovic2010,
Rogers2010,
Belomestny2011,
Belomestny2011a,
LabartLelong2011arXiv,
DesaiFariasMoallemi2012,
JainOosterlee2012,
KohlerKryzak2012,
Be13,
BSD13,
SchoenmakersZhangHuang2013,
Christensen2014,
GuyonHenry2014,
BelomestnyDickmannNagapetyan2015,
BelomestnyLadkauSchoenmakers2015,
JiangPowell2015,
Lelong2018,
KulikovGusyatnikov2016,
BenderSchweizerZhuo2017,
CompanyEgorovaJodarSoleymani2017,
SirignanoSpiliopoulos2017,
BayerTemponeWolfers2018arXiv,
BeckerCheriditoJentzen2019,
GoldbergChen2018arXiv,
SirignanoSpiliopoulos2018,
WangChenSudjiantoLiu2018arXiv,
BayerHappolaTempone2019,
FujiiTakahashiTakahashi2019,
GoudenegeMolentZanette2019arXiv,
LapeyreLelong2019arXiv,
BeckerCheriditoJentzen2020,
Lelong2019arXiv,
ChenWan2019arXiv}.
For example,
such approaches include
approximating the Snell envelope or continuation values
(cf., e.g.,
\cite{Tilley1993,
BarraquandMartineau1995,
Carriere1996,
LongstaffSchwartz2001}),
computing optimal exercise boundaries
(cf., e.g.,
\cite{Andersen2000}),
and dual methods
(cf., e.g.,
\cite{Rogers2002,
HaughKogan2004}).
Whereas in~\cite{HaughKogan2004,KohlerKrzyzakTodorovic2010}
artificial neural networks with one hidden layer
were employed to approximate continuation values,
more recently
numerical approximation methods
for
American and Bermudan option pricing 
that are based on deep learning
were introduced;
cf., e.g.,
\cite{SirignanoSpiliopoulos2018,
SirignanoSpiliopoulos2017,
BeckerCheriditoJentzen2019,
FujiiTakahashiTakahashi2019,
BeckerCheriditoJentzen2020,
LapeyreLelong2019arXiv,
ChenWan2019arXiv}.
More precisely,
in~\cite{SirignanoSpiliopoulos2018,SirignanoSpiliopoulos2017}
deep neural networks are used to
approximately solve the
corresponding obstacle partial differential equation problem,
in~\cite{BeckerCheriditoJentzen2019}
the corresponding optimal stopping problem
is tackled directly
with a deep learning-based algorithm,
\cite{FujiiTakahashiTakahashi2019}
applies an extension of the deep BSDE solver from~\cite{HanJentzenE2018,EHanJentzen2017}
to the corresponding
reflected backward stochastic differential equation problem,
\cite{ChenWan2019arXiv}
suggests a different deep learning-based algorithm that relies on discretising backward stochastic differential equations,
and in~\cite{BeckerCheriditoJentzen2020,LapeyreLelong2019arXiv}
deep neural network-based variants of the classical algorithm
introduced by Longstaff \& Schwartz~\cite{LongstaffSchwartz2001}
are examined.

In this work, we propose an algorithm
for solving general possibly high-dimensional optimal stopping problems;
cf.~Framework~\ref{algo:general} in Subsection~\ref{sec:general_case}.
In spirit, it is similar to the algorithm introduced in~\cite{BeckerCheriditoJentzen2019}.
The proposed algorithm is based on deep learning
and computes both approximations of an optimal stopping strategy
and the optimal expected pay-off
associated with the considered optimal stopping problem.
In the context of pricing early exercise options,
these correspond to approximations of an optimal exercise strategy
and the price of the considered option, respectively.
The derivation and implementation of the proposed algorithm consist of essentially the following three steps.
\begin{enumerate}[(I)]
\item
\label{item:reformulation}
A neural network architecture for, in an appropriate sense, `randomised' stopping times
(cf.\ \eqref{eq:randomised_stopping_time} in Subsection~\ref{sec:NN})
is established in such a way that
varying the neural network parameters
leads to different randomised stopping times
being expressed.
This neural network architecture is used to
replace the supremum of the expected pay-off over suitable stopping times
(which constitutes the generic optimal stopping problem)
by the supremum of a suitable objective function over neural network parameters
(cf.\ \eqref{eq:objective_motivation}--\eqref{eq:objective_function}
in Subsection~\ref{sec:objective_function}).
\item
\label{item:SGD}
A stochastic gradient ascent-type optimisation algorithm
is employed to compute neural network parameters
that approximately maximise the objective function
(cf.\ Subsection~\ref{sec:SGD}).
\item
\label{item:final_approximations}
From these neural network parameters
and the corresponding randomised stopping time,
a true stopping time is constructed
which serves as the approximation of an optimal stopping strategy
(cf.\ \eqref{eq:stopping_time_construction} and~\eqref{eq:approx_optimal_exercise_strategy} in Subsection~\ref{sec:optimal_exercise_time}).
In addition,
an approximation of the optimal expected pay-off
is obtained by computing 
a Monte Carlo approximation of the expected pay-off under this approximately optimal stopping strategy
(cf.\ \eqref{eq:price_final_approximation} in Subsection~\ref{sec:optimal_exercise_time}).
\end{enumerate}
It follows from~\eqref{item:final_approximations}
that the proposed algorithm
computes a low-biased approximation of the optimal expected pay-off
(cf.~\eqref{eq:low_biased} in Subsection~\ref{sec:optimal_exercise_time}).
Yet, a large number of numerical experiments
where a reference value is available
(cf.~Section~\ref{sec:examples})
show that the bias appears to become small quickly during training
and that a very satisfying accuracy can be achieved in short computation time,
even in high dimensions
(cf.\ the end of this introduction below for a brief overview of the numerical computations that were performed).
Moreover,
in~\eqref{item:reformulation}
we resort to randomised stopping times
in order to circumvent the discrete nature of stopping times that attain only finitely many different values.
As a result, it is possible in \eqref{item:SGD} to tackle the arising optimisation problem
with a stochastic gradient ascent-type algorithm.
Furthermore,
while the focus in this article lies on American and Bermudan option pricing,
the proposed algorithm can also be applied to optimal stopping problems that arise in other areas
where the underlying stochastic process can be efficiently simulated.
Apart from this,
we only rely on the assumption that
the stochastic process to be optimally stopped
is a Markov process
(cf.~Subsection~\ref{sec:NN}).
But this assumption is no substantial restriction
since, on the one hand, it is automatically fulfilled in many relevant problems
and, on the other hand,
a discrete stochastic process that is not a Markov process
can be replaced by a Markov process of higher dimension that aggregates all necessary information
(cf., e.g., \cite[Subsection~4.3]{BeckerCheriditoJentzen2019}
and, e.g.,
Subsection~\ref{sec:example_TsitsiklisVanRoy1999}).

Next we compare our algorithm to the one introduced in~\cite{BeckerCheriditoJentzen2019}.
The latter splits the original problem into smaller optimal stopping problems at each time step
where stopping is permitted and decides to stop at that point in time or later
(cf.~\cite[(4)~in Subsection~2.1]{BeckerCheriditoJentzen2019}).
Starting at maturity, these auxiliary problems are solved recursively backwards until the initial time is reached.
Thereby, in every new step, neural network parameters are learned
for an objective function that depends, in particular, on the parameters found in the previous steps
(cf.~\cite[Subsection~2.3]{BeckerCheriditoJentzen2019}).
In contrast,
in~\eqref{item:reformulation} a single objective function is designed.
This objective function
allows to search in~\eqref{item:SGD} for neural network parameters that
maximise the expected pay-off
simultaneously over (randomised) stopping times which may decide to stop at any of the admissible points in time.
Therefore,
the algorithm proposed here does not rely on a recursion over the different time points.
In addition,
the construction of the final approximation of an optimal stopping strategy
in~\eqref{item:final_approximations}
differs from a corresponding construction in~\cite{BeckerCheriditoJentzen2019}.
We refer to Subsection~\ref{sec:example_GA_put}
for a comparison between the two algorithms with respect to performance.

The remainder of this article is organised as follows.
In Section~\ref{sec:main_ideas}, we present the main ideas
from which the proposed algorithm is derived.
More specifically,
in Subsection~\ref{ssec:problem} we illustrate
how an optimal stopping problem in the context of American option pricing is typically formulated.
Thereafter,
a replacement of this continuous-time problem
by a corresponding discrete time optimal stopping problem is discussed by means of an example
in Subsection~\ref{sec:temporal_discretization}.
Subsection~\ref{sec:factorization_lemma} is devoted to the statement and proof of an elementary but crucial result
about factorising general discrete stopping times
in terms of compositions of measurable functions
(cf.\ Lemma~\ref{lem:yg_squared}),
which lies at the heart of the neural network architecture we propose in Subsection~\ref{sec:NN}
to approximate general discrete stopping times.
This construction, in turn, is exploited in Subsection~\ref{sec:objective_function}
to transform the discrete time optimal stopping problem from Subsection~\ref{sec:temporal_discretization}
into the search of a maximum of a suitable objective function
(cf.~\eqref{item:reformulation} above).
In Subsection~\ref{sec:SGD}, we suggest to employ 
stochastic gradient ascent-type optimisation algorithms
to find approximate maximum points of the objective function
(cf.~\eqref{item:SGD} above).
As a last step,
we explain in Subsection~\ref{sec:optimal_exercise_time}
how we calculate final approximations of both the American option price and an optimal exercise strategy
(cf.~\eqref{item:final_approximations} above).
In Section~\ref{sec:algorithm_details}, we introduce the proposed algorithm
in a concise way, first for a special case for the sake of clarity (cf.~Subsection~\ref{sec:specific})
and second in more generality so that, in particular, a rigorous description of our implementations is fully covered (cf.~Subsections~\ref{sec:general_case}--\ref{sec:comments}).
Following this,
in Section~\ref{sec:examples}
first a few theoretical results are presented (cf.~Subsection~\ref{sec:theory}),
which are used to design numerical example problems and to provide reference values.
Thereafter, we describe in detail
a large number of example problems,
on which our proposed algorithm was tested,
and present numerical results for each of these problems.
In particular,
the examples include
the optimal stopping of Brownian motions (cf.~Subsection~\ref{sec:example_BM}),
the pricing of certain exotic American geometric average put and call-type options (cf.~Subsection~\ref{sec:geometric_average}),
the pricing of Bermudan max-call options in up to 5000~dimensions (cf.~Subsection~\ref{sec:max_call}),
the pricing of an American strangle spread basket option in five dimensions (cf.~Subsection~\ref{sec:strangle_spread}),
the pricing of an American put basket option in Dupire's local volatility model in five dimensions (cf.~Subsection~\ref{sec:example_Dupire}),
and the pricing of an exotic path-dependent financial derivative of a single underlying, which is modelled as a 100-dimensional optimal stopping problem (cf.~Subsection~\ref{sec:example_TsitsiklisVanRoy1999}).
The numerical results for the examples in
Subsections~\ref{sec:example_BM_Bermudan},
\ref{sec:example_BM_American},
\ref{sec:example_GA_put},
\ref{sec:example_GA_call2},
\ref{sec:example_GA_call1},
and
\ref{sec:max-call_equity}
are compared to calculated reference values that can be easily obtained due to
the specific design of the considered optimal stopping problem.
Moreover, the examples in
Subsections~\ref{sec:example_GA_call2},
\ref{sec:max_std},
\ref{sec:max-call_equity},
\ref{sec:strangle_spread},
\ref{sec:example_Dupire},
and
\ref{sec:example_TsitsiklisVanRoy1999}
are taken from the literature, and
our corresponding numerical results are compared to reference values from the literature (where available).
Finally, {\sc Python} source codes employed to generate the numerical results in Subsections~\ref{sec:max_std}, \ref{sec:max-call_equity}, \ref{sec:strangle_spread}, \ref{sec:example_Dupire}, and~\ref{sec:example_TsitsiklisVanRoy1999}
can be found in Section~\ref{sec:source_codes}.

\section{Main ideas of the proposed algorithm}
\label{sec:main_ideas}

In this section, we outline the main ideas that lead to the formulation of the proposed algorithm
in Subsections~\ref{sec:specific}--\ref{sec:general_case}
by considering the example of pricing an American option.
The proposed algorithm in Framework~\ref{algo:general} in Subsection~\ref{sec:general_case} is, however, general enough to also be applied to optimal stopping problems
where there are no specific assumptions on the dynamics of the underlying stochastic process,
as long as it can be cheaply simulated
(cf.~Subsection~\ref{sec:comments}).
Furthermore,
often in practice
and, in particular, in the case of Bermudan option pricing
(cf.~many of the examples in Section~\ref{sec:examples}),
the optimal stopping problem of interest is not a continuous-time problem but is already formulated in discrete time.
In such a situation, there is no need for
a time discretisation, as described in Subsection~\ref{sec:temporal_discretization} below, and the proposed algorithm in Framework~\ref{algo:general} can be applied directly.

\subsection{The American option pricing problem}
\label{ssec:problem}

Let $ T \in (0,\infty) $, $ d \in \N = \{ 1, 2, 3, \ldots \} $,
let $ ( \Omega, \mathcal{F}, \P ) $ be a probability space
with a filtration $ \mathscr{F} = ( \mathscr{F}_t )_{ t \in [ 0, T ] } $
that satisfies the usual conditions
(cf., e.g., \cite[Definition~2.25 in Section~1.2]{ks91}),
let $ \xi \colon \Omega \to \R^d $
be an $\mathscr{F}_0$/$\mathcal{B}( \R^d )$-measurable function which satisfies
for all
$ p \in ( 0, \infty ) $
that
$ \E\bigl[
    \lVert \xi \rVert_{ \R^d }^p
\bigr] < \infty $,
let $ W \colon [0,T] \times \Omega \to \R^d $ be a standard
$ ( \Omega, \mathcal{F}, \P, \mathscr{F} ) $-Brownian motion with continuous sample paths,
let $ \mu \colon \R^d \to \R^d $
and $ \sigma \colon \R^d \to \R^{ d \times d } $
be Lipschitz continuous functions, 
let $ X \colon [0,T] \times \Omega \to \R^d $ 
be an $ \mathscr{F} $-adapted continuous solution process of the stochastic differential equation
\begin{equation}
\label{eq:FSDE}
  d X_t = \mu( X_t ) \, dt + \sigma( X_t ) \, dW_t
  ,
  \qquad
  X_0 = \xi ,
  \qquad
  t \in [0,T] ,
\end{equation}
let $ \mathbb{F} = ( \mathbb{F}_t )_{ t \in [0,T] } $ be the
filtration generated by $ X $,
and let $ g \colon [0,T] \times \R^d \to \R $ be a continuous and at most polynomially growing function.
We think of $ X $ as a model for the 
price processes of $ d $ underlyings (say, $ d $ stock prices) 
under the risk-neutral pricing measure $ \P $
(cf., e.g., Kallsen~\cite{Kallsen2009})
and we are interested in approximatively pricing the American option on the process 
$ ( X_t )_{ t \in [0,T] } $ 
with the discounted pay-off function $ g \colon [0,T] \times \R^d \to \R $, that is, we intend to 
compute the real number
\begin{equation}
\label{eq:option_price}
  \sup\!\left\{ 
    \E\big[ 
      g( \tau, X_{ \tau } )
    \big]
    \colon
    \substack{
      \tau \colon \Omega \to [0,T] \text{ is an}
    \\
      \mathbb{F}\text{-stopping time}
    }
  \right\}
  .
\end{equation}
In addition to the \emph{price} of the American option in the model~\eqref{eq:FSDE},
there is also a high demand from the financial engineering industry 
to compute an approximately \emph{optimal exercise strategy},
that is, to compute a stopping time which approximately 
reaches the supremum in \eqref{eq:option_price}.

In a very simple example of \eqref{eq:FSDE}--\eqref{eq:option_price}, we can think 
of an \emph{American put option} in the one-dimensional Black--Scholes model, in which there are
an interest rate $ r \in \R $,
a dividend yield $ \delta \in [ 0, \infty ) $,
a volatility $ \beta \in (0,\infty) $, 
and a strike price $ K \in (0,\infty) $
such that it holds 
for all $ x \in \R $, $ t \in [ 0, T ] $ that
$ d = 1 $, $ \mu( x ) = ( r - \delta ) \, x $, $ \sigma( x ) = \beta \, x $, 
and 
$ g( t, x ) = e^{ -r t } \max\{ K - x , 0 \} $.

\subsection{Temporal discretisation}
\label{sec:temporal_discretization}

To derive the proposed approximation algorithm,
we first apply the Euler--Maruyama scheme to the stochastic 
differential equation \eqref{eq:FSDE} 
(cf.\ \eqref{eq:delta_Xb}--\eqref{eq:curlX} below)
and we employ a suitable time discretisation for 
the optimal stopping problem~\eqref{eq:option_price}.
For this let 
$ N \in \N $ be a natural number and let 
$ t_0, t_1, \dots, t_N \in [0,T] $ be real numbers with
\begin{equation}
  0 = t_0 < t_1 < \dots < t_N = T 
\end{equation}
(such that the maximal mesh size 
$ 
  \max_{ n \in \{ 0, 1, \dots, N - 1 \} } 
  ( t_{ n + 1 } - t_n ) 
$ 
is sufficiently small).
Observe that \eqref{eq:FSDE}
ensures that 
for all $ n \in \{ 0, 1, \dots, N - 1 \} $ it holds $ \P $-a.s.\ that
\begin{equation}
\label{eq:delta_X}
\begin{split}
  X_{ t_{ n + 1 } }
& =
  X_{ t_n }
  +
  \int_{ t_n }^{ t_{ n + 1 } }
  \mu( X_s ) 
  \, ds
  +
  \int_{ t_n }^{ t_{ n + 1 } }
  \sigma( X_s ) 
  \, dW_s
  .
\end{split}
\end{equation}
Note that \eqref{eq:delta_X}
suggests for every $ n \in \{ 0, 1, \dots, N - 1 \} $ that
\begin{equation}
\label{eq:delta_Xb}
\begin{split}
  X_{ t_{ n + 1 } }
& 
\approx
  X_{ t_n }
  +
  \mu( X_{ t_n } ) 
  \left( t_{ n + 1 } - t_n \right)
  +
  \sigma( X_{ t_n } ) 
  \left( W_{ t_{ n + 1 } } - W_{ t_n } \right)
  .
\end{split}
\end{equation}
The approximation scheme associated with \eqref{eq:delta_Xb} is referred to as 
the Euler--Maruyama scheme in the literature 
(cf., e.g., Maruyama~\cite{m55} and Kloeden \& Platen~\cite{KloedenPlaten1992}).
More formally,
let
$ 
  \mathcal{X} = ( \mathcal{X}^{(1)}, \ldots, \mathcal{X}^{(d)} )
  \colon 
  \{ 0, 1, \dots, N \} \times \Omega \to \R^d
$
be the stochastic process which satisfies 
for all $ n \in \{ 0, 1, \dots, N - 1 \} $
that
$ 
\mathcal{X}_0 = \xi
$
and
\begin{equation}
\label{eq:curlX}
  \mathcal{X}_{ n + 1 }
  =
  \mathcal{X}_n
  +
  \mu( \mathcal{X}_n )
  \left( t_{ n + 1 } - t_n \right)
  +
  \sigma( \mathcal{X}_n )
  \left( W_{ t_{ n + 1 } } - W_{ t_n } \right)
\end{equation}
and let $ \mathfrak{F} = ( \mathfrak{F}_n )_{  n \in \{ 0, 1, \dots, N \} } $ be the filtration generated by $ \mathcal{X} $.
Combining this with \eqref{eq:delta_Xb} suggests the approximation
\begin{equation}
\begin{split}
\label{eq:option_price_DISCRETE}
  \sup\!\left\{ 
    \E\big[ 
      g( t_{ \tau }, \mathcal{X}_{ \tau } )
    \big]
    \colon
    \substack{
      \tau \colon \Omega \to \{ 0, 1, \dots, N \} \text{ is an}
    \\
      \mathfrak{F} \text{-stopping time}
    }
  \right\}
\approx
  \sup\!\left\{ 
    \E\big[ 
      g( \tau, X_{ \tau } )
    \big]
    \colon
    \substack{
      \tau \colon \Omega \to [0,T] \text{ is an}
    \\
      \mathbb{F}\text{-stopping time}
    }
  \right\}
\end{split}
\end{equation}
for the price~\eqref{eq:option_price} of the American option in Subsection~\ref{ssec:problem}.
Below we employ, in particular, \eqref{eq:option_price_DISCRETE} to derive the proposed approximation algorithm.

For every $n \in \{0,1, \dots, N\}$
note that
the current pay-off $g(t_n, \mathcal{X}_n)$ does not carry any information that is not already contained 
in $\mathcal{X}_n$.
But typically,
optimal exercise strategies are learned more efficiently if 
it is added as an additional feature.
Therefore, we introduce the $(d+1)$-dimensional Markov process
$\mathcal{Z} \colon \{0,1, \dots, N\} \times \Omega \to \R^{d+1}$
which satisfies for all $n \in \{0,1, \dots, N\}$ that
$\mathcal{Z}_n = \bigl(\mathcal{X}_n^{(1)}, \dots, \mathcal{X}_n^{(d)}, g(t_n, \mathcal{X}_n)\bigr)$.
Observe that $ \mathcal{Z} $ and $\mathcal{X}$ generate the same filtration $\mathfrak{F}$.

\subsection{Factorisation lemma for stopping times}
\label{sec:factorization_lemma}

The derivation of the proposed approximation algorithm is in parts based on 
an elementary reformulation of time-discrete stopping times (cf.\ the left-hand side of \eqref{eq:option_price_DISCRETE} above)
in terms of measurable functions that appropriately characterise the behaviour
of the stopping time; cf.\ \eqref{eq:stopping_rep_1} and \eqref{eq:stopping_rep_2} in Lemma~\ref{lem:yg_squared} below.
The proof of Lemma~\ref{lem:yg_squared} employs the following well-known factorisation result, 
Lemma~\ref{lem:factorization}. Lemma~\ref{lem:factorization} follows, e.g., from Klenke~\cite[Corollary~1.97]{Klenke2014}.

\begin{lemma}[Factorisation lemma]
\label{lem:factorization}
Let $ ( S, \mathcal{S} ) $
be a measurable space, 
let $ \Omega $ be a set,
let $ B \in \mathcal{B}( \R \cup \{ - \infty, \infty \} ) $,
and let $ X \colon \Omega \to S $ and $ Y \colon \Omega \to B $ be functions.
Then it holds that
$ Y $ 
is 
$ \{ X^{ - 1 }( A ) \colon A \in \mathcal{S} \} $/$ \mathcal{B}( B ) $-measurable 
if and only if
there exists an 
$
  \mathcal{S}
$/$
  \mathcal{B}( B )
$-measurable function 
$ f \colon S \to B $
such that
\begin{equation}
  Y = f \circ X 
  .
\end{equation}
\end{lemma}

We are now ready to present the above-mentioned Lemma~\ref{lem:yg_squared}. This elementary lemma 
is a consequence of Lemma~\ref{lem:factorization} above.

\begin{lemma}[Factorisation lemma for stopping times]
\label{lem:yg_squared}
Let $ d, N \in \N $, 
let $ ( \Omega, \mathcal{F}, \P ) $ 
be a probability space,
let 
$ 
  \mathcal{Z} \colon 
  \{ 0, 1, \dots, N \} \times \Omega \to \R^{d+1}
$ 
be a stochastic process, 
and 
let $ \mathfrak{F} = ( \mathfrak{F}_n )_{ n \in \{ 0, 1, \dots, N \} } $
be the filtration generated by $ \mathcal{Z} $.
Then 
\begin{enumerate}[(i)]
\item
\label{item:yg_squared_ii}
for all
Borel measurable functions
$
  \mathbb{U}_n \colon ( \R^{d+1} )^{ n + 1 } \to \{ 0, 1 \}
$,
$ n \in \{ 0, 1, \dots, N \} $,
with 
$
  \forall \, z_0, z_1, \dots, z_N \in \R^{d+1} \colon
  \sum_{ n = 0 }^N
  \mathbb{U}_n( z_0, z_1, \dots, z_n )
  = 1
$
it holds that the function
\begin{equation}
\label{eq:stopping_rep_2}
  \Omega \ni \omega 
  \mapsto
    \sum_{ n = 0 }^N
    n \, 
    \mathbb{U}_n\big( 
      \mathcal{Z}_0( \omega ), 
      \mathcal{Z}_1( \omega ), 
      \dots, 
      \mathcal{Z}_n( \omega ) 
    \big)
  \in 
  \{ 0, 1, \dots, N \}
\end{equation}
is an $ \mathfrak{F} $-stopping time
and
\item
\label{item:yg_squared_i}
for every 
$ \mathfrak{F} $-stopping time 
$ \tau \colon \Omega \to \{ 0, 1, \dots, N \} $
there exist Borel measurable functions
$
  \mathbb{U}_n \colon ( \R^{d+1} )^{ n + 1 } \to \{ 0, 1 \}
$,
$ n \in \{ 0, 1, \dots, N \} $,
which satisfy
$
  \forall \, z_0, z_1, \dots, z_N \in \R^{d+1} \colon
  \sum_{ n = 0 }^N
  \mathbb{U}_n( z_0, z_1, \dots, z_n )
  = 1
$
and
\begin{equation}
\label{eq:stopping_rep_1}
    \tau 
    =
    \sum_{ n = 0 }^N
    n \, \mathbb{U}_n( \mathcal{Z}_0, \mathcal{Z}_1, \dots, \mathcal{Z}_n )
    .
\end{equation}
\end{enumerate}
\end{lemma}

\begin{proof}[Proof of Lemma~\ref{lem:yg_squared}]
Note that for all
Borel measurable functions
$
  \mathbb{U}_n \colon ( \R^{d+1} )^{ n + 1 } \to \{ 0, 1 \}
$,
$ n \in \{ 0, 1, \dots, N \} $,
with 
$
  \forall \, z_0, z_1, \dots, z_N \in \R^{d+1} \colon
  \sum_{ n = 0 }^N
  \mathbb{U}_n( z_0, z_1, \dots, z_n )
  = 1
$
and all 
$ k \in \{ 0, 1, \dots, N \} $
it holds that 
\begin{equation}
\begin{split}
&
  \left\{ 
    \omega \in \Omega \colon
    \sum_{ n = 0 }^N
    n \, 
    \mathbb{U}_n\big( 
      \mathcal{Z}_0( \omega ), 
      \mathcal{Z}_1( \omega ), 
      \dots, 
      \mathcal{Z}_n( \omega ) 
    \big)
    =
    k
  \right\}
\\ & =
  \left\{ 
    \omega \in \Omega
    \colon
    \mathbb{U}_k\big( 
      \mathcal{Z}_0( \omega ), 
      \mathcal{Z}_1( \omega ), 
      \dots, 
      \mathcal{Z}_k( \omega ) 
    \big)
    =
    1
  \right\}
\\ & =
  \left\{ 
    \omega \in \Omega
    \colon
    ( 
      \mathcal{Z}_0( \omega ), 
      \mathcal{Z}_1( \omega ), 
      \dots, 
      \mathcal{Z}_k( \omega ) 
    )
    \in
    \underbrace{
      ( \mathbb{U}_k )^{ - 1 }( \{ 1 \} )
    }_{
      \in \mathcal{B}( ( \R^{d+1} )^{ k + 1 } )
    }
  \right\}
  \in
  \mathfrak{F}_k
  .
\end{split}
\end{equation}
This establishes \eqref{item:yg_squared_ii}.
It thus remains to prove \eqref{item:yg_squared_i}. 
For this let 
$ \tau \colon \Omega \to \{ 0, 1, \dots, N \} $
be an $ \mathfrak{F} $-stopping time. 
Observe that for every function 
$ \varrho \colon \Omega \to \{ 0, 1, \dots, N \} $
and every $ \omega \in \Omega $
it holds that
\begin{equation}
\label{eq:rho_rep}
  \varrho( \omega )
  =
  \sum_{ n = 0 }^N
  n 
  \,
  \mathbbm{1}_{
    \{ \varrho = n \}
  }( \omega )
  .
\end{equation}
Next note that 
for every $ n \in \{ 0, 1, \dots, N \} $
it holds that the function
\begin{equation}
  \Omega 
  \ni 
  \omega 
  \mapsto
  \mathbbm{1}_{
    \{ \tau = n \}
  }( \omega )
  \in
  \{ 0, 1 \}
\end{equation}
is $ \mathfrak{F}_n $/$ \mathcal{B}( \{ 0, 1 \} ) $-measurable. 
This and the fact that 
\begin{equation}
  \forall \, n \in \{ 0, 1, \dots, N \} \colon
  \sigma_{ \Omega }\big( ( \mathcal{Z}_0, \mathcal{Z}_1, \dots, \mathcal{Z}_n ) \big)
  =
  \mathfrak{F}_n
\end{equation}
ensure that 
for every $ n \in \{ 0, 1, \dots, N \} $
it holds that the function
\begin{equation}
  \Omega 
  \ni 
  \omega 
  \mapsto
  \mathbbm{1}_{
    \{ \tau = n \}
  }( \omega )
  \in
  \{ 0, 1 \}
\end{equation}
is $ \sigma_{ \Omega }( ( \mathcal{Z}_0, \mathcal{Z}_1, \dots, \mathcal{Z}_n ) ) $/$ \mathcal{B}( \{ 0, 1 \} ) $-measurable. 
Lemma~\ref{lem:factorization} hence demonstrates that
there exist 
Borel measurable functions
$
  \mathbb{V}_n \colon ( \R^{d+1} )^{ n + 1 } \to \{ 0, 1 \}
$,
$ n \in \{ 0, 1, \dots, N \} $,
which satisfy
for all
$ 
  n \in \{ 0, 1, \dots, N \} 
$,
$ \omega \in \Omega $ that
\begin{equation}
\label{eq:defVn}
  \mathbbm{1}_{
    \{ \tau = n \}
  }( \omega )
  =
  \mathbb{V}_n\big( \mathcal{Z}_0( \omega ), \mathcal{Z}_1( \omega ), \dots, \mathcal{Z}_n( \omega ) \big)
  .
\end{equation}
Next let 
$
  \mathbb{U}_n \colon ( \R^{d+1})^{ n + 1 } \to \R
$,
$ n \in \{ 0, 1, \dots, N \} $,
be the functions which satisfy 
for all $ n \in \{ 0, 1, \dots, N \} $,
$ z_0, z_1, \dots, z_n \in \R^{d+1} $
that
\begin{equation}
\label{eq:defUn}
\begin{split}
&
  \mathbb{U}_n( z_0, z_1, \dots, z_n )
\\ & 
=
  \max\!\left\{
    \mathbb{V}_n( z_0, z_1, \dots, z_n ) 
    ,
    n + 1 - N
  \right\}
  \left[
    1 
    - 
    \sum_{ k = 0 }^{ n - 1 }
    \mathbb{U}_k( z_0, z_1, \dots, z_k ) 
  \right]
  .
\end{split}
\end{equation}
Observe that 
\eqref{eq:defUn}, in particular, ensures that 
for all $ z_0, z_1, \dots, z_N \in \R^{ d + 1 } $
it holds that
\begin{equation}
  \mathbb{U}_N( z_0, z_1, \dots, z_N )
=
  \left[
    1 
    - 
    \sum_{ k = 0 }^{ N - 1 }
    \mathbb{U}_k( z_0, z_1, \dots, z_k ) 
  \right]
  .
\end{equation}
Hence, we obtain that
for all $ z_0, z_1, \dots, z_N \in \R^{d+1} $
it holds that
\begin{equation}
\label{eq:sumUn}
  \sum_{ k = 0 }^N
  \mathbb{U}_k( z_0, z_1, \dots, z_k ) 
  = 1
  .
\end{equation}
In addition, note that \eqref{eq:defUn} 
assures that for all $ z_0 \in \R^{d+1} $ it holds that
\begin{equation}
\label{eq:U0_V0}
  \mathbb{U}_0( z_0 ) = \mathbb{V}_0( z_0 )
  .
\end{equation}
Induction, the fact that 
\begin{equation}
\label{eq:Vn_in_01}
  \forall \, 
  n \in \{ 0, 1, \dots, N \}, \,
  z_0, z_1, \dots, z_n \in \R^{d+1}
  \colon
  \mathbb{V}_n( z_0, z_1, \dots, z_n ) \in \{ 0, 1 \}
  ,
\end{equation}
and \eqref{eq:defUn}
hence demonstrate that
for all $ n \in \{ 0, 1, \dots, N \} $,
$ z_0, z_1, \dots, z_n \in \R^{ d + 1 } $
it holds that
\begin{equation}
  \left\{ 
    \mathbb{U}_0( z_0 ), \mathbb{U}_1( z_0, z_1 ), \dots, \mathbb{U}_n( z_0, z_1, \dots, z_n ) ,
    \sum_{ k = 0 }^n
    \mathbb{U}_k( z_0, z_1, \dots, z_k )
  \right\} 
  \subseteq 
  \{ 0, 1 \}
  .
\end{equation}
Moreover, note that
\eqref{eq:defUn},
induction,
and the fact that the functions 
$
  \mathbb{V}_n \colon ( \R^d )^{ n + 1 } \to \{ 0, 1 \}
$,
$ n \in \{ 0, 1, \dots, N \} $,
are Borel measurable 
ensure that 
for every $ n \in \{ 0, 1, \dots, N \} $
it holds that
the function 
\begin{equation} 
\label{eq:Un_measurable}
  ( \R^d )^{ n + 1 } \ni (z_0, z_1, \dots, z_n) \mapsto 
  \mathbb{U}_n( z_0, z_1, \dots, z_n ) \in \{ 0, 1 \}
\end{equation}
is also Borel measurable. 
In the next step, we observe that 
\eqref{eq:U0_V0}, \eqref{eq:defUn}, \eqref{eq:Vn_in_01}, 
and 
induction
assure that
for all 
$ 
  n \in \{ 0, 1, \dots, N \} 
$,
$ 
 z_0, z_1, \dots, z_n \in \R^{d+1}
$
with
$
  n + 1 - N
  \leq
  \sum_{ k = 0 }^n
    \mathbb{V}_k( z_0, z_1, \dots, z_k )
    \leq 1
$
it holds that
\begin{equation}
\label{eq:Un=Vn}
  \forall \, k \in \{ 0, 1, \dots, n \} 
  \colon
  \mathbb{U}_k( z_0, z_1, \dots, z_k )
  =
  \mathbb{V}_k( z_0, z_1, \dots, z_k )
  .
\end{equation}
In addition,
note that
\eqref{eq:defVn}
shows that
for all
$ \omega \in \Omega $
it holds that
\begin{equation}
\sum_{ k = 0 }^N
    \mathbb{V}_k\big( \mathcal{Z}_0( \omega ), \mathcal{Z}_1( \omega ), \dots, \mathcal{Z}_k( \omega ) \big)
=
\sum_{ k = 0 }^N
    \mathbbm{1}_{
        \{ \tau = k \}
    }( \omega )
= 1.
\end{equation}
This, \eqref{eq:Un=Vn}, and again \eqref{eq:defVn}
imply that
for all
$ 
  k \in \{ 0, 1, \dots, N \} 
$,
$ \omega \in \Omega $
it holds that
\begin{equation}
\mathbb{U}_k\big( \mathcal{Z}_0( \omega ), \mathcal{Z}_1( \omega ), \dots, \mathcal{Z}_k( \omega ) \big)
=
\mathbb{V}_k\big( \mathcal{Z}_0( \omega ), \mathcal{Z}_1( \omega ), \dots, \mathcal{Z}_k( \omega ) \big)
=
\mathbbm{1}_{
    \{ \tau = k \}
}( \omega )
.
\end{equation}
Equation~\eqref{eq:rho_rep}
hence proves that
for all $ \omega \in \Omega $
it holds that
\begin{equation}
\tau( \omega )
=
\sum_{ n = 0 }^N
    n 
    \,
    \mathbb{U}_n\big( \mathcal{Z}_0( \omega ), \mathcal{Z}_1( \omega ), \dots, \mathcal{Z}_n( \omega ) \big)
.
\end{equation}
Combining this with 
\eqref{eq:sumUn}
and 
\eqref{eq:Un_measurable}
establishes
\eqref{item:yg_squared_i}. 
The proof of Lemma~\ref{lem:yg_squared} is thus complete.
\end{proof}

\subsection{Neural network architectures for stopping times}
\label{sec:NN}

In the next step, we employ
multilayer neural network approximations
of the functions 
$ \mathbb{U}_n \colon ( \R^{d+1} )^{ n + 1 } \to \{ 0, 1 \} $,
$ n \in \{ 0, 1, \dots, N \} $,
in the factorisation lemma, Lemma~\ref{lem:yg_squared} above.
In the following,
we refer to these functions as `stopping time factors'.
Consider again the setting in Subsections~\ref{ssec:problem}--\ref{sec:temporal_discretization},
for every $ \mathfrak{F} $-stopping time $ \tau \colon \Omega \to \{ 0, 1, \dots, N \} $
let
$
  \mathbb{U}_{n, \tau} \colon ( \R^{d+1} )^{ n + 1 } \to \{ 0, 1 \}
$,
$ n \in \{ 0, 1, \dots, N \} $,
be Borel measurable functions
which satisfy
$
  \forall \, z_0, z_1, \dots, z_N \in \R^{d+1} \colon
  \sum_{ n = 0 }^N
  \mathbb{U}_{n, \tau}( z_0, z_1, \dots, z_n )
  = 1
$
and
\begin{equation}
\label{eq:defUn_rho}
\tau 
=
\sum_{ n = 0 }^N
n \, \mathbb{U}_{n, \tau}( \mathcal{Z}_0, \mathcal{Z}_1, \dots, \mathcal{Z}_n )
\end{equation}
(cf.~\eqref{item:yg_squared_i} of Lemma~\ref{lem:yg_squared}),
let $ \nu \in \N $
be a sufficiently large natural number,
and for every 
$ n \in \{ 0, 1, \dots, N \} $, 
$ \theta \in \R^{ \nu } $
let 
$ 
  u_{ n, \theta } \colon \R^{d+1} \to (0,1)
$
and
$ 
  U_{ n, \theta } \colon ( \R^{d+1} )^{n+1} \to (0,1)
$
be Borel measurable functions which satisfy for all 
$ z_0, z_1, \dots, z_n \in \R^{d+1} $ that
\begin{equation}
\label{eq:defUntheta}
  U_{ n, \theta }( z_0, z_1, \dots, z_n )
  = 
  \max\!\left\{ 
    u_{ n, \theta }( z_n ) 
    ,
    n + 1 - N
  \right\}
  \left[ 
    1 
    - 
    \sum_{ k = 0 }^{ n - 1 }
    U_{ k, \theta }( z_0, z_1, \dots, z_k )
  \right]
\end{equation}
(cf.\ \eqref{eq:defUn} above).
We think of $ \nu \in \N $
as the number of parameters in the
employed artificial neural networks
(for instance, in the case of \eqref{eq:activation}--\eqref{eq:A_matrix_0} below
we have $ \nu = N (l^2 + l(d+4) + 1 )$, where $l$ is the number of neurons in each of the two hidden layers),
for every $ n \in \{ 0, 1, \dots, N \} $,
$ \theta \in \R^{\nu} $ we think of $ u_{n,\theta} \colon \R^{d+1} \rightarrow (0,1) $
as an appropriate artificial neural network
(cf.\ the next paragraph for further details),
and for every appropriate $ \mathfrak{F} $-stopping time
$ \tau \colon \Omega \to \{ 0, 1, \dots, N \} $
and every $ n \in \{ 0, 1, \dots, N \} $
we think of the function
$ U_{ n, \theta } \colon ( \R^{d+1} )^{n+1} \to (0,1) $
for suitable $ \theta \in \R^{ \nu } $
as an appropriate approximation of the stopping time factor
$
  \mathbb{U}_{n, \tau} \colon ( \R^{d+1} )^{ n + 1 } \to \{ 0, 1 \}
$.
In addition,
it shall be noted
that
for all $ n \in \{ 0, 1, \dots, N \} $,
in comparison with
the function
$ \mathbb{V}_n \colon ( \R^{d+1} )^{ n + 1 } \to \{ 0, 1 \} $
in~\eqref{eq:defUn} above,
the functions
$ u_{ n, \theta } \colon \R^{d+1} \to (0,1) $,
$ \theta \in \R^{ \nu } $,
in~\eqref{eq:defUntheta}
are defined only on $ \R^{d+1} $ instead of $ ( \R^{d+1} )^{n+1} $
and, therefore, only depend on
$ z_n \in \R^{d+1} $ instead of the whole vector $ ( z_0, z_1, \ldots, z_n ) \in ( \R^{d+1} )^{n+1} $.
Even though this constitutes a significant simplification,
\cite[Theorem~1 and Remark~2 in~Subsection~2.1]{BeckerCheriditoJentzen2019}
suggest that, due to the fact that $ \mathcal{Z} \colon \{ 0, 1, \dots, N \} \times \Omega \to \R^{d+1} $
is a Markov process, the approximate stopping time factors
$ U_{ n, \theta } \colon ( \R^{d+1} )^{n+1} \to (0,1) $,
$ \theta \in \R^{ \nu } $,
$ n \in \{ 0, 1, \dots, N \} $,
still possess enough flexibility to represent an optimal stopping time for the discrete stopping problem corresponding to the left side of~\eqref{eq:option_price_DISCRETE}.
Furthermore, observe that
for all $ \theta \in \R^{ \nu } $, $z_0, z_1, \dots, z_N \in \R^{d+1} $ it holds that
\begin{equation}
\label{eq:sum_to_1}
\sum_{ n = 0 }^{ N }
    U_{ n, \theta }( z_0, z_1, \dots, z_n )
= 1.
\end{equation}
Because of this,
for every $ \theta \in \R^{ \nu } $
the stochastic process
\begin{equation}
\label{eq:randomised_stopping_time}
\{ 0, 1, \dots, N \} \times \Omega \ni ( n, \omega )
\mapsto
U_{ n, \theta }\bigl( \mathcal{Z}_0( \omega ), \mathcal{Z}_1( \omega ), \dots, \mathcal{Z}_n( \omega ) \bigr)
\in ( 0, 1 )
\end{equation}
can also be viewed as an
appropriate sense
`randomised stopping time'
(cf., e.g., \cite[Definition~1 in Subsection~3.1]{SolanTsirelsonVieille2012arXiv}
and, e.g.,
\cite[Section~1.1]{FergusonChap1online}). 

We suggest to choose the functions 
$ u_{ n, \theta } \colon \R^{d+1} \to (0,1) $,
$ \theta \in \R^{ \nu } $,
$ n \in \{ 0, 1, \dots, N - 1 \} $,
as multilayer feedforward neural networks
(cf.~\cite[Corollary~5 in~Subsection~2.2]{BeckerCheriditoJentzen2019}
and, e.g.,
\cite{Cybenko1989,HornikStinchcombeWhite1989,Barron1993}).
For example,
for every $ k \in \N $
let $ \mathcal{L}_k \colon \R^k \to \R^k $ be the function 
which satisfies for all $ x = ( x_1, \dots, x_k ) \in \R^k $ that
\begin{equation}
\label{eq:activation}
  \mathcal{L}_k( x ) 
  =
  \left(
    \frac{ \exp( x_1 ) }{
      \exp( x_1 ) + 1
    }
    ,
    \frac{ \exp( x_2 ) }{
      \exp( x_2 ) + 1
    }
    ,
    \dots
    ,
    \frac{ 
      \exp( x_k ) 
    }{
      \exp( x_k ) + 1
    }
  \right)
  ,
\end{equation}
for every 
$ \theta = ( \theta_1, \dots, \theta_{ \nu } ) \in \R^{ \nu } $, 
$ v \in \N_0 = \{ 0, 1, 2, \ldots \} $,
$ k, j \in \N $
with 
$
  v + k (j + 1 ) \leq \nu
$
let $ A^{ \theta, v }_{ k, j } \colon \R^j \to \R^k $ be the affine linear function which 
satisfies for all $ x = ( x_1, \dots, x_j ) \in \R^j $ that
\begin{equation}\label{eq:A_matrix_0}
A^{ \theta, v }_{ k, j }( x )
=
\begin{pmatrix}
  \theta_{ v + 1 }
&
  \theta_{ v + 2 }
&
  \dots
&
  \theta_{ v + j }
\\
  \theta_{ v + j + 1 }
&
  \theta_{ v + j + 2 }
&
  \dots
&
  \theta_{ v + 2 j }
\\
  \theta_{ v + 2 j + 1 }
&
  \theta_{ v + 2 j + 2 }
&
  \dots
&
  \theta_{ v + 3 j }
\\
  \vdots
&
  \vdots
&
  \vdots
&
  \vdots
\\
  \theta_{ v + ( k - 1 ) j + 1 }
&
  \theta_{ v + ( k - 1 ) j + 2 }
&
  \dots
&
  \theta_{ v + k j }
\end{pmatrix}
\begin{pmatrix}
  x_1
\\
  x_2
\\
  x_3
\\
  \vdots 
\\
  x_j
\end{pmatrix}
+
\begin{pmatrix}
  \theta_{ v + k j + 1 }
\\
  \theta_{ v + k j + 2 }
\\
  \theta_{ v + k l + 3 }
\\
  \vdots 
\\
  \theta_{ v + k j + k }
\end{pmatrix}
,
\end{equation}
let $l \in \N$, and assume for all 
$ n \in \{ 0, 1, \dots, N - 1 \} $,
$ \theta \in \R^{ \nu } $
that
$ \nu \geq N (l^2 + l(d+4) + 1 ) $
and
\begin{equation}
\label{eq:neural_network}
u_{ n, \theta }
=
\mathcal{L}_1
\circ
A^{ \theta, n(l^2 + l(d+4) + 1) + l(l+d+3)  }_{ 1, l }
\circ
\mathcal{L}_l
\circ
A^{ \theta, n(l^2 + l(d+4) + 1) + l(d+2) }_{l, l} 
\circ 
\mathcal{L}_l
\circ 
A^{ \theta, n (l^2 + l(d+4) + 1) }_{ l, d+1 } 
.
\end{equation}
The functions in \eqref{eq:neural_network} provide 
artificial neural networks with $ 4 $ layers 
($ 1 $ input layer with $ d+1 $ neurons, $ 2 $ hidden layers with $ l $ neurons each, and $ 1 $ output layer with $ 1 $ neuron)
and the multidimensional version of the standard logistic function
$ \R \ni x \mapsto \nicefrac{\exp(x)}{(\exp(x)+1)} \in ( 0, 1 ) $
(cf.~\eqref{eq:activation} above)
as activation functions.
In our numerical simulations in Section~\ref{sec:examples},
we use this type of activation function
only
just in front of the output layer
and we employ instead
the multidimensional version of the rectifier function
$ \R \ni x \mapsto \max\{ x, 0 \} \in [ 0, \infty ) $
as activation functions just in front of the hidden layers.
But in order to keep the illustration here
as short as possible,
we only employ the multidimensional version of the standard logistic function
as activation functions in \eqref{eq:activation}--\eqref{eq:neural_network} above.
Furthermore, note that
in contrast to the choice of the functions
$ u_{ n, \theta } \colon \R^{d+1} \to (0,1) $,
$ \theta \in \R^{ \nu } $,
$ n \in \{ 0, 1, \dots, N - 1 \} $,
the choice of the functions
$ u_{ N, \theta } \colon \R^{d+1} \to (0,1) $,
$ \theta \in \R^{ \nu } $,
has no influence on the approximate stopping time factors
$ U_{ n, \theta } \colon ( \R^{d+1} )^{n+1} \to (0,1) $,
$ \theta \in \R^{ \nu } $,
$ n \in \{ 0, 1, \dots, N \} $
(cf.~\eqref{eq:defUntheta} above).

\subsection{Formulation of the objective function}
\label{sec:objective_function}

Recall that we intend to compute the real number
\begin{equation}
\label{eq:option_price_discrete}
\sup\Bigl\{ 
\E\bigl[ 
  g( t_{ \tau }, \mathcal{X}_{ \tau } )
\bigr]
\colon
\substack{
  \tau \colon \Omega \to \{ 0, 1, \dots, N \} \text{ is an}
\\
  \mathfrak{F}\text{-stopping time}
}
\Bigr\}
\end{equation}
as an approximation of the American option price \eqref{eq:option_price}
(cf.~\eqref{eq:option_price_DISCRETE} in Subsection~\ref{sec:temporal_discretization}).
By employing neural network architectures for stopping times (cf.\ Subsection~\ref{sec:NN} above),
we next propose to replace the search over all $ \mathfrak{F} $-stopping times
for finding the supremum in~\eqref{eq:option_price_discrete}
by a search over the artificial neural network parameters $ \theta \in \R^\nu $ (cf.~\eqref{eq:objective_motivation} below).
For this,
observe that \eqref{eq:defUn_rho} implies for all
$ \mathfrak{F} $-stopping times
$ \tau \colon \Omega \to \{ 0, 1, \dots, N \} $
and all
$ 
  n \in \{ 0, 1, \dots, N \} 
$
that
\begin{equation}
\mathbbm{1}_{
\{ \tau = n \}
}
=
\mathbb{U}_{ n, \tau }( \mathcal{Z}_0, \mathcal{Z}_1, \dots, \mathcal{Z}_n )
.
\end{equation}
Therefore,
for all
$ \mathfrak{F} $-stopping times
$ \tau \colon \Omega \to \{ 0, 1, \dots, N \} $
it holds that
\begin{equation}
g( t_{ \tau }, \mathcal{X}_{ \tau } )
=
\sum_{ n = 0 }^N
\mathbbm{1}_{
\{ \tau = n \}
}
\,
g( t_n, \mathcal{X}_n )
=
\sum_{ n = 0 }^N
\mathbb{U}_{n, \tau}( \mathcal{Z}_0, \mathcal{Z}_1, \dots, \mathcal{Z}_n )
\,
g( t_n, \mathcal{X}_n )
.
\end{equation}
Combining this with \eqref{item:yg_squared_ii} of Lemma~\ref{lem:yg_squared}
and \eqref{eq:sum_to_1}
inspires the approximation
\begin{align*}
\label{eq:objective_motivation}
& \sup\Bigl\{ 
\E\bigl[ 
  g( t_{ \tau }, \mathcal{X}_{ \tau } )
\bigr]
\colon
\substack{
  \tau \colon \Omega \to \{ 0, 1, \dots, N \} \text{ is an}
\\
  \mathfrak{F}\text{-stopping time}
}
\Bigr\}
\\
& =
\sup\Biggl\{ 
\E\Biggl[ 
  \sum_{ n = 0 }^N
  \mathbb{U}_{n, \tau}( \mathcal{Z}_0, \mathcal{Z}_1, \dots, \mathcal{Z}_n )
  \,
  g( t_n, \mathcal{X}_n )
\Biggr]
\colon
\substack{
  \tau \colon \Omega \to \{ 0, 1, \dots, N \} \text{ is an}
\\
  \mathfrak{F}\text{-stopping time}
}
\Biggr\}
\\
& =
\sup\Biggl\{ 
\E\Biggl[ 
  \sum_{ n = 0 }^N
  \mathbb{V}_n( \mathcal{Z}_0, \mathcal{Z}_1, \dots, \mathcal{Z}_n )
  \,
  g( t_n, \mathcal{X}_n )
\Biggr]
\colon
\substack{
  \mathbb{V}_n \colon ( \R^{d+1} )^{ n + 1 } \to \{ 0, 1 \},\, n \in \{ 0, 1, \dots, N \},
\\
  \text{are Borel measurable functions with}\\
    \forall \, z_0, z_1, \dots, z_N \in \R^{d+1} \colon
    \sum_{ n = 0 }^N
    \mathbb{V}_n( z_0, z_1, \dots, z_n )
    = 1
}
\Biggr\}
\\
& = \yesnumber \allowdisplaybreaks
\sup\Biggl\{ 
\E\Biggl[ 
  \sum_{ n = 0 }^N
  \mathfrak{V}_n( \mathcal{Z}_0, \mathcal{Z}_1, \dots, \mathcal{Z}_n )
  \,
  g( t_n, \mathcal{X}_n )
\Biggr]
\colon
\substack{
  \mathfrak{V}_n \colon ( \R^{d+1} )^{ n + 1 } \to [ 0, 1 ],\, n \in \{ 0, 1, \dots, N \},
\\
  \text{are Borel measurable functions with}\\
    \forall \, z_0, z_1, \dots, z_N \in \R^{d+1} \colon
    \sum_{ n = 0 }^N
    \mathfrak{V}_n( z_0, z_1, \dots, z_n )
    = 1
}
\Biggr\}
\\
& =
\sup\Biggl\{ 
\E\Biggl[ 
  \sum_{ n = 0 }^N
  \mathfrak{U}_n( \mathcal{Z}_0, \mathcal{Z}_1, \dots, \mathcal{Z}_n )
  \,
  g( t_n, \mathcal{X}_n )
\Biggr]
\colon
\substack{
  \mathfrak{U}_n \colon ( \R^{d+1} )^{ n + 1 } \to ( 0, 1 ),\, n \in \{ 0, 1, \dots, N \},
\\
  \text{are Borel measurable functions with}\\
    \forall \, z_0, z_1, \dots, z_N \in \R^{d+1} \colon
    \sum_{ n = 0 }^N
    \mathfrak{U}_n( z_0, z_1, \dots, z_n )
    = 1
}
\Biggr\}
\\
& \qquad\qquad\qquad\qquad\qquad\quad\
\approx
\sup\Biggl\{ 
\E\Biggl[ 
  \sum_{ n = 0 }^N
  U_{n, \theta}( \mathcal{Z}_0, \mathcal{Z}_1, \dots, \mathcal{Z}_n )
  \,
  g( t_n, \mathcal{X}_n )
\Biggr]
\colon
\theta \in \R^\nu
\Biggr\}
.
\end{align*}
In view of this, our numerical solution for approximatively computing~\eqref{eq:option_price_discrete}
consists of trying to find an approximate maximiser of the objective function
\begin{equation}
\label{eq:objective_function}
\R^\nu \ni \theta
\mapsto
\E\Biggl[ 
  \sum_{ n = 0 }^N
  U_{n, \theta}( \mathcal{Z}_0, \mathcal{Z}_1, \dots, \mathcal{Z}_n )
  \,
  g( t_n, \mathcal{X}_n )
\Biggr]
\in \R.
\end{equation}

\subsection{Stochastic gradient ascent optimisation algorithms}
\label{sec:SGD}

Local/global maxima of the objective function~\eqref{eq:objective_function}
can be approximately reached by maximising
the expectation of the random objective function
\begin{equation}
\label{eq:random_objective_function}
\R^\nu \times \Omega \ni ( \theta, \omega )
\mapsto
  \sum_{ n = 0 }^N
  U_{n, \theta}( \mathcal{Z}_0( \omega ), \mathcal{Z}_1( \omega ), \dots, \mathcal{Z}_n( \omega ) )
  \,
  g( t_n, \mathcal{X}_n( \omega ) )
\in \R
\end{equation}
by means of a stochastic gradient ascent-type optimisation algorithm.
This yields a sequence of random parameter vectors along which we expect the objective function~\eqref{eq:objective_function} to increase.
More formally,
applying under suitable hypotheses 
stochastic gradient ascent-type optimisation algorithms
to \eqref{eq:objective_function} 
results 
in random approximations 
\begin{equation}
\Theta_m
=
( \Theta_m^{ (1) }, \dots, \Theta_m^{ (\nu) } ) \colon \Omega \to \R^{ \nu } 
\end{equation}
for $ m \in \{ 0, 1, 2, \dots \} $
of the local/global maximum points of the objective function \eqref{eq:objective_function},
where $ m \in \{ 0, 1, 2, \dots \} $
is the number of steps of the employed
stochastic gradient ascent-type optimisation algorithm.

\subsection{Price and optimal exercise time for American-style options}
\label{sec:optimal_exercise_time}

The approximation algorithm sketched in Subsection~\ref{sec:SGD} above
allows us to approximatively compute both 
the \emph{price}
and an \emph{optimal exercise strategy} for the American option
(cf.\ Subsection~\ref{ssec:problem}).
Let $ M \in \N $ and consider a realisation $ \widehat{\Theta}_M \in \R^{ \nu } $
of the random variable $\Theta_M \colon \Omega \to \R^{ \nu } $.
Then for sufficiently large $ N, \nu, M \in \N $
a candidate for a suitable approximation of the American option price is the real number
\begin{equation}
\label{eq:candidate_price}
\E \Biggl[
  \sum_{ n = 0 }^N
  U_{n, \widehat{\Theta}_M}( \mathcal{Z}_0, \mathcal{Z}_1, \dots, \mathcal{Z}_n )
  \,
  g( t_n, \mathcal{X}_n )
\Biggr]
\end{equation}
and
a candidate for a suitable approximation of an optimal exercise strategy for the American option is the function
\begin{equation}
\label{eq:candidate_strategy}
\Omega \ni \omega
\mapsto
  \sum_{ n = 0 }^N
  n \,
  U_{n, \widehat{\Theta}_M}( \mathcal{Z}_0( \omega ), \mathcal{Z}_1( \omega ), \dots, \mathcal{Z}_n( \omega ) )
\in [ 0, N ].
\end{equation}
Note, however, that in general
the function \eqref{eq:candidate_strategy}
does not take values in $ \{ 0, 1, \ldots, N \} $
and hence is not a proper stopping time.
Similarly, note that in general it is not clear
whether there exists an exercise strategy such that
the number \eqref{eq:candidate_price}
is equal to the expected discounted pay-off under this exercise strategy.
For these reasons, we suggest other candidates for
suitable approximations of the price
and an optimal exercise strategy for the American option.
More specifically,
for every $ \theta \in \R^{ \nu } $
let 
$ \tau_{ \theta } \colon \Omega \to \{ 0, 1, \dots, N \} $
be the 
$ \mathfrak{F} $-stopping time given by
\begin{equation}
\label{eq:stopping_time_construction}
\tau_\theta
= 
\min\Biggl\{ 
n \in \{ 0, 1, \dots, N \}
  \colon
    \sum_{ k = 0 }^n
    U_{ k, \theta }(
      \mathcal{Z}_0 ,
      \dots ,
      \mathcal{Z}_k
    )
    \geq 
    1 
    -
    U_{ n, \theta }(
      \mathcal{Z}_0 ,
      \dots ,
      \mathcal{Z}_n
    )
\Biggr\}
\end{equation}
(cf.~\eqref{eq:sum_to_1} above).
Then
for sufficiently large
$ N, \nu, M \in \N $
we use a suitable Monte Carlo approximation of the real number
\begin{equation}
\label{eq:price_final_approximation}
\E \Bigl[
g\bigl( 
  t_{ \tau_{ \widehat{\Theta}_M } }
  ,
  \mathcal{X}_{ \tau_{ \widehat{\Theta}_M } }
\bigr)
\Bigr]
\end{equation}
as a suitable implementable approximation
of the price of the American option
(cf.\ \eqref{eq:option_price} in Subsection~\ref{ssec:problem} above and \eqref{eq:price_Monte_Carlo_special} in Subsection~\ref{sec:specific} below)
and
we use the random variable
\begin{equation}
\label{eq:approx_optimal_exercise_strategy}
  \tau_{ \widehat{\Theta}_M }
  \colon 
  \Omega \to 
  \{ 0, 1, \dots, N \}
\end{equation}
as a suitable implementable 
approximation of an optimal exercise strategy 
for the American option.
Note that~\eqref{eq:sum_to_1} ensures that
\begin{equation}
\begin{split}
\tau_{ \widehat{\Theta}_M }
& = 
\min\Biggl\{ 
n \in \{ 0, 1, \dots, N \}
  \colon
  U_{ n, \widehat{\Theta}_M }(
        \mathcal{Z}_0 ,
        \dots ,
        \mathcal{Z}_n
      )
    \geq
    1
    -
    \sum_{ k = 0 }^n
    U_{ k, \widehat{\Theta}_M }(
      \mathcal{Z}_0 ,
      \dots ,
      \mathcal{Z}_k
    )
\Biggr\}
\\ & =
\min\Biggl\{ 
n \in \{ 0, 1, \dots, N \}
  \colon
  U_{ n, \widehat{\Theta}_M }(
        \mathcal{Z}_0 ,
        \dots ,
        \mathcal{Z}_n
      )
    \geq 
    \sum_{ k = n+1 }^N
    U_{ k, \widehat{\Theta}_M }(
      \mathcal{Z}_0 ,
      \dots ,
      \mathcal{Z}_k
    )
\Biggr\}
.
\end{split}
\end{equation}
This shows that
the exercise strategy $ \tau_{ \widehat{\Theta}_M } \colon \Omega \to \{ 0, 1, \dots, N \} $
exercises at the first time index $ n \in \{ 0, 1, \dots, N \} $
for which the approximate stopping time factor associated with the mesh point $ t_n $
is at least as large as the combined approximate stopping time factors associated with all later mesh points.
Finally, observe that

\begin{equation}
\label{eq:low_biased}
\E \Bigl[
g\bigl( 
  t_{ \tau_{ \widehat{\Theta}_M } }
  ,
  \mathcal{X}_{ \tau_{ \widehat{\Theta}_M } }
\bigr)
\Bigr]
\leq
\sup\Bigl\{ 
\E\bigl[ 
  g( t_{ \tau }, \mathcal{X}_{ \tau } )
\bigr]
\colon
\substack{
  \tau \colon \Omega \to \{ 0, 1, \dots, N \} \text{ is an}
\\
  \mathfrak{F}\text{-stopping time}
}
\Bigr\}
.
\end{equation}
This implies that Monte Carlo approximations of the number~\eqref{eq:price_final_approximation}
typically are low-biased approximations of the American option price~\eqref{eq:option_price}.

\section{Details of the proposed algorithm}
\label{sec:algorithm_details}

\subsection{Formulation of the proposed algorithm in a special case}
\label{sec:specific}

In this subsection, we describe the proposed algorithm 
in the specific situation 
where the objective is to solve the American option pricing problem 
described in Subsection~\ref{ssec:problem},
where \emph{batch normalisation} (cf.\ Ioffe \& Szegedy~\cite{IoffeSzegedy2015}) is not employed 
in the proposed algorithm,
and where the plain vanilla stochastic gradient ascent approximation method 
with a constant learning rate $ \gamma \in (0,\infty) $
and without mini-batches is the employed stochastic approximation algorithm.
The general framework, which includes the setting in this subsection 
as a special case, can be found in Subsection~\ref{sec:general_case} below.

\begin{algo}[Specific case]
\label{algo:special}
Let $ T, \gamma \in (0,\infty) $, $ d, N, l \in \N $,
$ \nu = N(l^2 + l(d+4) + 1) $,
let 
$ \mu \colon \R^d \to \R^d $,
$ \sigma \colon \R^d \to \R^{ d \times d } $,
and
$ g \colon [0,T] \times \R^d \to \R $
be Borel measurable functions,
let $ ( \Omega, \mathcal{F}, \P ) $ be a probability space,
let
$ \xi^m \colon \Omega \to \R^d $,
$ m \in \N $,
be independent random variables,
let 
$ W^m \colon [0,T] \times \Omega \to \R^d $, 
$ m \in \N $,
be independent $ \P $-standard Brownian motions with continuous sample paths,
assume that
$ ( \xi^m )_{ m \in \N } $ and $ ( W^m )_{ m \in \N } $
are independent,
let $ t_0, t_1, \dots, t_N \in [0,T] $
be real numbers with
$
  0 = t_0 < t_1 < \ldots < t_N = T
$,
let 
$
  \mathcal{X}^m = (\mathcal{X}^{m,(1)}, \ldots, \mathcal{X}^{m,(d)} ) \colon
  \{0, 1, \dots, N \} \times \Omega \to \R^d
$,
$ m \in \N $,
and 
$
  \mathcal{Z}^m \colon
  \{0, 1, \dots, N \} \times \Omega \to \R^{d+1},
$ 
$ m \in \N $,
be the stochastic processes which satisfy for all 
$ m \in \N $, $ n \in \{ 0, 1, \dots, N - 1 \} $,
$ \mathfrak{n} \in \{ 0, 1, \dots, N \} $
that
$ \mathcal{X}^m_0 = \xi^m $,
\begin{equation}
\mathcal{X}^m_{ n + 1 } 
=
\mathcal{X}^m_{n }
+
\mu\big( \mathcal{X}^m_{ n } \big)
(
t_{ n + 1 } 
-
t_n 
)
+
\sigma\big( \mathcal{X}^m_{ n } \big)
\big(
W_{ t_{ n + 1 } }^m
-
W_{ t_n }^m
\big),
\end{equation}
and
\begin{equation}
\mathcal{Z}^m_{\mathfrak{n}} = \big(\mathcal{X}^{m,(1)}_{\mathfrak{n}}, \dots, \mathcal{X}^{m,(d)}_{\mathfrak{n}}, g(t_{\mathfrak{n}}, \mathcal{X}^m_{\mathfrak{n}})\big),
\end{equation}
for every $k \in \N$ let $ \mathcal{L}_k \colon \R^k \to \R^k $ be the function 
which satisfies for all $ x = ( x_1, \dots, x_k ) \in \R^k $ that
\begin{equation}
\label{eq:logistic}
\mathcal{L}_k( x ) 
=
\left(
\frac{ \exp( x_1 ) }{
  \exp( x_1 ) + 1
}
,
\frac{ \exp( x_2 ) }{
  \exp( x_2 ) + 1
}
,
\dots
,
\frac{ 
  \exp( x_k ) 
}{
  \exp( x_k ) + 1
}
\right)
,
\end{equation}
for every $ \theta = ( \theta_1, \dots, \theta_{ \nu } ) \in \R^{ \nu } $, 
$ v \in \N_0 $,
$k,j \in \N $ with 
$ v + k ( j + 1 ) \leq \nu $
let $ A^{ \theta, v }_{ k, j } \colon \R^j \to \R^k $ be the function which 
satisfies for all $ x = ( x_1, \dots, x_j ) \in \R^j$ that
\begin{equation}
A^{ \theta, v }_{ k, j }( x )
=
\biggl(
\theta_{ v + k j + 1 }
+
\biggl[
\smallsum_{ i = 1 }^j 
x_i 
\, 
\theta_{ v + i }
\biggr]
,
\dots 
,
\theta_{ v + k j + k }
+
\biggl[
\smallsum_{ i = 1 }^j
x_i 
\, 
\theta_{ v + ( k - 1 ) j + i }
\biggr]
\biggr)
,
\end{equation}
for every 
$ \theta \in \R^{ \nu } $
let 
$ u_{ n, \theta } \colon \R^{d+1} \to (0,1) $, $ n \in \{ 0, 1, \dots, N \} $,
be functions
which satisfy
for all $ n \in \{ 0, 1, \dots, N - 1 \} $
that
\begin{equation}
u_{ n, \theta }
=
\mathcal{L}_1
\circ
A^{ \theta, n(l^2 + l(d+4) + 1) + l(l+d+3) }_{ 1, l }
\circ
\mathcal{L}_l
\circ
A^{ \theta, n(l^2 + l(d+4) + 1) + l(d+2) }_{ l, l } 
\circ 
\mathcal{L}_l
\circ 
A^{ \theta, n(l^2 + l(d+4) + 1) }_{ l, d+1 }
,
\end{equation}
for every
$ n \in \{ 0, 1, \dots, N \} $,
$ \theta \in \R^{ \nu } $
let 
$ 
  U_{ n, \theta } \colon 
  ( \R^{d+1} )^{n+1}
  \to (0,1) 
$
be the function
which satisfies 
for all 
$ z_0, z_1, \dots, z_n \in \R^{ d + 1 } $
that
\begin{equation}
  U_{ n, \theta }( z_0, z_1, \dots, z_n )
  = 
  \max\!\left\{ 
    u_{ n, \theta }( z_n ) 
    ,
    n + 1 - N
  \right\}
  \left[ 
    1 
    - 
    \sum_{ k = 0 }^{ n - 1 }
    U_{ k, \theta }( z_0, z_1, \dots, z_k )
  \right]
  ,
\end{equation}
for every 
$ m \in \N $
let
$
  \phi^m
  \colon
  \R^{ \nu } 
  \times 
  \Omega
  \to 
  \R
$
be the function which satisfies for all 
$ \theta \in \R^{ \nu } $,
$ \omega \in \Omega $
that
\begin{equation}
\begin{split}
  \phi^m( \theta, \omega ) 
  = 
  \sum_{ n = 0 }^N
  \Big[
    U_{ n, \theta }\big( 
      \mathcal{Z}^m_{ 0 }( \omega )
      ,
      \mathcal{Z}^m_{ 1 }( \omega )
      ,
      \dots 
      ,
      \mathcal{Z}^m_{ n }( \omega )
    \big)
    \,
    g\big(
      t_n ,
      \mathcal{Z}_{ n }^m( \omega )
    \big) 
  \Big]
  ,
\end{split}
\end{equation}
for every 
$ m \in \N $
let
$
  \Phi^m
  \colon
  \R^{ \nu } 
  \times 
  \Omega
  \to 
  \R^{ \nu }
$
be the function which satisfies 
for all
$ 
  \theta \in 
    \R^{ \nu } 
$,
$ \omega \in \Omega $
that
\begin{equation}
  \Phi^m( \theta, \omega ) 
  =
  ( \nabla_{ \theta } \phi^m )( \theta, \omega ) 
  ,
\end{equation}
let 
$ 
  \Theta 
  \colon \N_0 \times \Omega \to 
  \R^{ \nu }
$ 
be a stochastic process
which satisfies for all $ m \in \N $ that 
\begin{equation}
\label{eq:plain_vanilla}
  \Theta_m 
  =
  \Theta_{ m - 1 } 
  +
  \gamma \cdot
  \Phi^m(
    \Theta_{ m - 1 } 
  )
  ,
\end{equation}
and for every $ j \in \N $, $ \theta \in \R^{ \nu } $ 
let 
$ \tau_{ j, \theta } \colon \Omega \to \{0, 1, \dots, N \} $
be the random variable given by
\begin{equation}
\begin{split}
\tau_{ j, \theta }
= 
\min\biggl\{ &
n \in \{0,1,\dots, N\} 
\colon
    \smallsum_{ k = 0 }^n
    U_{ k, \theta }\bigl(
      \mathcal{Z}^j_{ 0 } ,
      \dots ,
      \mathcal{Z}^j_{ k }
    \bigr)
    \geq 
    1 
    -
    U_{ n, \theta }\bigl(
      \mathcal{Z}^j_{ 0 } ,
      \dots ,
      \mathcal{Z}^j_{ n }
    \bigr)
\biggr\}
.
\end{split}
\end{equation}
\end{algo}

Consider the setting of Framework~\ref{algo:special},
assume that $ \mu $ and $ \sigma $ are globally Lipschitz continuous,
and assume that $ g $ is continuous and at most polynomially growing.
In the case of sufficiently large
$ N, M, J \in \N $ 
and sufficiently small $ \gamma \in (0,\infty) $,
we then think of the random number
\begin{equation}
\label{eq:price_Monte_Carlo_special}
\frac{ 1 }{ J }
\sum_{ j = 1 }^J
g\bigl( 
  t_{ \tau_{ M + j, \Theta_M } }
  ,
  \mathcal{X}^{ M + j }_{ 
    \tau_{ M + j, \Theta_M }
  }
\bigr)
\end{equation}
as an approximation of the price of the American option 
with the discounted pay-off function $ g $
and for every $ j \in \N $ we think of the random variable
\begin{equation}
\tau_{ M + j, \Theta_M }
\colon \Omega \to \{ 0, 1, \dots, N \}
\end{equation}
as an approximation of an \emph{optimal exercise strategy} 
associated with the underlying time-discrete path 
$ ( \mathcal{X}^{ M + j }_{ n } )_{ n \in \{ 0, 1, \dots, N \} } $
(cf.\ Subsection~\ref{ssec:problem} above and Section~\ref{sec:examples} below).

\subsection{Formulation of the proposed algorithm in the general case}
\label{sec:general_case}

In this subsection, we extend the framework in Subsection~\ref{sec:specific} above
and describe the proposed algorithm in the general case.

\begin{algo}
\label{algo:general}
Let $ T \in (0,\infty) $, $ d, N, M, \nu, \varsigma, \varrho \in \N $,
let
$ g \colon [0,T] \times \R^d \to \R $
be a Borel measurable function,
let $ ( \Omega, \mathcal{F}, \P ) $ be a probability space,
let $ t_0, t_1, \dots, t_N \in [0,T] $
be real numbers with
$ 0 = t_0 < t_1 < \ldots < t_N = T $,
let
$ \mathcal{X}^{ m, j } = ( \mathcal{X}^{ m, j, (1) }, \ldots, \mathcal{X}^{ m, j, (d) } )
\colon \allowbreak \{0, 1, \dots, N \} \times \Omega \to \R^d $, $ m \in \N_0 $,
$ j \in \N $, be i.i.d.\ stochastic processes,
for every $ m \in \N_0 $, $ j \in \N $
let 
$ \mathcal{Z}^{ m, j } \colon \{0, 1, \dots, N \} \times \Omega \to \R^{d+1} $ be the stochastic process
which satisfies for all $n \in \{0,1,\dots, N\}$ that
$\mathcal{Z}^{ m, j }_n = \big(\mathcal{X}^{ m, j, (1) }_n , \dots, \mathcal{X}^{ m, j, (d) }_n , g(t_n, \mathcal{X}^{ m, j }) \big)$, for every
$ n \in \{ 0, 1, \dots, N \} $,
$ \theta \in \R^{ \nu } $,
$ \mathbf{s} \in \R^{ \varsigma } $
let 
$ u_{ n }^{ \theta, \mathbf{s} } \colon \R^{d+1} \to (0,1) $
be a
function, 
for every
$ n \in \{ 0, 1, \dots, N \} $,
$ \theta \in \R^{ \nu } $,
$ \mathbf{s} \in \R^{ \varsigma } $
let 
$ U_{ n }^{ \theta, \mathbf{s} } \colon 
  ( \R^{d+1} )^{ n + 1 }
  \to (0,1) $
be the function
which satisfies 
for all 
$ z_0, z_1, \dots, z_n \in \R^{d+1} $
that
\begin{equation}
\label{eq:defUnthetaS}
  U_{ n }^{ \theta, \mathbf{s} }( z_0, z_1, \dots, z_n )
  = 
  \max\!\left\{ 
    u_{ n }^{ \theta, \mathbf{s} }( z_n ) 
    ,
    n + 1 - N
  \right\}
  \left[ 
    1 
    - 
    \sum_{ k = 0 }^{ n - 1 }
    U_{ k }^{ \theta, \mathbf{s} }( z_0, z_1, \dots, z_k )
  \right]
  ,
\end{equation}
let $ ( J_m )_{ m \in \N_0 } \subseteq \N $ be a sequence,
for every
$ m \in \N $,
$ \mathbf{s} \in \R^{ \varsigma } $
let
$ \phi^{ m, \mathbf{s} }\colon \R^{ \nu } \times \Omega \to \R $
be the function which satisfies for all 
$ \theta \in \R^{ \nu } $,
$ \omega \in \Omega $
that
\begin{equation}
\label{eq:objective_S}
\begin{split}
\phi^{ m, \mathbf{s} }( \theta, \omega ) 
=
\frac{1}{J_m}
\sum_{j=1}^{J_m}
\sum_{ n = 0 }^N
\Big[
U_{ n }^{ \theta, \mathbf{s} }\big(
  \mathcal{Z}^{m,j}_{ 0 } ( \omega )
  ,
  \mathcal{Z}^{m,j}_{ 1 }( \omega )
  ,
  \dots 
  ,
  \mathcal{Z}^{m,j}_{ n } ( \omega )
\big)
\,
g \big(t_n, \mathcal{X}^{m,j}_n( \omega ) \big)
\Big]
,
\end{split}
\end{equation}
for every 
$ m \in \N $,
$ \mathbf{s} \in \R^{ \varsigma } $
let
$ \Phi^{ m, \mathbf{s} } \colon \R^{ \nu } \times \Omega \to \R^{ \nu } $
be a function which satisfies 
for all
$ \omega \in \Omega $,
$ \theta \in
\{ \eta \in \R^{ \nu } \colon
\phi^{ m, \mathbf{s} }( \cdot, \omega )\colon \R^{ \nu } \to \R
\text{ is differentiable at } \eta \} $
that
\begin{equation}
\Phi^{ m, \mathbf{s} }( \theta, \omega ) 
=
( \nabla_{ \theta } \phi^{ m, \mathbf{s} } )( \theta, \omega ) 
,
\end{equation}
let $ \mathcal{S} \colon \R^\varsigma \times \R^\nu \times ( \R^{d+1} )^{ \{ 0, 1, \ldots, N-1 \} \times \N } \to \R^\varsigma $
be a function,
for every $ m \in \N $
let $ \Psi_m \colon \R^\varrho \times \R^\nu \to \R^\varrho $
and
$ \psi_m \colon \R^\varrho \to \R^\nu $
be functions,
let
$ \mathbb{S} \colon \N_0 \times \Omega \to \R^\varsigma $,
$ \Xi \colon \N_0 \times \Omega \to \R^\varrho $,
and
$ \Theta \colon \N_0 \times \Omega \to \R^{ \nu } $
be stochastic processes
which satisfy for all $ m \in \N $ that
\begin{equation}
\label{eq:batch_normalization}
\mathbb{S}_m
=
\mathcal{S}\bigl(
\mathbb{S}_{ m - 1 },
\Theta_{ m - 1 },
( \mathcal{Z}^{ m, j }_{ n } )_{ ( n, j ) \in \{ 0, 1, \dots, N - 1 \} \times \N }
\bigr),
\end{equation}
\begin{equation}
\label{eq:gradient_ascent}
\Xi_m
= \Psi_m( \Xi_{ m - 1 }, \Phi^{ m, \mathbb{S}_m }( \Theta_{ m - 1 } ) ),
\qquad
\text{and}
\qquad
\Theta_m 
=
\Theta_{ m - 1 }
+
\psi_m( \Xi_m ),
\end{equation}
for every $ j \in \N $,
$ \theta \in \R^{ \nu } $,
$ \mathbf{s} \in \R^{ \varsigma } $
let 
$ \tau^{ j, \theta, \mathbf{s} } \colon \Omega \to \{ 0, 1, \dots, N \} $
be the random variable given by
\begin{align*}
\yesnumber
\tau^{ j, \theta, \mathbf{s} }
=
\smash{\min\biggl\{ }&
n \in \{0,1, \dots, N\}
\colon
    \smallsum_{ k = 0 }^n
    U_{ k }^{ \theta, \mathbf{s} }\bigl(
      \mathcal{Z}^{0,j}_{ 0 } ,
      \dots ,
      \mathcal{Z}^{0,j}_{ k }
    \bigr)
    \geq 
    1 
    -
    U_{ n }^{ \theta, \mathbf{s} } \bigl(
      \mathcal{Z}^{0,j}_{ 0 } ,
      \dots ,
      \mathcal{Z}^{0,j}_{ n }
    \bigr)
\biggr\}
,
\end{align*}
and let $ \mathcal{P} \colon \Omega \to \R $
be the random variable which satisfies for all $ \omega \in \Omega $ that
\begin{equation}
\label{apprprice}
\mathcal{P}( \omega )
=
\frac{ 1 }{ J_0 }
\sum_{ j = 1 }^{ J_0 }
g \big(t_{ \tau^{ j, \Theta_M( \omega ), \mathbb{S}_M( \omega )}}, 
   \mathcal{X}^{ 0, j }_{ \tau^{ j, \Theta_M( \omega ), \mathbb{S}_M( \omega ) }( \omega ) }( \omega ) \big)
.
\end{equation}
\end{algo}

Consider the setting of Framework~\ref{algo:general}. Under suitable further assumptions, in the case of sufficiently large
$ N, M, \nu, J_0 \in \N $,
we think of the random number
\begin{equation}
\mathcal{P}
=
\frac{ 1 }{ J_0 }
\sum_{ j = 1 }^{J_0}
g\bigl( t_{
  \tau^{ j, \Theta_M, \mathbb{S}_M }}
  ,
  \mathcal{X}^{ 0, j }_{ 
    \tau^{ j, \Theta_M, \mathbb{S}_M }
  }
\bigr)
\end{equation}
as an approximation of the price of the American option 
with the discounted pay-off function $ g $
and for every $ j \in \N $ we think of the random variable
\begin{equation}
\tau^{ j, \Theta_M, \mathbb{S}_M }
\colon \Omega \to \{ 0, 1, \dots, N \}
\end{equation}
as an approximation of an \emph{optimal exercise strategy} 
associated with the underlying time-discrete path 
$ ( \mathcal{X}^{ 0, j }_{ n } )_{ n \in \{ 0, 1, \dots, N \} } $
(cf.\ Subsection~\ref{ssec:problem} above and Section~\ref{sec:examples} below).

\subsection{Comments on the proposed algorithm}
\label{sec:comments}

Note that the lack in Framework \ref{algo:general} of any assumptions on the dynamics of the stochastic process
$ ( \mathcal{X}_n^{ 0, 1 } )_{ n \in \{ 0, 1, \dots, N \} } $ allows us to approximatively compute
the optimal pay-off as well as an optimal exercise strategy
for very general optimal stopping problems
where, in particular, the stochastic process under consideration
is not necessarily related to the solution of a stochastic differential equation.
We only require that $ ( \mathcal{X}_n^{ 0, 1 } )_{ n \in \{ 0, 1, \dots, N \} } $
can be simulated efficiently
and formally we still rely on the Markov assumption (cf.~Subsection~\ref{sec:NN} above).
In addition, observe that the choice of the functions
$ u_{ N }^{ \theta, \mathbf{s} } \colon \R^{d+1} \to (0,1) $,
$ \mathbf{s} \in \R^{ \varsigma } $,
$ \theta \in \R^{ \nu } $,
has no influence on the proposed algorithm
(cf.~\eqref{eq:defUnthetaS}).
Furthermore,
the dynamics in~\eqref{eq:gradient_ascent}
associated with the stochastic processes
$ ( \Xi_m )_{ m \in \N_0 } $
and
$ ( \Theta_m )_{ m \in \N_0 } $
allow us to incorporate different stochastic approximation algorithms such as
\begin{itemize}
\item
plain vanilla stochastic gradient ascent with or without mini-batches
(cf.\ \eqref{eq:plain_vanilla} above)
as well as
\item
adaptive moment estimation (Adam) with mini-batches
(cf.\ Kingma \& Ba~\cite{KingmaBa2015} and
\eqref{eq:adam_1}--\eqref{eq:adam_2} in Subsection~\ref{sec:setting} below)
\end{itemize}
into the algorithm in Subsection~\ref{sec:general_case}
(cf.\ E, Han, \& Jentzen~\cite[Subsection~3.3]{EHanJentzen2017}).
The dynamics in~\eqref{eq:batch_normalization}
associated with the stochastic process
$ ( \mathbb{S}_m )_{ m \in \N_0 }$,
in turn, allow us to incorporate
batch normalisation (cf.\ Ioffe \& Szegedy~\cite{IoffeSzegedy2015} and the beginning of Section~\ref{sec:examples} below)
into the algorithm in Subsection~\ref{sec:general_case}.
In that case,
we think of
$ ( \mathbb{S}_m )_{ m \in \N_0 } $
as a bookkeeping process
keeping track of
approximatively calculated means and standard deviations
as well as
of the number of steps $ m \in \N_0 $
of the employed stochastic approximation algorithm.

\section{Numerical examples of pricing American-style de\-riv\-a\-tives}
\label{sec:examples}

In this section, we test the algorithm of Framework~\ref{algo:general}
on several examples of pricing American-style financial derivatives.
In each of the examples below, we employ the general approximation algorithm of Framework~\ref{algo:general} above
in conjunction with the Adam optimiser (cf.\ Kingma \& Ba~\cite{KingmaBa2015})
with varying learning rates and with mini-batches
(cf.\ Subsection~\ref{sec:setting} below for a precise description).

In the example in~Subsection~\ref{sec:example_TsitsiklisVanRoy1999} below,
the initial value $ \mathcal{X}_{0}^{0,1} $ is random. Therefore, we use 
$ N  $ fully connected feedforward neural networks to model the functions 
$ u_{ 0 }^{ \theta, \mathbf{s} }, \ldots, u_{ N - 1 }^{ \theta, \mathbf{s} } \colon \allowbreak \R^{d+1} \to (0,1) $,
$ \mathbf{s} \in \R^{ \varsigma } $, $ \theta \in \R^{ \nu } $. However, in all the other examples,
$\mathcal{X}_{0}^{0,1} $ is deterministic. So it is enough to learn $N-1$ networks describing the functions
$ u_{ 1 }^{ \theta, \mathbf{s} }, \ldots, u_{ N - 1 }^{ \theta, \mathbf{s} } \colon \R^{d+1} \to (0,1) $,
$ \mathbf{s} \in \R^{ \varsigma } $, $ \theta \in \R^{ \nu } $. Then it can be decided 
whether it is better to stop at time $0$ or not by comparing the deterministic pay-off 
$g(0, \mathcal{X}_{0}^{0,1} )$ to a standard Monte Carlo estimate of the expected pay-off generated by the stopping 
strategy given by $u_{ 0 }^{ \theta, \mathbf{s} } = 0$ and the functions 
$ u_{ 1 }^{ \theta, \mathbf{s} }, \ldots, u_{ N - 1 }^{ \theta, \mathbf{s} } \colon \R^{d+1} \to (0,1) $;
cf.\ \cite[Remark~6 in~Subsection~2.3]{BeckerCheriditoJentzen2019}.

The standard network architecture we use in this paper consists of a 
$(d+1)$-dimensional input layer, two $(d+50)$-dimensional hidden layers, and a one-dimensional output layer. 
As non-linear activation functions just in front of 
the hidden layers, we employ the multidimensional version of the rectifier function
$ \R \ni x \mapsto \max\{ x, 0 \} \in [ 0, \infty ) $,
whereas just in front of the output layer we employ
the standard logistic function
$ \R \ni x \mapsto \nicefrac{\exp(x)}{(\exp(x)+1)} \in ( 0, 1 ) $.
In addition, batch normalisation (cf.~Ioffe \& Szegedy~\cite{IoffeSzegedy2015})
is applied just before the first linear transformation,
just before each of the non-linear activation functions in front of the hidden layers
as well as just before the non-linear activation function in front of the output layer.
We use Xavier initialisation (cf.~Glorot \& Bengio~\cite{GlorotBengio2010})
to initialise all weights in the neural networks.

Two hidden layers work well in all our examples. However, the examples in Subsection~\ref{sec:low_dimensional} 
have an underlying one-dimensional structure, and as a consequence, fewer hidden layers yield 
equally good results; see Tables~\ref{tab:ex_BM_Bermudan2}--\ref{tab:ex_BM_Bermudan3} below. 
On the other hand, the examples in Subsection~\ref{sec:high_dimensional}
are more complex. In particular, it can be seen from Table~\ref{tab:ex_max_layers} that
for the max-call option in Subsection~\ref{sec:max_std}, two hidden layers 
give better results than zero or one hidden layer, but more than two hidden layers do
not improve the results.

All examples presented below were implemented in {\sc Python}.
The corresponding\linebreak
{\sc Python} source codes (cf.\ Section~\ref{sec:source_codes}) were run,
unless stated otherwise
(cf.\ Subsection \ref{sec:max_big} as well as the last sentence in Subsection~\ref{sec:max-call_equity} below),
in single precision (float32)
on a NVIDIA GeForce RTX~2080 Ti GPU.
The underlying system consisted of an
AMD Ryzen 9 3950X CPU with 64 GB DDR4 memory running Tensorflow~2.1 on Ubuntu~19.10.
We would like to point out that no special emphasis was put on optimising computation speed.
In many cases, some of the algorithm parameters could be adjusted in order to obtain similarly accurate results in shorter runtime.

\subsection{Theoretical considerations}
\label{sec:theory}

Before we present the optimal stopping problem examples
on which we tested the algorithm of Framework~\ref{algo:general}
(cf.~Subsections~\ref{sec:low_dimensional}--\ref{sec:high_dimensional} below),
we recall a few theoretical results, which are used to design some of these examples,
determine reference values, and provide further insights.  

\subsubsection{Option prices in the Black--Scholes model}

The elementary and well-known result in Lemma~\ref{lem:id} below specifies the distributions of
linear combinations of independent and identically distributed centred Gaussian random variables
which take values in a separable normed $ \R $-vector space.

\begin{lemma}
\label{lem:id}
Let $ n \in \N $, $ \gamma = ( \gamma_1, \ldots, \gamma_n ) \in \R^n $,
let $ ( V, \left\| \cdot \right\|_V ) $
be a separable normed $ \R $-vector space, 
let $ ( \Omega, \mathcal{F}, \P ) $ be a probability space, 
and let $ X_i \colon \Omega \to V $, $ i \in \{ 1, \ldots, n \} $,
be i.i.d.\ centred Gaussian random variables.
Then
it holds that
\begin{equation}
\biggl( \smallsum_{i = 1}^{n} \gamma_i \, X_i \biggr) ( \P )_{ \mathcal{B}(V) }
= ( \| \gamma \|_{ \R^n } \, X_1 ) (\P)_{\mathcal{B}(V)}
.
\end{equation}
\end{lemma}
\begin{proof}[Proof of Lemma~\ref{lem:id}]
Throughout this proof
let $ Y_1, Y_2 \colon \Omega \to V $
be the random variables
given by
$ Y_1 = \sum_{i = 1}^{n} \gamma_i \, X_i $
and
$ Y_2 = \| \gamma \|_{ \R^n } \, X_1 $.
Note that for every continuous linear functional $\varphi \colon V \to \R$ it holds that
$ \varphi \circ X_i \colon \Omega \to \R $, $ i \in \{ 1, \ldots, n \} $,
are independent and identically distributed centred Gaussian random variables.
This implies for all
continuous linear functionals $\varphi \colon V \to \R$ that
\begin{equation}
\begin{split}
& \E \bigl[ 
 e^{ \mathbf{i} \, \varphi( Y_1 ) } 
\bigr] 
=
\smallprod_{i=1}^{n}
    \E \bigl[ 
         e^{ \mathbf{i} \, ( \gamma_i \, \varphi ) ( X_i ) }
    \bigr]
=
\smallprod_{i=1}^{n}
    \exp\bigl( - \tfrac{1}{2} \, \E \bigl[ | ( \gamma_i \, \varphi ) ( X_i ) |^2 \bigr] \bigr)
\\ & =
\smallprod_{i=1}^{n}
    \exp\bigl( - \tfrac{1}{2} \, \E \bigl[ | ( \gamma_i \, \varphi ) ( X_1 ) |^2 \bigr] \bigr)
=
\exp\biggl( - \tfrac{1}{2} \, \E \biggl[ \smallsum_{i=1}^{n} | \gamma_i \, \varphi( X_1 ) |^2 \biggr] \biggr)
\\ & =
\exp\bigl( - \tfrac{1}{2} \, \E \bigl[ | ( \| \gamma \|_{ \R^n } \, \varphi ) ( X_1 ) |^2 \bigr] \bigr)
=
\E \bigl[
    e^{ \mathbf{i} \, \| \gamma \|_{ \R^n } \, \varphi ( X_1 ) }
\bigr]
=
\E \bigl[ 
 e^{ \mathbf{i} \, \varphi( Y_2 ) } 
\bigr]
.
\end{split}
\end{equation}
This and, e.g., Jentzen, Salimova, \& Welti~\cite[Lemma~4.10]{JentzenSalimovaWelti2017}
establish that
$Y_1( \P )_{ \mathcal{B}( V ) }
=
Y_2( \P )_{ \mathcal{B}( V ) }$.
The proof of Lemma~\ref{lem:id} is thus complete.
\end{proof}

The next elementary and well-known corollary follows directly from Lemma~\ref{lem:id}.

\begin{corollary}
\label{cor:id_BM}
Let $ d \in \N $,
$ \gamma = ( \gamma_1, \ldots, \gamma_d ) \in \R^d $,
let $ ( \Omega, \mathcal{F}, \P ) $ be a probability space,
and let $ W = ( W^{(1)}, \ldots, W^{(d)} ) \colon [0,T] \times \Omega \to \R^d $ be a $ \P $-standard Brownian motion
with continuous sample paths.
Then
it holds that
\begin{equation}
\begin{split}
\biggl(
    \smallsum_{i=1}^{d} \gamma_i \, W^{(i)}
\biggr) ( \P )_{ \mathcal{B}( C( [ 0, T ], \R ) ) }
=
\bigl( \| \gamma \|_{ \R^d } \, W^{(1)} \bigr) ( \P )_{ \mathcal{B}( C( [ 0, T ], \R ) ) }.
\end{split}
\end{equation}
\end{corollary}

The next elementary result, Proposition~\ref{prop:dim_reduction},
states that the distribution of a product of multiple correlated geometric Brownian motions
is equal to the distribution of a single particular geometric Brownian motion.

\begin{proposition}
\label{prop:dim_reduction}
Let
$ T, \epsilon \in ( 0, \infty ) $,
$ d \in \N $,
$ \mathfrak{S} = ( \varsigma_1, \ldots, \varsigma_d ) \in \R^{ d \times d } $,
$ \xi = ( \xi_1, \ldots, \xi_d ),
\alpha = ( \alpha_1, \ldots, \alpha_d ),
\beta = ( \beta_1, \ldots, \beta_d )
\in \R^d $,
let $ ( \Omega, \mathcal{F}, \P ) $ be a probability space,
let $ \mathscr{F}^{(i)} = \linebreak ( \mathscr{F}_t^{(i)} )_{ t \in [0,T] } $, $ i \in \{ 1 ,2 \} $,
be filtrations on $ ( \Omega, \mathcal{F}, \P ) $
that satisfy the usual conditions,
let $ W = ( W^{(1)}, \ldots, W^{(d)} ) \colon [0,T] \times \Omega \to \R^d $
be a standard $ ( \Omega, \mathcal{F}, \P, \mathscr{F}^{(1)} ) $-Brownian motion
with continuous sample paths, 
let $ \textsc{w} \colon [0,T] \times \Omega \to \R $
be a standard $ ( \Omega, \mathcal{F}, \P, \mathscr{F}^{(2)} ) $-Brownian motion
with continuous sample paths, 
let $ \mu \colon \R^d \to \R^d $,
$ \sigma \colon \R^d \to \R^{ d \times d } $,
$ \mathbf{P} \colon C( [ 0, T ], \R^d ) \to C( [ 0, T ], \R ) $,
and $ \mathbf{G} \colon C( [ 0, T ], \R ) \to C( [ 0, T ], \R ) $
be the functions which satisfy
for all $ x = ( x_1, \ldots, x_d ) \in \R^d $,
$ u^{(1)} = ( u_s^{(1)} )_{ s \in [ 0, T ] }, \ldots, u^{(d)} = ( u_s^{(d)} )_{ s \in [ 0, T ] } \in C( [ 0, T ], \R ) $,
$ t \in [ 0, T ] $
that
$ \mu(x) = ( \alpha_1 x_1, \ldots, \alpha_d x_d ) $,
$ \sigma(x) = \operatorname{diag}( \beta_1 x_1, \ldots, \beta_d x_d ) \, \mathfrak{S}^* $,
$ ( \mathbf{G}[ u^{(1)} ] )_t
=
\exp
\bigl( \epsilon \,
    \bigl[
        \sum_{i=1}^{d} \alpha_i - \nicefrac{ \| \beta_i \varsigma_i \|_{ \R^d }^2 }{2}
    \bigr] t
     \mathop{+}
    \epsilon \,
    \| \mathfrak{S} \, \beta \|_{ \R^d } \, u_t^{(1)}
\bigr)
\prod_{i=1}^{d} | \xi_i |^\epsilon
$,
and
$ ( \mathbf{P}[ ( u^{(1)}, \ldots, u^{(d)} ) ] )_t = \prod_{i=1}^{d} \bigl| u_t^{(i)} \bigr|^\epsilon $,
let $ X = \linebreak ( X^{(1)}, \ldots, X^{(d)} ) \colon [0,T] \times \Omega \to \R^d $ 
be an $ \mathscr{F}^{(1)} $-adapted stochastic process with continuous sample paths,
let $ Y \colon [0,T] \times \Omega \to \R $
be an $ \mathscr{F}^{(2)} $-adapted stochastic process with continuous sample paths,
and assume
that for all $ t \in [ 0, T ] $
it holds $ \P $-a.s.\ that
\begin{align}
X_t
& = \xi +
\int_{0}^{t} \mu( X_s ) \, ds +
\int_{0}^{t} \sigma( X_s ) \, dW_s,
\\
Y_t
& =
\smallprod_{i=1}^{d} | \xi_i |^\epsilon
+
\biggl( \epsilon \,
    \biggl[
        \smallsum_{i=1}^{d} \alpha_i - \tfrac{ \| \beta_i \varsigma_i \|_{ \R^d }^2 }{2}
    \biggr]
    +
    \tfrac{ \| \epsilon \, \mathfrak{S} \, \beta \|_{ \R^d }^2 }{2}
\biggr)
\int_{0}^{t} Y_s \, ds
+
\epsilon \, \| \mathfrak{S} \, \beta \|_{ \R^d } \int_{0}^{t} Y_s \, d\textsc{w}_s
.
\end{align}
Then
\begin{enumerate}[(i)]
\item
\label{item:geometricBM}
for all $ i \in \{ 1, \ldots, d \} $,
$ t \in [ 0, T ] $
it holds $ \P $-a.s.\ that
\begin{equation}
X_t^{(i)}
=
\exp\Bigl(
\Bigl[ \alpha_i - \tfrac{ \| \beta_i \varsigma_i \|_{ \R^d }^2 }{2} \Bigr] t
+
\beta_i \langle \varsigma_i, W_t \rangle_{ \R^d }
\Bigr)
\, \xi_i,
\end{equation}
\item
\label{item:continuous}
it holds that
$ \mathbf{P} $ 
and $ \mathbf{G} $ 
are continuous functions,
and
\item
\label{item:equal_dist}
it holds that
\begin{equation}
( \mathbf{P} \circ X )( \P )_{ \mathcal{B}( C( [ 0, T ], \R ) ) }
=
( \mathbf{G} \circ \textsc{w} ) ( \P )_{ \mathcal{B}( C( [ 0, T ], \R ) ) }
=
Y ( \P )_{ \mathcal{B}( C( [ 0, T ], \R ) ) }.
\end{equation}
\end{enumerate}
\end{proposition}
\begin{proof}[Proof of Proposition~\ref{prop:dim_reduction}]
Throughout this proof
let $ \gamma = ( \gamma_1, \ldots, \gamma_d ) \in \R^d $
be the vector given by
$ \gamma = \mathfrak{S} \, \beta $,
let
$ Z^{(i)} \colon [ 0, T ] \times \Omega \to \R $, $ i \in \{ 1, \ldots, d \} $,
be the stochastic processes
which satisfy
for all $ i \in \{ 1, \ldots, d \} $, $ t \in [ 0, T ] $ that
$ Z_t^{(i)}
=
[ \alpha_i - \nicefrac{ \| \beta_i \varsigma_i \|_{ \R^d }^2 }{2} ] t
+
\beta_i \langle \varsigma_i, W_t \rangle_{ \R^d } $,
and let
$ \tilde{\mathbf{G}} \colon C( [ 0, T ], \R ) \to C( [ 0, T ], \R ) $
be the function which satisfies
for all $ u = ( u_s )_{ s \in [ 0, T ] } \in C( [ 0, T ], \R ) $, $ t \in [ 0, T] $ that
\begin{equation}
( \tilde{\mathbf{G}}[ u ] )_t
=
\exp
\biggl( \epsilon \,
    \biggl[
        \smallsum_{i=1}^{d} \alpha_i - \tfrac{ \| \beta_i \varsigma_i \|_{ \R^d }^2 }{2}
    \biggr] t
    +
    \epsilon \, u_t
\biggr)
\smallprod_{i=1}^{d} | \xi_i |^\epsilon
.
\end{equation}
Observe that
for all
$ i \in \{ 1, \ldots, d \} $,
$ t \in [ 0, T ] $
it holds $ \P $-a.s.\ that
\begin{equation}
\label{eq:X^i_evolution}
X_t^{(i)}
= \xi_i +
\alpha_i \int_{0}^{t} X_s^{(i)} \, ds +
\beta_i \int_{0}^{t} X_s^{(i)} \langle \varsigma_i, dW_s \rangle_{ \R^d }.
\end{equation}
In addition, note that for all
$ i \in \{ 1, \ldots, d \} $,
$ t \in [ 0, T ] $
it holds $ \P $-a.s.\ that
\begin{equation}
Z_t^{(i)}
=
\int_{0}^{t}
\Bigl[
\alpha_i - \tfrac{ \| \beta_i \varsigma_i \|_{ \R^d }^2 }{2}
\Bigr] \, ds
+
\int_{0}^{t} \beta_i \langle \varsigma_i, dW_s \rangle_{ \R^d }.
\end{equation}
It\^o's formula hence shows that
for all
$ i \in \{ 1, \ldots, d \} $,
$ t \in [ 0, T ] $
it holds $ \P $-a.s.\ that
\begin{align}
\nonumber
e^{Z_t^{(i)}} \xi_i
& =
\xi_i
+
\Bigl[ \alpha_i - \tfrac{ \| \beta_i \varsigma_i \|_{ \R^d }^2 }{2} \Bigr]
\int_{0}^{t}
    e^{Z_s^{(i)}} \xi_i
\, ds
+
\beta_i
\int_{0}^{t}
     e^{Z_s^{(i)}} \xi_i
    \langle \varsigma_i, dW_s \rangle_{ \R^d }
+
\tfrac{ \| \beta_i \varsigma_i \|_{ \R^d }^2 }{2}
\int_{0}^{t}
    e^{Z_s^{(i)}} \xi_i
\, ds
\\ & =
\xi_i +
\alpha_i \int_{0}^{t} e^{Z_s^{(i)}} \xi_i \, ds +
\beta_i \int_{0}^{t} e^{Z_s^{(i)}} \xi_i \langle \varsigma_i, dW_s \rangle_{ \R^d }.
\end{align}
Combining this and \eqref{eq:X^i_evolution}
with, e.g., Da~Prato \& Zabczyk~\cite[(i) in Theorem~7.4]{DaPratoZabczyk1992}
proves that
for all
$ i \in \{ 1, \ldots, d \} $,
$ t \in [ 0, T ] $
it holds $ \P $-a.s.\ that
\begin{equation}
X_t^{(i)}
=
e^{Z_t^{(i)}} \xi_i
=
\exp\Bigl(
\Bigl[ \alpha_i - \tfrac{ \| \beta_i \varsigma_i \|_{ \R^d }^2 }{2} \Bigr] t
+
\beta_i \langle \varsigma_i, W_t \rangle_{ \R^d }
\Bigr)
\, \xi_i.
\end{equation}
This establishes~\eqref{item:geometricBM}.
In the next step note that~\eqref{item:continuous} is clear.
It thus remains to prove~\eqref{item:equal_dist}.
For this
observe that
\eqref{item:geometricBM}
establishes that
for all
$ t \in [ 0, T ] $
it holds $ \P $-a.s.\ that
\begin{equation}
\begin{split}
( \mathbf{P}[ X ] )_t
& =
\smallprod_{i=1}^{d} \bigl| X_t^{(i)} \bigr|^\epsilon
=
\smallprod_{i=1}^{d}
\Bigl[
\exp\Bigl(
\epsilon \, \Bigl[ \alpha_i - \tfrac{ \| \beta_i \varsigma_i \|_{ \R^d }^2 }{2} \Bigr] t
+
\epsilon \, \beta_i \langle \varsigma_i, W_t \rangle_{ \R^d }
\Bigr)
| \xi_i |^\epsilon
\Bigr]
\\ & =
\exp\biggl(
\epsilon \, \biggl[ \smallsum_{i=1}^{d} \alpha_i - \tfrac{ \| \beta_i \varsigma_i \|_{ \R^d }^2 }{2} \biggr] t
+
\epsilon \, \langle \gamma, W_t \rangle_{ \R^d }
\biggr)
\smallprod_{i=1}^{d}
| \xi_i |^\epsilon
=
\biggl( \tilde{\mathbf{G}} \biggl[ \smallsum_{i=1}^{d} \gamma_i \, W^{(i)} \biggr] \biggr)_t
.
\end{split}
\end{equation}
Continuity
hence implies that
it holds $ \P $-a.s.\ that
\begin{equation}
\label{eq:law_P(X)}
\mathbf{P}[ X ]
=
\tilde{\mathbf{G}} \biggl[ \smallsum_{i=1}^{d} \gamma_i \, W^{(i)} \biggr]
.
\end{equation}
Moreover, note that
\eqref{item:geometricBM}
shows
that for all $ t \in [ 0, T ] $
it holds $ \P $-a.s.\ that
\begin{equation}
\begin{split}
Y_t
& =
\exp
\biggl( \biggl\{ \epsilon \,
    \biggl[
        \smallsum_{i=1}^{d} \alpha_i - \tfrac{ \| \beta_i \varsigma_i \|_{ \R^d }^2 }{2}
    \biggr]
    + \tfrac{ \| \epsilon \, \mathfrak{S} \, \beta \|_{ \R^d }^2 }{2}
    - \tfrac{ \| \epsilon \, \mathfrak{S} \, \beta \|_{ \R^d }^2 }{2}
    \biggr\} t
    +
    \epsilon \,
    \| \mathfrak{S} \, \beta \|_{ \R^d } \, \textsc{w}_t
\biggr)
\smallprod_{i=1}^{d} | \xi_i |^\epsilon
\\ & =
\exp
\biggl( \epsilon \,
    \biggl[
        \smallsum_{i=1}^{d} \alpha_i - \tfrac{ \| \beta_i \varsigma_i \|_{ \R^d }^2 }{2}
    \biggr] t
    +
    \epsilon \,
    \| \mathfrak{S} \, \beta \|_{ \R^d } \, \textsc{w}_t
\biggr)
\smallprod_{i=1}^{d} | \xi_i |^\epsilon
=
( \mathbf{G} [ \textsc{w} ] )_t
.
\end{split}
\end{equation}
This and continuity
establish that
it holds $ \P $-a.s.\ that
\begin{equation}
\label{eq:law_Y}
Y
=
\mathbf{G} [ \textsc{w} ]
.
\end{equation}
Furthermore,
observe that
Corollary~\ref{cor:id_BM}
ensures that
\begin{equation}
\begin{split}
\biggl(
    \smallsum_{i=1}^{d} \gamma_i \, W^{(i)}
\biggr) ( \P )_{ \mathcal{B}( C( [ 0, T ], \R ) ) }
=
\bigl( \| \gamma \|_{ \R^d } \, W^{(1)} \bigr) ( \P )_{ \mathcal{B}( C( [ 0, T ], \R ) ) }.
\end{split}
\end{equation}
The fact that
$ \tilde{\mathbf{G}} \colon C( [ 0, T ], \R ) \to C( [ 0, T ], \R ) $
is a 
Borel measurable function,
\eqref{item:continuous},
the fact that
$ \forall \, u \in C( [ 0, T ], \R ) \colon
\tilde{\mathbf{G}} [ \| \gamma \|_{ \R^d } \, u ] = \mathbf{G}[ u ] $,
\eqref{eq:law_P(X)}, and \eqref{eq:law_Y} hence
demonstrate that
\begin{equation}
\begin{split}
& ( \mathbf{P} \circ X )( \P )_{ \mathcal{B}( C( [ 0, T ], \R ) ) }
=
\biggl(
    \tilde{\mathbf{G}} \circ \biggl( \smallsum_{i=1}^{d} \gamma_i \, W^{(i)} \biggr)
\biggr)( \P )_{ \mathcal{B}( C( [ 0, T ], \R ) ) }
\\ & =
\bigl(
    \tilde{\mathbf{G}} \circ \bigl( \| \gamma \|_{ \R^d } \, W^{(1)} \bigr)
\bigr)( \P )_{ \mathcal{B}( C( [ 0, T ], \R ) ) }
=
\bigl( \mathbf{G} \circ W^{(1)} \bigr) ( \P )_{ \mathcal{B}( C( [ 0, T ], \R ) ) }
\\ & =
( \mathbf{G} \circ \textsc{w} ) ( \P )_{ \mathcal{B}( C( [ 0, T ], \R ) ) }
=
Y ( \P )_{ \mathcal{B}( C( [ 0, T ], \R ) ) }
.
\end{split}
\end{equation}
The proof of Proposition~\ref{prop:dim_reduction} is thus complete.
\end{proof}

In the next result, Lemma~\ref{lem:black_scholes},
we recall the well-known formula for the price of a European call option
on a single stock in the Black--Scholes model
(cf., e.g., {\O}ksendal~\cite[Corollary~12.3.8]{Oksendal2003}).

\begin{lemma} \label{lem:black_scholes}
Let
$ T, \xi, \sigma \in ( 0, \infty ) $,
$ r, c \in \R $,
let $ \Phi \colon \R \to \R $ be the function which 
satisfies for all $ x \in \R $ that
$
\Phi( x ) = \int_{ - \infty }^x \frac{ 1 }{ \sqrt{ 2 \pi } } \, e^{ - \frac{ 1 }{ 2 } y^2 } \, dy
$,
let $ ( \Omega, \mathcal{F}, \P ) $ be a probability space
with a filtration $  \mathscr{F} = ( \mathscr{F}_t )_{ t \in [0,T] } $
that satisfies the usual conditions,
let $ W \colon [0,T] \times \Omega \to \R $
be a standard $ ( \Omega, \mathcal{F}, \P, \mathscr{F} ) $-Brownian motion
with continuous sample paths,
and
let $ X \colon [0,T] \times \Omega \to \R $
be an $ \mathscr{F} $-adapted stochastic process with continuous sample paths
which satisfies that
for all $ t \in [ 0, T ] $
it holds $ \P $-a.s.\ that
\begin{equation}
X_t
= \xi +
( r - c ) \int_{0}^{t} X_s \, ds +
\sigma \int_{0}^{t} X_s \, dW_s
.
\end{equation}
Then it holds
for all $ K \in \R $ that
\begin{equation}
\begin{split}
& \E \bigl[
e^{ -r T }
\max \{ X_T - K, 0 \}
\bigr]
\\ & =
\begin{cases}
e^{ - c \, T} \xi
\,
\Phi \Bigl(
\tfrac{ {\textstyle (} r - c + \frac{ \sigma^2 }{2} {\textstyle )} T + \ln( \nicefrac{\xi}{K} ) }{ \sigma \sqrt{T} }
\Bigr)
-
K e^{ - r \, T}
\,
\Phi \Bigl(
\tfrac{ {\textstyle (} r - c - \frac{ \sigma^2 }{2} {\textstyle )} T + \ln( \nicefrac{\xi}{K} ) }{ \sigma \sqrt{T} }
\Bigr)
&
\colon
K > 0
\\[1ex]
e^{ - c \, T} \xi
- K e^{ - r \, T}
&
\colon
K \leq 0
\end{cases}
.
\end{split}
\end{equation}
\end{lemma}

\subsubsection{Approximating American options with Bermudan options}

In our numerical simulations, we approximate Bermudan options 
with a finite number of execution times 
rather than American options, which theoretically can be executed 
at infinitely many time points (any time before maturity). However, the following result 
shows that the prices of American options can be approximated with prices of Bermudan options 
with equidistant execution times if the number of execution times is sufficiently large. 

\begin{lemma}
\label{lem:American_Bermudian}
Let $ T \in (0,\infty) $, $ d, N \in \N $, $ \mathfrak{c}, \mathfrak{C} \in \R $, 
let $ ( \Omega, \mathcal{F}, \P ) $ be a probability space 
with a filtration $ \mathbb{F} = ( \mathbb{F}_t )_{ t \in [ 0, T ] } $, 
let $ X \colon [0,T] \times \Omega \to \R^d $ be a stochastic 
process with continuous sample paths, 
let $ \left\| \cdot \right\| \colon \R^d \to [0,\infty) $ be a norm, 
assume 
$
  \E[ \sup_{ t \in [0,T] } \| X_t \| ] < \infty 
$, 
and let 
$ \lceil \cdot \rceil \colon [0,T] \to [0,T] $ 
and 
$ g \colon [0,T] \times \R^d \to \R $ satisfy for all 
$ t, \mathfrak{t} \in [0,T] $, $ x, \mathfrak{x} \in \R^d $ that 
$ 
  \lceil t \rceil = \min( [t,T] \cap \{ 0, \nicefrac{ T }{ N }, \dots, T \} )
$
and 
$ 
  | g(t,x) - g(\mathfrak{t},\mathfrak{x}) | \leq \mathfrak{c} | t - \mathfrak{t} |^{ \nicefrac{ 1 }{ 2 } } 
  + \mathfrak{C} \| x - \mathfrak{x} \| 
$. Then
\begin{enumerate}[(i)]
\item 
\label{item:i_AmericanBermudan}
it holds that 
$ 
  \E\bigl[ \sup_{ t \in [ 0, T ] } | g(t, X_t) | \bigr] 
  \leq 
  | g(0, 0) |
  +
  \mathfrak{c} T^{ \nicefrac{ 1 }{ 2 } } 
  +
  \mathfrak{C} \,
  \E\bigl[
    \sup\nolimits_{ t \in [0,T] }
    \| X_t \|
  \bigr]
  < \infty  
$
and 
\item 
\label{item:ii_AmericanBermudan}
it holds that
\begin{align}
&
  \left|
  \sup\!\left\{ 
    \E\bigl[ 
      g( \tau, X_{ \tau } )
    \bigr]
    \colon
    \substack{
      \tau \colon \Omega \to [0,T] \text{ is an}
    \\
      \mathbb{F}\text{-stopping time}
    }
  \right\}
  -
  \sup\!\left\{ 
    \E\bigl[ 
      g( \tau, X_{ \tau } )
    \bigr]
    \colon
    \substack{
      \tau \colon \Omega \to \{ 0, \nicefrac{ T }{ N }, \dots, T \}\text{ is an}
    \\
      ( \mathbb{F}_{ \nicefrac{ n T }{ N } } )_{ n \in \{ 0, 1, \dots, N \} } \text{-stopping time}
    }
  \right\}
  \right|
\nonumber
\\ &
\leq 
  \mathfrak{c} \left| \tfrac{ T }{ N } \right|^{ \nicefrac{ 1 }{ 2 } }
  +
  \mathfrak{C}
  \,
  \E\bigl[ 
    \sup\nolimits_{ t \in [0,T] }
    \| X_t - X_{ \lceil t \rceil } \|
  \bigr]
  .
\end{align}
\end{enumerate}
\end{lemma}
\begin{proof}[Proof of 
Lemma~\ref{lem:American_Bermudian}] 
Throughout this proof let 
$ \lfloor \cdot \rfloor \colon [0,T] \to [0,T] $ 
satisfy for all 
$
  t \in [0,T]
$
that 
$ 
  \lfloor t \rfloor = \max( [0,t] \cap \{ 0, \nicefrac{ T }{ N }, \dots, T \} )
$. 
Note that the assumption 
that for all
$ t, \mathfrak{t} \in [0,T] $, $ x, \mathfrak{x} \in \R^d $ 
it holds that 
$ 
  | g(t,x) - g(\mathfrak{t},\mathfrak{x}) | \leq \mathfrak{c} | t - \mathfrak{t} |^{ \nicefrac{ 1 }{ 2 } } 
  + \mathfrak{C} \| x - \mathfrak{x} \| 
$ 
implies that 
\begin{equation}
\begin{split}
  \E\bigl[ 
    \sup\nolimits_{ t \in [0,T] }
    | g( t, X_t ) |
  \bigr]
& 
\leq 
  \E\bigl[ 
    \sup\nolimits_{ t \in [0,T] }
    | g( t, X_t ) - g(0, 0) |
  \bigr]
  +
  | g(0, 0) |
\\ & \leq 
  \E\bigl[
    \sup\nolimits_{ t \in [0,T] }
    \big(
      \mathfrak{c} | t |^{ \nicefrac{ 1 }{ 2 } } 
      +
      \mathfrak{C} \| X_t \|
    \big)
  \bigr]
  +
  | g(0, 0) |
\\ &
\leq
  | g(0, 0) |
  +
  \mathfrak{c} T^{ \nicefrac{ 1 }{ 2 } } 
  +
  \mathfrak{C} \,
  \E\bigl[
    \sup\nolimits_{ t \in [0,T] }
    \| X_t \|
  \bigr]
  < \infty .
\end{split}
\end{equation}
This establishes~\eqref{item:i_AmericanBermudan}.
Next observe that the fact 
$ \forall \, r \in \{ 0, \nicefrac{ T }{ N }, \dots, T \}, \, t \in [r,T] \colon
  r \leq \lfloor t \rfloor 
$
ensures that 
for all $ r, t \in [0,T] $ 
with 
$
  \lceil r \rceil \leq t
$
it holds that 
$
  \lceil r \rceil \leq \lfloor t \rfloor
$. 
This implies that
for every 
$
  \mathbb{F}
$-stopping time 
$
  \rho \colon \Omega \to [0,T]
$
and every $ t \in [0,T] $ 
it holds that
\begin{equation}
  \left\{ \lceil \rho \rceil \leq t \right\}
  =
  \left\{ \lceil \rho \rceil \leq \lfloor t \rfloor \right\}
  =
  \left\{ \rho \leq \lfloor t \rfloor \right\} 
  .
\end{equation}
The fact that 
for every 
$
  \mathbb{F}
$-stopping time 
$
  \rho \colon \Omega \to [0,T]
$
and every $ t \in [0,T] $ 
it holds that
$
  \left\{ \rho \leq \lfloor t \rfloor \right\} 
  \in 
  \mathbb{F}_{ 
    \lfloor t \rfloor  
  }
$
and the fact that for every 
$ t \in [0,T] $
it holds that
$
  \mathbb{F}_{ 
    \lfloor t \rfloor  
  }
  \subseteq 
  \mathbb{F}_t
$
hence show that
for every 
$
  \mathbb{F}
$-stopping time 
$
  \rho \colon \Omega \to [0,T]
$
and every $ t \in [0,T] $ 
it holds that
$
  \left\{ \lceil \rho \rceil \leq t \right\}
  \in 
  \mathbb{F}_t
$. 
This proves that 
for every 
$
  \mathbb{F}
$-stopping time 
$
  \rho \colon \Omega \to [0,T]
$
it holds that
$
  \Omega \ni \omega \mapsto 
  \lceil \rho( \omega ) \rceil \in [0,T]
$
is an 
$
  \mathbb{F}
$-stopping time. 
Hence, we obtain that
for every 
$
  \mathbb{F}
$-stopping time 
$
  \rho \colon \Omega \to [0,T]
$
it holds that
$
  \Omega \ni \omega \mapsto 
  \lceil \rho( \omega ) \rceil \in \{ 0, \nicefrac{ T }{ N }, \dots, T \}
$
is an 
$
  ( \mathbb{F}_{ \nicefrac{ n T }{ N } } )_{ n \in \{ 0, 1, \dots, N \} }
$-stopping time.
This ensures that 
for every 
$
  \mathbb{F}
$-stopping time 
$
  \rho \colon \Omega \to [0,T]
$
it holds that
\begin{equation}
\begin{split}
&
  \E\bigl[ 
    g( \rho, X_{ \rho } )
  \bigr]
  -
  \sup\!\left\{ 
    \E\bigl[ 
      g( \tau, X_{ \tau } )
    \bigr]
    \colon
    \substack{
      \tau \colon \Omega \to \{ 0, \nicefrac{ T }{ N }, \dots, T \}\text{ is an}
    \\
      ( \mathbb{F}_{ \nicefrac{ n T }{ N } } )_{ n \in \{ 0, 1, \dots, N \} } \text{-stopping time}
    }
  \right\}
\\ &
\leq 
  \E\bigl[ 
    g( \rho, X_{ \rho } )
  \bigr]
  -
  \E\bigl[ 
    g( \lceil \rho \rceil, X_{ \lceil \rho \rceil } )
  \bigr]
\leq 
  \big|
  \E\bigl[ 
    g( \rho, X_{ \rho } )
  \bigr]
  -
  \E\bigl[ 
    g( \lceil \rho \rceil, X_{ \lceil \rho \rceil } )
  \bigr]
  \big|
\\ &
\leq
  \E\bigl[ 
    |
      g( \rho, X_{ \rho } )
      -
      g( \lceil \rho \rceil, X_{ \lceil \rho \rceil } )
    |
  \bigr]
\leq 
  \E\bigl[
    \mathfrak{c} | \rho - \lceil \rho \rceil |^{ \nicefrac{ 1 }{ 2 } }
    +
    \mathfrak{C} 
    \| 
      X_{ \rho } - X_{ \lceil \rho \rceil }
    \|
  \bigr]
  .
\end{split}
\end{equation}
Therefore, we obtain that 
for every 
$
  \mathbb{F}
$-stopping time 
$
  \rho \colon \Omega \to [0,T]
$
it holds that
\begin{equation}
\begin{split}
&
  \E\bigl[ 
    g( \rho, X_{ \rho } )
  \bigr]
  -
  \sup\!\left\{ 
    \E\bigl[ 
      g( \tau, X_{ \tau } )
    \bigr]
    \colon
    \substack{
      \tau \colon \Omega \to \{ 0, \nicefrac{ T }{ N }, \dots, T \}\text{ is an}
    \\
      ( \mathbb{F}_{ \nicefrac{ n T }{ N } } )_{ n \in \{ 0, 1, \dots, N \} } \text{-stopping time}
    }
  \right\}
\\ & 
\leq 
  \E\biggl[ 
    \sup_{ t \in [0,T] }
    \Bigl( 
      \mathfrak{c} \left| t - \lceil t \rceil \right|^{ \nicefrac{ 1 }{ 2 } }
      +
      \mathfrak{C}
      \| X_t - X_{ \lceil t \rceil } \|
    \Bigr)
  \biggr]
\leq 
  \E\biggl[ 
    \sup_{ t \in [0,T] }
    \Bigl( 
      \mathfrak{c} \left| \tfrac{ T }{ N } \right|^{ \nicefrac{ 1 }{ 2 } }
      +
      \mathfrak{C}
      \| X_t - X_{ \lceil t \rceil } \|
    \Bigr)
  \biggr]
\\ & = 
  \mathfrak{c} \left| \tfrac{ T }{ N } \right|^{ \nicefrac{ 1 }{ 2 } }
  +
  \mathfrak{C}
  \,
  \E\bigl[ 
    \sup\nolimits_{ t \in [0,T] }
    \| X_t - X_{ \lceil t \rceil } \|
  \bigr]
  .
\end{split}
\end{equation}
This implies that
\begin{align*}
&
  \left|
  \sup\!\left\{ 
    \E\bigl[ 
      g( \tau, X_{ \tau } )
    \bigr]
    \colon
    \substack{
      \tau \colon \Omega \to [0,T] \text{ is an}
    \\
      \mathbb{F}\text{-stopping time}
    }
  \right\}
  -
  \sup\!\left\{ 
    \E\bigl[ 
      g( \tau, X_{ \tau } )
    \bigr]
    \colon
    \substack{
      \tau \colon \Omega \to \{ 0, \nicefrac{ T }{ N }, \dots, T \}\text{ is an}
    \\
      ( \mathbb{F}_{ \nicefrac{ n T }{ N } } )_{ n \in \{ 0, 1, \dots, N \} } \text{-stopping time}
    }
  \right\}
  \right|
\\ &
=
  \sup\!\left\{ 
    \E\bigl[ 
      g( \tau, X_{ \tau } )
    \bigr]
    \colon
    \substack{
      \tau \colon \Omega \to [0,T] \text{ is an}
    \\
      \mathbb{F}\text{-stopping time}
    }
  \right\}
  -
  \sup\!\left\{ 
    \E\bigl[ 
      g( \tau, X_{ \tau } )
    \bigr]
    \colon
    \substack{
      \tau \colon \Omega \to \{ 0, \nicefrac{ T }{ N }, \dots, T \}\text{ is an}
    \\
      ( \mathbb{F}_{ \nicefrac{ n T }{ N } } )_{ n \in \{ 0, 1, \dots, N \} } \text{-stopping time}
    }
  \right\}
\\ & \yesnumber
\leq 
  \mathfrak{c} \left| \tfrac{ T }{ N } \right|^{ \nicefrac{ 1 }{ 2 } }
  +
  \mathfrak{C}
  \,
  \E\bigl[ 
    \sup\nolimits_{ t \in [0,T] }
    \| X_t - X_{ \lceil t \rceil } \|
  \bigr]
  .
\end{align*}
This shows~\eqref{item:ii_AmericanBermudan}. 
The proof of Lemma~\ref{lem:American_Bermudian} is thus complete.
\end{proof}

\subsection{Setting}
\label{sec:setting}

\begin{algo}
\label{algo:examples-setting}
Assume Framework~\ref{algo:general}, let $ \zeta_1 = 0.9 $,
$ \zeta_2 = 0.999 $, $ \varepsilon \in ( 0, \infty ) $,
$ ( \gamma_m )_{ m \in \N } \subseteq ( 0, \infty ) $,
$\xi = ( \xi_1, \ldots, \xi_d ) \in \R^d $,
let $ \mathscr{F} = ( \mathscr{F}_t )_{ t \in [ 0, T ] } $ be a filtration on $ ( \Omega, \mathcal{F}, \P ) $
that satisfies the usual conditions, 
let $ W^{m,j} = ( W^{ m,j, (1) }, \ldots, W^{ m,j, (d) } ) \colon [0,T] \times \Omega \to \R^d $, $ m \in \N_0 $, $ j \in \N $,
be independent standard $ ( \Omega, \mathcal{F}, \P, \mathscr{F} ) $-Brownian motions
with continuous sample paths,
let $ \mu \colon \R^d \to \R^d $
and $ \sigma \colon \R^d \to \R^{ d \times d } $
be Lipschitz continuous functions,
let $ X = ( X^{(1)}, \ldots, X^{(d)} ) \colon [0,T] \times \Omega \to \R^d $ 
be an $ \mathscr{F} $-adapted stochastic process with continuous sample paths
which satisfies
that for all $ t \in [ 0, T ] $
it holds $ \P $-a.s.\ that
\begin{equation}
X_t
 = \xi +
\int_{0}^{t} \mu( X_s ) \, ds +
\int_{0}^{t} \sigma( X_s ) \, dW_s^{0,1},
\end{equation}
assume for all
$ n \in \{ 0, 1, \ldots, N \} $
that
$ \varrho = 2 \nu $,
$ \Xi_0 = 0 $,
and $ t_n = \frac{nT}{N} $,
and
assume for all $ m \in \N $,
$ x = ( x_1, \ldots, x_\nu ),\,
y = ( y_1, \ldots, y_\nu ),\,
\eta = ( \eta_1, \ldots, \eta_\nu )
\in \R^\nu $
that
\begin{equation}
\label{eq:adam_1}
\Psi_m( x, y, \eta )
=
\bigl(
    \zeta_1 x + ( 1 - \zeta_1) \eta,
    \zeta_2 y + ( 1 - \zeta_2 ) ( ( \eta_1 )^2, \ldots, ( \eta_\nu )^2 )
\bigr)
\end{equation}
and
\begin{equation}
\label{eq:adam_2}
\psi_m( x, y )
=
\Biggl(
    \biggl[
        \sqrt{
            \tfrac{ \left| y_1 \right| }{ 1 - ( \zeta_2 )^m }
        }
        + \varepsilon
    \biggr]^{ -1 }
    \frac{ \gamma_m x_1 }{ 1 - ( \zeta_1 )^m },
    \ldots,
    \biggl[
        \sqrt{
            \tfrac{ \left| y_\nu \right| }{ 1 - ( \zeta_2 )^m }
        }
        + \varepsilon
    \biggr]^{ -1 }
    \frac{ \gamma_m x_\nu }{ 1 - ( \zeta_1 )^m }
\Biggr)
.
\end{equation}
\end{algo}
Equations \eqref{eq:adam_1}--\eqref{eq:adam_2} in Framework~\ref{algo:examples-setting}
describe the Adam optimiser
with possibly varying learning rates
(cf.\ Kingma \& Ba~\cite{KingmaBa2015},
e.g., E, Han, \& Jentzen~\cite[(4.3)--(4.4) in Subsection~4.1 and (5.4)--(5.5) in Subsection~5.2]{EHanJentzen2017},
and
lines 93--96 in {\sc Python} code~\ref{code:common} in Subsection~\ref{sec:code_common} below).
Furthermore, in the context of pricing American-style financial derivatives, we think
\begin{itemize}
\item
of $ T $ as the maturity,
\item
of $ d $ as the dimension of the associated optimal stopping problem,
\item
of $ N $ as the time discretisation parameter employed,
\item
of $ M $ as the total number of training steps employed in the Adam optimiser,
\item
of $ g $ as the discounted pay-off function,
\item
of $ \{ t_0, t_1, \ldots, t_N \} $ as the discrete time grid employed,
\item
of $ J_0 $ as the number of Monte Carlo samples employed in the final integration for the price approximation,
\item
of $ ( J_m )_{ m \in \N } $ as the sequence of batch sizes employed in the Adam optimiser,
\item
of $ \zeta_1 $ as the momentum decay factor,
of $ \zeta_2 $ as the second momentum decay factor,
and of $ \varepsilon $ as the regularising factor employed in the Adam optimiser,
\item
of $ ( \gamma_m )_{ m \in \N } $ as the sequence of learning rates employed in the Adam optimiser,
\item
and, where applicable, of
$ X $ as a continuous-time model for $ d $ underlying stock prices
with initial prices $ \xi $, drift coefficient function $ \mu $, and diffusion coefficient function $ \sigma $.
\end{itemize}
Moreover,
note that for every $ m \in \N_0 $, $ j \in \N $
the stochastic processes
$ W^{ m,j, (1) } = \linebreak ( W^{ m,j, (1) }_t )_{ t \in [ 0, T ] } $,
\ldots,
$ W^{ m,j, (d) } = ( W^{ m,j, (d) }_t )_{ t \in [ 0, T ] } $
are the components of the
$ d $-dimensional standard Brownian motion
$ W^{m,j} = ( W^{m,j}_t )_{ t \in [ 0, T ] } $
and hence constitute each
a one-dimen\-sion\-al standard Brownian motion.

\subsection{Examples with known one-dimensional representation}
\label{sec:low_dimensional}

In this subsection, we test the algorithm of Framework~\ref{algo:general}
on different $d$-dimensional optimal stopping problems that
can be represented as one-dimensional optimal stopping problems.
This representation allows us to employ a numerical method for the
one-dimensional optimal stopping problem to compute reference values
for the original $ d $-dimensional optimal stopping problem.
We refer to Subsection~\ref{sec:high_dimensional} below for more challenging examples
where a one-dimensional representation is not known.

\subsubsection{Optimal stopping of a Brownian motion}
\label{sec:example_BM}

\paragraph[A Bermudan three-exercise put-type example]{A Bermudan put-type example with three exercise opportunities}
\label{sec:example_BM_Bermudan}

In this subsection, we test the algorithm of Framework~\ref{algo:general}
on the example of optimally stopping
a correlated Brownian motion
under a put option inspired pay-off function
with three possible exercise dates.
Among other things, we examine the performance of the algorithm for different numbers of hidden layers of the employed neural networks.

Assume Framework~\ref{algo:examples-setting},
let
$ H \in \N_0 $,
$ r = 0.02 = 2 \% $,
$ \beta = 0.3 = 30 \% $,
$ \chi = 95 $,
$ K = 90 $,
$ Q = ( Q_{ i, j } )_{ ( i, j ) \in \{ 1, \ldots, d \}^2 },
\mathfrak{S} \in \R^{d \times d} $
satisfy
for all $ i \in \{ 1, \ldots, d \} $ that
$ Q_{ i, i } = 1 $,
$ \forall \, j \in \{ 1, \ldots, d \} \setminus \{ i \} \colon Q_{ i, j } = 0.1 $,
and
$ \mathfrak{S} \, \mathfrak{S}^* = Q $,
let $ \mathbb{F} = ( \mathbb{F}_t )_{ t \in [ 0, T ] } $
be the filtration generated by $ W^{ 0, 1 } $,
let $ \mathfrak{F} = ( \mathfrak{F}_t )_{ t \in [0,T] } $ be the
filtration generated by $ W^{ 0, 1, (1) } $,
assume that each of the
employed neural networks
has $ H $ hidden layers,
and
assume
for all
$ m, j \in \N $,
$ n \in \{ 0, 1, 2 \} $,
$ s \in [ 0, T ] $,
$ x = ( x_1, \ldots, x_d ) \in \R^d $
that
$ T = 1 $,
$ N = 2 $,
$ M = 500 $,
$ \mathcal{X}_{ n }^{ m-1, j } = \mathfrak{S} \, W_{ t_n }^{ m-1, j } $,
$ J_0 = $ 4,096,000, 
$ J_m = 8192$,
$ \varepsilon = 10^{-8} $,
$ \gamma_m =
5 \, [ 10^{ -2 } \, \mathbbm{1}_{ [ 1, 100 ] }(m)
+ 10^{ -3 } \, \mathbbm{1}_{ ( 100, 300 ] }(m)
+ 10^{ -4 } \, \mathbbm{1}_{ ( 300, \infty ) }(m) ] $,
and
\begin{equation}
g(s,x)
=
e^{ - r s }
\max\biggl\{ 
K -
\exp\biggl(
  \bigl[ r - \tfrac{1}{2}\beta^2 \bigr] s
  +
  \tfrac{ \beta \sqrt{10} }{ \sqrt{ d ( d + 9 ) } }
  [ x_1 + \ldots + x_d ]
\biggr) 
\chi
, 0
\biggr\}
.
\end{equation}
The random variable $ \mathcal{P} $ given in \eqref{apprprice} provides approximations of the real number
\begin{equation}
\label{eq:exact_BM_Bermudan}
\sup\!\left\{ 
\E\bigl[
  g( \tau, \mathfrak{S} \, W_{ \tau }^{ 0, 1 } )
\bigr]
\colon
\substack{
  \tau \colon \Omega \to \{ t_0, t_1, t_2 \} \text{ is an}
\\
  ( \mathbb{F}_t )_{ t \in \{ t_0, t_1, t_2 \} }\text{-stopping time}
}
\right\}
.
\end{equation}
The numbers in Table~\ref{tab:ex_BM_Bermudan} were obtained with our standard 
network architecture with two hidden layers. It shows approximations of the mean of $ \mathcal{P} $,
of the standard deviation of $ \mathcal{P} $,
and of the relative $ L^1 $-approximation error associated with $ \mathcal{P} $,
the uncorrected sample standard deviation of the relative approximation error associated with $ \mathcal{P} $,
and the average runtime in seconds needed for calculating one realisation of $ \mathcal{P} $
for $ d \in \{ 1, 5, 10, 50, 100, 500, 1000 \} $. 
For each case, the calculations of the results in Tables~\ref{tab:ex_BM_Bermudan}--\ref{tab:ex_BM_Bermudan3}
are based on $10$ independent realisations of $ \mathcal{P} $,
which were obtained from an implementation in {\sc Python}.
Furthermore, in the approximative calculations of the relative approximation error associated with $ \mathcal{P} $,
the exact number~\eqref{eq:exact_BM_Bermudan}
was replaced, independently of the dimension $ d $, by the real number
\begin{equation}
\label{eq:BM_Bermudan}
\sup\biggl\{ 
\E\Bigl[ 
  e^{ - r \tau }
  \max\Bigl\{ 
  K -
  \exp\Bigl(
    \bigl[ r - \tfrac{1}{2}\beta^2 \bigr] \tau
    +
    \beta \, W_\tau^{0,1,(1)}
  \Bigr) 
  \chi
  , 0
  \Bigr\}
\Bigr]
\colon
\substack{
  \tau \colon \Omega \to \{ t_0, t_1, t_2 \} \text{ is an}
\\
  ( \mathfrak{F}_t )_{ t \in \{ t_0, t_1, t_2 \} }\text{-stopping time}
}
\biggr\}
\end{equation}
(cf.\ Corollary~\ref{cor:id_BM}),
which, in turn, was replaced by the value $ 7.894 $.
The latter was computed in {\sc Matlab} R2017b
using the binomial tree method implemented as {\sc Matlab}'s function \texttt{optstockbycrr}
with 20,000 nodes.
Note that
\eqref{eq:BM_Bermudan}
corresponds to the price of a Bermudan put option
on a single stock in the Black--Scholes model
with initial stock price $ \chi $,
interest rate $ r $,
volatility $ \beta $,
strike price $ K $,
maturity $ T $,
and $ N $ possible exercise dates.

\begin{table}[!b]
\centering
\begin{tabular}{|c|c|c|c|c|c|}
\hline
Dimen-&Mean&Standard&Rel.\ $L^1$-&Standard&Average\\
sion $d$&of $ \mathcal{P} $&deviation&approx.&deviation&runtime \\
&&of $ \mathcal{P} $&error&of the rel.&in sec.\ for \\
&&&&approx. &one realisa-\\
& & & & error &tion of $ \mathcal{P} $\\
\hline
1 & 7.896 & 0.005 & 0.0005 & 0.0003 & 2.9 \\
5 & 7.893 & 0.007 & 0.0007 & 0.0005 & 2.8 \\
10 & 7.895 & 0.004 & 0.0005 & 0.0003 & 2.7 \\
50 & 7.889 & 0.005 & 0.0007 & 0.0005 & 2.8 \\
100 & 7.890 & 0.002 & 0.0005 & 0.0002 & 3.0 \\
500 & 7.894 & 0.005 & 0.0004 & 0.0003 & 7.8 \\
1000 & 7.892 & 0.005 & 0.0006 & 0.0004 & 17.2 \\
\hline
\end{tabular}
\caption{\label{tab:ex_BM_Bermudan}
Numerical simulations of the algorithm of 
Framework~\ref{algo:general} for optimally stopping
a correlated Brownian motion in the case of the Bermudan put option with 
three exercise opportunities of Subsection~\ref{sec:example_BM_Bermudan} with $ H = 2 $.
In the approximative calculations of the relative approximation errors,
the exact number~\eqref{eq:exact_BM_Bermudan}
was replaced by the value $ 7.894 $,
which was approximatively computed in {\sc Matlab}.}
\end{table}

Due to the underlying one-dimensional structure, all examples in Subsection~\ref{sec:low_dimensional}
admit an optimal stopping rule which, at each possible exercise date $t_n$, checks whether the 
current pay-off is above a threshold $c_n \in \R$. Therefore, we also apply our algorithm
to the example in Subsection~\ref{sec:example_BM_Bermudan} using networks with one input neuron and no hidden layers.
This corresponds to learning the thresholds $c_n$ from simulated pay-offs with one-dimensional 
logistic regressions. We used the same number of simulations as in Table~\ref{tab:ex_BM_Bermudan}
and batch normalisation before the first linear transformation but no batch normalisation before the 
logistic function. As can be seen from Table~\ref{tab:ex_BM_Bermudan2}, 
the resulting approximations have the same accuracy as the ones of 
Table~\ref{tab:ex_BM_Bermudan} and, in addition, the computations times are shorter.
However, as we will see in Subsection~\ref{sec:high_dimensional}, it cannot be hoped 
that good results can be obtained with a simplified network architecture
if the stopping problem is more complex.

\begin{table}[!t]
\centering
\begin{tabular}{|c|c|c|c|c|c|}
\hline
Dimen-&Mean&Standard&Rel.\ $L^1$-&Standard&Average\\
sion $d$&of $ \mathcal{P} $&deviation&approx.&deviation&runtime \\
&&of $ \mathcal{P} $&error&of the rel.&in sec.\ for \\
&&&&approx. &one realisa-\\
& & & & error &tion of $ \mathcal{P} $\\
\hline
1 & 7.893 & 0.007 & 0.0006 & 0.0007 & 1.6 \\
5 & 7.896 & 0.005 & 0.0005 & 0.0003 & 1.6 \\ 
10 & 7.895 & 0.005 & 0.0005 & 0.0003 & 1.6 \\ 
50 & 7.892 & 0.004 & 0.0005 & 0.0003 & 1.7 \\ 
100 & 7.890 & 0.004 & 0.0005 & 0.0004 & 1.8 \\ 
500 & 7.893 & 0.005 & 0.0005 & 0.0005 & 3.0 \\ 
1000 & 7.894 & 0.003 & 0.0003 & 0.0002 & 5.7 \\ 
\hline
\end{tabular}
\caption{\label{tab:ex_BM_Bermudan2}
Numerical approximations of the Bermudan put option with three
exercise opportunities of Subsection~\ref{sec:example_BM_Bermudan}
based on logistic regressions using only the pay-off as input ($ H = 0 $).}
\end{table}

Table \ref{tab:ex_BM_Bermudan3} shows approximations of \eqref{eq:exact_BM_Bermudan} 
for $d =10$ obtained with networks with $(d+1)$-dimensional input layers and different numbers of hidden layers.
Again, it can be seen that in this example hidden layers do not improve the accuracy of the results.

\begin{table}[!t]
\centering
\begin{tabular}{|c|c|c|c|c|c|c|}
\hline
Number&Mean&Standard&Rel.\ $L^1$-&Standard&Average\\
of hidden&of $ \mathcal{P} $&deviation&approx.&deviation&runtime \\
layers $ H $&&of $ \mathcal{P} $&error&of the rel.&in sec.\ for\\
&&&&approx.&one realisa-\\
& & & &  error &tion of $ \mathcal{P} $\\
\hline
0 & 7.894 & 0.005 & 0.0005 & 0.0002 & 1.9 \\ 
1 & 7.892 & 0.005 & 0.0006 & 0.0004 & 1.9 \\ 
2 & 7.892 & 0.006 & 0.0006 & 0.0004 & 1.8 \\ 
3 & 7.891 & 0.004 & 0.0005 & 0.0003 & 1.7 \\ 
4 & 7.888 & 0.009 & 0.0007 & 0.0011 & 1.7 \\ 
5 & 7.889 & 0.006 & 0.0008 & 0.0007 & 1.8 \\ 
\hline
\end{tabular}
\caption{\label{tab:ex_BM_Bermudan3}
Approximations of the price of the three-exercise-opportunities put option 
of Subsection~\ref{sec:example_BM_Bermudan} for $d =10$ with networks 
with $(d+1)$-dimensional input layers and $H \in \{0,1,2,3,4,5\}$ hidden layers.}
\end{table}

\paragraph{An American put-type example}
\label{sec:example_BM_American}

In this subsection, we test the algorithm of Framework~\ref{algo:general}
on the example of optimally stopping
a standard Brownian motion
under a put option inspired pay-off function.

Assume Framework~\ref{algo:examples-setting},
let
$ r = 0.06 = 6 \% $,
$ \beta = 0.4 = 40 \% $,
$ \chi = K  = 40 $,
let $ \mathbb{F} = ( \mathbb{F}_t )_{ t \in [0,T] } $ be the
filtration generated by $ W^{ 0, 1 } $,
let $ \mathfrak{F} = ( \mathfrak{F}_t )_{ t \in [0,T] } $ be the
filtration generated by $ W^{ 0, 1, (1) } $,
and
assume
for all
$ m, j \in \N $,
$ n \in \{ 0, 1, \ldots, N \} $,
$ s \in [ 0, T ] $,
$ x = ( x_1, \ldots, x_d ) \in \R^d $
that
$ T = 1 $,
$ N = 50 $,
$ M =
1500 \, \mathbbm{1}_{ [ 1, 50 ] }(d) 
+ 1800 \, \mathbbm{1}_{ ( 50, 100 ] }(d)
+ 3000 \, \mathbbm{1}_{ ( 100, \infty ) }(d)
$,
$ \mathcal{X}_{ n }^{ m-1, j } = W_{ t_n }^{ m-1, j } $,
$ J_0 =$  4,096,000,
$ J_m =
8192 \, \mathbbm{1}_{ [ 1, 50 ] }(d) 
+ 4096 \, \mathbbm{1}_{ ( 50, 100 ] }(d)
+ 2048 \, \mathbbm{1}_{ ( 100, \infty ) }(d)
$,
$ \varepsilon = 0.001 $,
$ \gamma_m =
5 \, [ 10^{ -2 } \, \mathbbm{1}_{ [ 1, \nicefrac{M}{3} ] }(m)
+ 10^{ -3 } \, \mathbbm{1}_{ ( \nicefrac{M}{3}, \nicefrac{2M}{3} ] }(m)
+ 10^{ -4 } \, \mathbbm{1}_{ ( \nicefrac{2M}{3}, \infty ) }(m) ] $,
and
\begin{equation}
g(s,x) = e^{ - r s }
\max\Bigl\{ 
K -
\exp\Bigl(
  \bigl[ r - \tfrac{1}{2}\beta^2 \bigr] s
  +
  \tfrac{ \beta }{ \sqrt{ d } }
  [ x_1 + \ldots + x_d ]
\Bigr) 
\chi
, 0
\Bigr\}.
\end{equation}
The random variable $ \mathcal{P} $ given in \eqref{apprprice}
provides approximations of the real number
\begin{equation}
\label{eq:exact_BM_American}
\sup\!\left\{ 
\E\bigl[ 
  g( \tau, W_{ \tau }^{ 0, 1 } )
\bigr]
\colon
\substack{
  \tau \colon \Omega \to [0,T] \text{ is an}
\\
  \mathbb{F}\text{-stopping time}
}
\right\}
.
\end{equation}
We report approximations of the mean of $ \mathcal{P} $,
of the standard deviation of $ \mathcal{P} $,
and of the relative $ L^1 $-approximation error associated with $ \mathcal{P} $,
the uncorrected sample standard deviation of the relative approximation error associated with $ \mathcal{P} $,
and
the average runtime in seconds needed for calculating one realisation of $ \mathcal{P} $
for $ d \in \{ 1, 5, 10, 50, 100, 500, 1000 \} $
in Table~\ref{tab:ex_BM_American}.
For each case, the calculations of the results in Table~\ref{tab:ex_BM_American}
are based on $10$ independent realisations of $ \mathcal{P} $,
which were obtained from an implementation in {\sc Python}.
Furthermore, in the approximative calculations of the relative approximation error associated with $ \mathcal{P} $,
the exact number~\eqref{eq:exact_BM_American}
was replaced, independently of the dimension $ d $, by the real number
\begin{equation}
\label{eq:American_BM_American}
\sup\biggl\{ 
\E\Bigl[ 
  e^{ - r \tau }
  \max\Bigl\{ 
  K -
  \exp\Bigl(
    \bigl[ r - \tfrac{1}{2}\beta^2 \bigr] \tau
    +
    \beta \, W_\tau^{0,1,(1)}
  \Bigr) 
  \chi
  , 0
  \Bigr\}
\Bigr]
\colon
\substack{
  \tau \colon \Omega \to [0,T] \text{ is an}
\\
  \mathfrak{F}\text{-stopping time}
}
\biggr\}
\end{equation}
(cf.\ Corollary~\ref{cor:id_BM}),
which, in turn, was replaced by the value $ 5.318 $
(cf.\ Longstaff \& Schwartz~\cite[Table~1 in Section~3]{LongstaffSchwartz2001}).
This value was calculated using the binomial tree method on Smirnov's website~\cite{SmirnovWebsite}
with 20,000 nodes.
Note that
\eqref{eq:American_BM_American}
corresponds to the price of an American put option
on a single stock in the Black--Scholes model
with initial stock price $ \chi $,
interest rate $ r $,
volatility $ \beta $,
strike price $ K $,
and maturity $ T $. 
Moreover, note that in this example we have
that for all $ m \in \N $ the batch size is
$ J_m =
8192 \, \mathbbm{1}_{ [ 1, 50 ] }(d) 
+ 4096 \, \mathbbm{1}_{ ( 50, 100 ] }(d)
+ 2048 \, \mathbbm{1}_{ ( 100, \infty ) }(d)
$. In particular, for all $ m \in \N $ the batch size
$ J_m $ is decreasing as $ d $ increases. Such kind of fine-tuning is 
necessary to fit the examples into the memory available on the GPU
here and below.

\begin{table}[!hb]
\centering
\begin{tabular}{|c|c|c|c|c|c|}
\hline
Dimen-&Mean&Standard&Rel.\ $L^1$-&Standard&Average\\
sion $d$&of $ \mathcal{P} $&deviation&approx.&deviation&runtime \\
&&of $ \mathcal{P} $&error&of the rel.&in sec.\ for \\
&&&&approx. &one realisa-\\
& & & & error &tion of $ \mathcal{P} $\\
\hline
1 & 5.311 & 0.004 & 0.0014 & 0.0007 & 37.9 \\
5 & 5.310 & 0.003 & 0.0015 & 0.0006 & 43.2 \\
10 & 5.310 & 0.003 & 0.0016 & 0.0006 & 44.4 \\
50 & 5.307 & 0.003 & 0.0020 & 0.0005 & 75.0 \\
100 & 5.305 & 0.004 & 0.0024 & 0.0007 & 73.8 \\
500 & 5.299 & 0.004 & 0.0035 & 0.0007 & 347.7 \\
1000 & 5.296 & 0.005 & 0.0042 & 0.0009 & 696.7 \\
\hline
\end{tabular}
\caption{\label{tab:ex_BM_American}%
Numerical simulations of the algorithm of 
Framework~\ref{algo:general} for
optimally stopping
a standard Brownian motion
in the case of the American put-type example in Subsection~\ref{sec:example_BM_American}.
In the approximative calculations of the relative approximation errors,
the exact number~\eqref{eq:exact_BM_American}
was replaced by the value $ 5.318 $,
which was obtained using
the binomial tree method on Smirnov's website~\cite{SmirnovWebsite}.}
\end{table}

\subsubsection{Geometric average-type options}
\label{sec:geometric_average}

\paragraph{An American geometric average put-type example}
\label{sec:example_GA_put}

In this subsection, we test the algorithm of Framework~\ref{algo:general}
on the example of pricing
an American geometric average put-type option
on up to 200 distinguishable stocks in the Black--Scholes model.
Moreover,
we compare its performance with that of the algorithm introduced in~\cite{BeckerCheriditoJentzen2019}.

Assume Framework~\ref{algo:examples-setting},
assume that $ d \in \{ 40, 80, 120, \ldots \} $,
let
$ \beta = ( \beta_1, \ldots, \beta_d ) \in \R^d $,
$ \rho, \tilde{\delta}, \tilde{\beta}, \delta_1, \delta_2, \ldots, \delta_d \in \R $,
$ r = 0.6 $,
$ K = 95 $,
$ \tilde{\xi} = 100 $
satisfy
for all $ i \in \{ 1, \ldots, d \} $
that
$ \beta_i = \min \{ 0.04 \, [ ( i - 1 ) \bmod 40 ], 1.6 - 0.04 \, [ ( i - 1 ) \bmod 40 ] \} $,
$ \rho = \frac{1}{d} \| \beta \|_{ \R^d }^2 = \frac{1}{40} \sum_{i=1}^{40} ( \beta_i )^2 = 0.2136 $,
$ \delta_i = r - \frac{\rho}{d}\bigl( i - \frac{1}{2} \bigr) - \frac{1}{5\sqrt{d}}$,
$ \tilde{\delta} = r - \frac{1}{\sqrt{d}} \sum_{i=1}^{d} ( r - \delta_i ) + \frac{\sqrt{d}-1}{2 d} \| \beta \|_{ \R^d }^2
= r - \frac{\rho}{2} - \frac{1}{5} = 0.2932 $,
and
$ \tilde{\beta}
=
\tfrac{1}{\sqrt{d}} \| \beta \|_{ \R^d }
= \sqrt{\rho} $,
let $ Y \colon [0,T] \times \Omega \to \R $
be an $ \mathscr{F} $-adapted stochastic process with continuous sample paths
which satisfies
that for all $ t \in [ 0, T ] $ it holds $ \P $-a.s.\ that
\begin{equation}
Y_t
=
\tilde{\xi}
+
( r - \tilde{\delta} )
\int_{0}^{t} Y_s \, ds
+
\tilde{\beta}
\int_{0}^{t} Y_s \, dW_s^{0,1,(1)}
,
\end{equation}
let $ \mathbb{F} = ( \mathbb{F}_t )_{ t \in [0,T] } $ be the
filtration generated by $ X $,
let $ \mathfrak{F} = ( \mathfrak{F}_t )_{ t \in [0,T] } $ be the
filtration generated by $ Y $,
and assume for all
$ m, j \in \N $,
$ n \in \{ 0, 1, \ldots, N \} $,
$ i \in \{ 1, \ldots, d \} $,
$ s \in [ 0, T ] $,
$ x = ( x_1, \ldots, x_d ) \in \R^d $
that
$ T =  1 $,
$ N = 100 $,
$ M =
1800 \, \mathbbm{1}_{ [ 1, 120 ] }(d)
+ 3000 \, \mathbbm{1}_{ ( 120, \infty ) }(d)
$,
$ J_0 =$ 4,096,000,
$ J_m =
8192 \, \mathbbm{1}_{ [ 1, 120 ] }(d) 
+ 4096 \, \mathbbm{1}_{ ( 120, \infty ) }(d)
$,
$ \varepsilon = 10^{-8} $,
$ \gamma_m =
5 \, [ 10^{ -2 } \, \mathbbm{1}_{ [ 1, \nicefrac{M}{3} ] }(m)
+ 10^{ -3 } \, \mathbbm{1}_{ ( \nicefrac{M}{3}, \nicefrac{2M}{3} ] }(m)
+ 10^{ -4 } \, \mathbbm{1}_{ ( \nicefrac{2M}{3}, \infty ) }(m) ] $,
$ \xi_i = (100)^{ \nicefrac{1}{ \sqrt{d} } } $,
$ \mu(x) = ( ( r - \delta_1 ) \, x_1, \ldots, ( r - \delta_d ) \, x_d ) $,
$ \sigma(x) = \operatorname{diag}( \beta_1 x_1, \ldots, \beta_d x_d ) $,
that
\begin{equation}
\mathcal{X}_{ n }^{ m-1, j, (i) }
=
\exp\bigl(
\bigl[ r - \delta_i - \tfrac{1}{2}( \beta_i )^2 \bigr] t_n
+
\beta_i \, W_{ t_n }^{ m-1, j, (i) }
\bigr)
\, \xi_i,
\end{equation}
and that
\begin{equation}
g(s,x)
=
e^{ - r s }
\max \Biggl\{ K - \Biggl[ \prod_{k=1}^{d} | x_k |^{ \nicefrac{1}{\sqrt{d}} } \Biggr], 0 \Biggr\}
.
\end{equation}
The random variable $ \mathcal{P} $ given in \eqref{apprprice}
provides approximations of the price
\begin{equation}
\label{eq:exact_GA_put}
\sup\!\left\{ 
\E\bigl[ 
  g( \tau, X_{ \tau } )
\bigr]
\colon
\substack{
  \tau \colon \Omega \to [0,T] \text{ is an}
\\
  \mathbb{F}\text{-stopping time}
}
\right\}
.
\end{equation}
In Table~\ref{tab:example_GA_put},
we show
approximations
of the mean of $ \mathcal{P} $,
of the standard deviation of $ \mathcal{P} $,
and of the relative $ L^1 $-approximation error associated with $ \mathcal{P} $,
the uncorrected sample standard deviation of the relative approximation error associated with $ \mathcal{P} $,
and
the average runtime in seconds needed for calculating one realisation of $ \mathcal{P} $
for $ d \in \{ 40, 80, 120, 160, 200 \} $.
For each case, the calculations of the results in Table~\ref{tab:example_GA_put}
are based on $10$ independent realisations of $ \mathcal{P} $,
which were obtained from an implementation in {\sc Python}.
Furthermore, in the approximative calculations of the relative approximation error associated with $ \mathcal{P} $,
the exact value of the price~\eqref{eq:exact_GA_put}
was replaced, independently of the dimension $ d $, by the real number
\begin{equation}
\label{eq:American_GA_put}
\sup\!\left\{ 
\E\bigl[ 
  e^{ - r \tau }
  \max \{ K - Y_\tau, 0 \}
\bigr]
\colon
\substack{
  \tau \colon \Omega \to [0,T] \text{ is an}
\\
  \mathfrak{F}\text{-stopping time}
}
\right\},
\end{equation}
(cf.\ Proposition~\ref{prop:dim_reduction}),
which, in turn, was replaced by the value $ 6.545 $.
The latter was calculated using the binomial tree method on Smirnov's website~\cite{SmirnovWebsite}
with 20,000 nodes. Note that
\eqref{eq:American_GA_put}
corresponds to the price of an American put option
on a single stock in the Black--Scholes model
with initial stock price $ \tilde{\xi} $,
interest rate $ r $,
dividend yield $ \tilde{\delta} $,
volatility $ \tilde{\beta} $,
strike price $ K $,
and maturity $ T $.

\begin{table}[!ht]
\centering
\begin{tabular}{|c|c|c|c|c|c|}
\hline
Dimen-&Mean&Standard&Rel.\ $L^1$-&Standard&Average\\
sion $d$&of $ \mathcal{P} $&deviation&approx.&deviation&runtime \\
&&of $ \mathcal{P} $&error&of the rel.&in sec.\ for \\
&&&&approx. &one realisa-\\
& & & & error &tion of $ \mathcal{P} $\\
\hline
40 & 6.512 & 0.004 & 0.0051 & 0.0006 & 166.4 \\
80 & 6.509 & 0.003 & 0.0056 & 0.0005 & 263.8 \\
120 & 6.507 & 0.003 & 0.0058 & 0.0004 & 350.9 \\
160 & 6.504 & 0.003 & 0.0062 & 0.0005 & 358.5 \\
200 & 6.501 & 0.006 & 0.0067 & 0.0009 & 432.2 \\
\hline
\end{tabular}
\caption{\label{tab:example_GA_put}%
Numerical simulations of the algorithm of Framework~\ref{algo:general}
for pricing the American geometric average put-type option
from the example in Subsection~\ref{sec:example_GA_put}.
In the approximative calculations of the relative approximation errors,
the exact value of the price~\eqref{eq:exact_GA_put}
was replaced by the value $ 6.545 $,
which was obtained using
the binomial tree method on Smirnov's website~\cite{SmirnovWebsite}.}
\end{table}

In addition, we computed approximations of the price~\eqref{eq:exact_GA_put}
using the algorithm introduced in~\cite{BeckerCheriditoJentzen2019},
where the random variable $ \hat{L} $ from \cite[Subsection~3.1]{BeckerCheriditoJentzen2019}
plays the role analogous to $ \mathcal{P} $.
In order to maximise comparability,
the hyperparameters and neural network architectures
employed for the algorithm from~\cite{BeckerCheriditoJentzen2019}
were chosen to be identical to the corresponding ones
used for computing realisations of $ \mathcal{P} $.
Table~\ref{tab:example_GA_put_other_method}
shows
approximations
of the mean of $ \hat{L} $ (cf.~\cite[Subsection~3.1]{BeckerCheriditoJentzen2019}),
of the standard deviation of $ \hat{L} $,
and of the relative $ L^1 $-approximation error associated with $ \hat{L} $,
the uncorrected sample standard deviation of the relative approximation error associated with $ \hat{L} $,
and
the average runtime in seconds needed for calculating one realisation of $ \hat{L} $
for $ d \in \{ 40, 80, 120, 160, 200 \} $.
For each case, the calculations of the results in Table~\ref{tab:example_GA_put_other_method}
are based on $10$ independent realisations of $ \hat{L} $,
which were obtained from an implementation in {\sc Python}.
In the approximative calculations of the relative approximation error associated with $ \hat{L} $,
the exact value of the price~\eqref{eq:exact_GA_put}
again was replaced, independently of the dimension $ d $, by the value $ 6.545 $.
Comparing Table~\ref{tab:example_GA_put} with Table~\ref{tab:example_GA_put_other_method},
we note that in the present cases
the algorithm of Framework~\ref{algo:general}
and
the algorithm from~\cite{BeckerCheriditoJentzen2019}
exhibit very similar performance in terms of both accuracy and speed,
with a slight runtime advantage for the algorithm of Framework~\ref{algo:general}.

\begin{table}[!t]
\centering
\begin{tabular}{|c|c|c|c|c|c|}
\hline
Dimen-&Mean&Standard&Rel.\ $L^1$-&Standard&Average\\
sion $d$&of $ \hat{L} $&deviation&approx.&deviation&runtime \\
&&of $ \hat{L} $&error&of the rel.&in sec.\ for \\
&&&&approx. &one realisa-\\
& & & & error &tion of $ \hat{L} $\\
\hline
40 & 6.505 & 0.003 & 0.0061 & 0.0005 & 167.8 \\
80 & 6.503 & 0.003 & 0.0065 & 0.0005 & 264.9 \\
120 & 6.502 & 0.003 & 0.0066 & 0.0004 & 381.1 \\
160 & 6.505 & 0.006 & 0.0061 & 0.0008 & 359.7 \\
200 & 6.506 & 0.005 & 0.0060 & 0.0008 & 433.2 \\
\hline
\end{tabular}
\caption{\label{tab:example_GA_put_other_method}%
Numerical simulations of the algorithm from~\cite{BeckerCheriditoJentzen2019}
for pricing the American geometric average put-type option
from the example in Subsection~\ref{sec:example_GA_put}.
In the approximative calculations of the relative approximation errors,
the exact value of the price~\eqref{eq:exact_GA_put}
was again replaced by the value $ 6.545 $.}
\end{table}

\paragraph{An American geometric average call-type example}
\label{sec:example_GA_call2}

In this subsection, we test the algorithm of Framework~\ref{algo:general}
on the example of pricing
an American geometric average call-type option
on up to 200 correlated stocks in the Black--Scholes model.
This example is taken from Sirignano \& Spiliopoulos~\cite[Subsection~4.3]{SirignanoSpiliopoulos2018}.

Assume Framework~\ref{algo:examples-setting},
let
$ r = 0 \% $,
$ \delta = 0.02 = 2 \% $,
$ \beta = 0.25 = 25 \% $,
$ K = \tilde{\xi} = 1 $,
$ Q = ( Q_{ i, j } )_{ ( i, j ) \in \{ 1, \ldots, d \}^2 },
\mathfrak{S} = ( \varsigma_1, \ldots, \varsigma_d ) \in \R^{d \times d} $,
$ \tilde{\delta}, \tilde{\beta} \in \R $
satisfy
for all $ i \in \{ 1, \ldots, d \} $ that
$ Q_{ i, i } = 1 $,
$ \forall \, j \in \{ 1, \ldots, d \} \setminus \{ i \} \colon Q_{ i, j } = 0.75 $,
$ \mathfrak{S}^* \, \mathfrak{S} = Q $,
$ \tilde{\delta} = \delta + \frac{1}{2}( \beta^2 - ( \tilde{\beta} )^2 ) $,
and
$ \tilde{\beta}
=
\tfrac{\beta}{2 d} \sqrt{d \, ( 3 d + 1 ) }
$,
let $ Y \colon [0,T] \times \Omega \to \R $
be an $ \mathscr{F} $-adapted stochastic process with continuous sample paths
which satisfies
that for all $ t \in [ 0, T ] $
it holds $ \P $-a.s.\ that
\begin{equation}
Y_t
=
\tilde{\xi}
+
( r - \tilde{\delta} )
\int_{0}^{t} Y_s \, ds
+
\tilde{\beta}
\int_{0}^{t} Y_s \, dW_s^{0,1,(1)}
,
\end{equation}
let $ \mathbb{F} = ( \mathbb{F}_t )_{ t \in [0,T] } $ be the
filtration generated by $ X $,
let $ \mathfrak{F} = ( \mathfrak{F}_t )_{ t \in [0,T] } $ be the
filtration generated by $ Y $,
and assume for all
$ m, j \in \N $,
$ n \in \{ 0, 1, \ldots, N \} $,
$ i \in \{ 1, \ldots, d \} $,
$ s \in [ 0, T ] $,
$ x = ( x_1, \ldots, x_d ) \in \R^d $
that
$ T =  2 $,
$ N = 50 $,
$ M = 1600$,
$ J_0 =$ 4,096,000,
$ J_m = 8192 $,
$ \varepsilon = 10^{-8} $,
$ \gamma_m =
5 \, [ 10^{ -2 } \, \mathbbm{1}_{ [ 1, 400 ] }(m)
+ 10^{ -3 } \, \mathbbm{1}_{ ( 400 , 800 ] }(m)
\allowbreak
+ 10^{ -4 } \, \mathbbm{1}_{ ( 800, \infty ) }(m) ] $,
$ \xi_i = 1 $,
$ \mu( x ) = ( r - \delta ) \, x $,
$ \sigma(x) = \beta \operatorname{diag}( x_1, \dots, x_d ) \, \mathfrak{S}^* $,
that
\begin{equation}
\mathcal{X}_{ n }^{ m-1, j, (i) }
=
\exp\Bigl(
\bigl[ r - \delta - \tfrac{1}{2} \beta^2  \bigr] t_n
+
\beta \, \bigl\langle \varsigma_i, W_{ t_n }^{ m-1, j } \bigr\rangle_{ \R^d }
\Bigr)
\, \xi_i,
\end{equation}
and that
\begin{equation}
g(s,x)
=
e^{ - r s }
\max \Biggl\{ \Biggl[ \prod_{k=1}^{d} | x_k |^{ \nicefrac{1}{d} } \Biggr] - K, 0 \Biggr\}
.
\end{equation}
The random variable $ \mathcal{P} $ given in \eqref{apprprice} provides approximations of the price
\begin{equation}
\label{eq:exact_GA_call2}
\sup\!\left\{ 
\E\bigl[ 
  g( \tau, X_{ \tau } )
\bigr]
\colon
\substack{
  \tau \colon \Omega \to [0,T] \text{ is an}
\\
  \mathbb{F}\text{-stopping time}
}
\right\}
.
\end{equation}
We report
approximations
of the mean of $ \mathcal{P} $,
of the standard deviation of $ \mathcal{P} $,
of the real number
\begin{equation}
\label{eq:american_GA_call2}
\sup\!\left\{ 
\E\bigl[ 
  e^{ - r \tau }
  \max \{ Y_\tau - K, 0 \}
\bigr]
\colon
\substack{
  \tau \colon \Omega \to [0,T] \text{ is an}
\\
  \mathfrak{F}\text{-stopping time}
}
\right\},
\end{equation}
and of the relative $ L^1 $-approximation error associated with $ \mathcal{P} $,
the uncorrected sample standard deviation of the relative approximation error associated with $ \mathcal{P} $,
and the average runtime in seconds needed for calculating one realisation of $ \mathcal{P} $
for $ d \in \{ 3, 20, 100, 200 \} $ in Table~\ref{tab:example_GA_call2}.
The approximative calculations of the mean of $ \mathcal{P} $, of the standard deviation of $ \mathcal{P} $,
and of the relative $ L^1 $-approximation error associated with $ \mathcal{P} $,
the computations of
the uncorrected sample standard deviation of the relative approximation error associated with $ \mathcal{P} $
as well as the computations of
the average runtime for calculating one realisation of $ \mathcal{P} $
in Table~\ref{tab:example_GA_call2}
each are based on $10$ independent realisations of $ \mathcal{P} $,
which were obtained from an implementation in {\sc Python}.
Furthermore, in the approximative calculations of the relative approximation error associated with $ \mathcal{P} $,
the exact value of the price~\eqref{eq:exact_GA_call2}
was replaced by the number~\eqref{eq:american_GA_call2}
(cf.\ Proposition~\ref{prop:dim_reduction}),
which was approximatively calculated using the binomial tree method on Smirnov's website~\cite{SmirnovWebsite}
with 20,000 nodes. Note that~\eqref{eq:american_GA_call2}
corresponds to the price of an American call option
on a single stock in the Black--Scholes model
with initial stock price $ \tilde{\xi} $, 
interest rate $ r $, 
dividend yield $ \tilde{\delta} $, 
volatility $ \tilde{\beta} $, 
strike price $ K $, 
and maturity $ T $. 

\begin{table}
\centering
\begin{tabular}{|c|c|c|c|c|c|c|}
\hline
Dimen-&Mean&Standard&Price&Rel.\ $L^1$-&Standard&Average\\
sion $d$&of $ \mathcal{P} $&deviation&\eqref{eq:american_GA_call2}&approx.&deviation&runtime \\
&&of $ \mathcal{P} $&&error&of the rel.&in sec.\ for\\
&&&&&approx.&one realisa-\\
& & & & & error &tion of $ \mathcal{P} $\\
\hline
3 & 0.10698 & 0.00006 & 0.10719 & 0.0020 & 0.0005 & 41.7 \\
20 & 0.10006 & 0.00006 & 0.10033 & 0.0027 & 0.0006 & 57.0 \\
100 & 0.09905 & 0.00006 & 0.09935 & 0.0030 & 0.0006 & 137.2 \\
200 & 0.09891 & 0.00009 & 0.09923 & 0.0032 & 0.0009 & 243.0 \\
\hline
\end{tabular}
\caption{\label{tab:example_GA_call2}%
Numerical simulations of the algorithm of Framework~\ref{algo:general}
for pricing the American geometric average call-type option
from the example in Subsection~\ref{sec:example_GA_call2}.
In the approximative calculations of the relative approximation errors,
the exact value of the price~\eqref{eq:exact_GA_call2}
was replaced by the number~\eqref{eq:american_GA_call2},
which was approximatively calculated using
the binomial tree method on Smirnov's website~\cite{SmirnovWebsite}.}
\end{table}

\paragraph{Another American geometric average call-type example}
\label{sec:example_GA_call1}

In this subsection, we test the algorithm of Framework~\ref{algo:general}
on the example of pricing
an American geometric average call-type option
on up to 400 distinguishable stocks in the Black--Scholes model.

Assume Framework~\ref{algo:examples-setting},
assume that $ d \in \{ 40, 80, 120, \ldots \} $,
let
$ \beta = ( \beta_1, \ldots, \beta_d ) \in \R^d $,
$ \alpha_1, \ldots, \alpha_d \in \R $,
$ r, \tilde{\beta} \in ( 0, \infty ) $,
$ K = 95 $,
$ \tilde{\xi} = 100 $
satisfy
for all $ i \in \{ 1, \ldots, d \} $
that
$ \beta_i = \frac{0.4 \, i}{d} $,
$ \alpha_i = \min \{ 0.01 \, [ ( i - 1 ) \bmod 40 ], 0.4 - 0.01 \, [ ( i - 1 ) \bmod 40 ] \} $,
$ r = \frac{1}{d} \sum_{i=1}^{d} \alpha_i - \frac{d-1}{2 d^2} \| \beta \|_{ \R^d }^2
= 0.1 - \frac{0.08}{d^2} ( d - 1 ) \bigl( \frac{d}{3} + \frac{1}{2} + \frac{1}{6 d} \bigr)
$,
and
$ \tilde{\beta}
=
\tfrac{1}{d} \| \beta \|_{ \R^d }
= \frac{0.4}{d} \bigl( \frac{d}{3} + \frac{1}{2} + \frac{1}{6 d} \bigr)^{ \nicefrac{1}{2} }
$,
let $ Y \colon [0,T] \times \Omega \to \R $
be an $ \mathscr{F} $-adapted stochastic process with continuous sample paths
which satisfies
that for all $ t \in [ 0, T ] $
it holds $ \P $-a.s.\ that
\begin{equation}
Y_t
=
\tilde{\xi}
+
r
\int_{0}^{t} Y_s \, ds
+
\tilde{\beta}
\int_{0}^{t} Y_s \, dW_s^{0,1,(1)}
,
\end{equation}
let $ \mathbb{F} = ( \mathbb{F}_t )_{ t \in [0,T] } $ be the
filtration generated by $ X $,
let $ \mathfrak{F} = ( \mathfrak{F}_t )_{ t \in [0,T] } $ be the
filtration generated by $ Y $,
and assume for all
$ m, j \in \N $,
$ n \in \{ 0, 1, \ldots, N \} $,
$ i \in \{ 1, \ldots, d \} $,
$ s \in [ 0, T ] $,
$ x = ( x_1, \ldots, x_d ) \in \R^d $
that
$ T =  3 $,
$ N = 50 $,
$ M = 1500 $,
$ J_0 =$ 4,096,000,
$ J_m = 8192 $,
$ \varepsilon = 10^{-8} $,
$ \gamma_m =
5 \, [ 10^{ -2 } \, \mathbbm{1}_{ [ 1, \nicefrac{M}{3} ] }(m)
+ 10^{ -3 } \, \mathbbm{1}_{ ( \nicefrac{M}{3}, \nicefrac{2M}{3} ] }(m)
+ 10^{ -4 } \, \mathbbm{1}_{ ( \nicefrac{2M}{3}, \infty ) }(m) ] $,
$ \xi_i = 100 $,
$ \mu(x) = ( \alpha_1 x_1, \ldots, \alpha_d x_d ) $,
$ \sigma(x) = \operatorname{diag}( \beta_1 x_1, \ldots, \beta_d x_d ) $,
that
\begin{equation}
\mathcal{X}_{ n }^{ m-1, j, (i) }
=
\exp\bigl(
\bigl[ \alpha_i - \tfrac{1}{2}( \beta_i )^2 \bigr] t_n
+
\beta_i \, W_{ t_n }^{ m-1, j, (i) }
\bigr)
\, \xi_i,
\end{equation}
and that
\begin{equation}
g(s,x)
=
e^{ - r s }
\max \Biggl\{ \Biggl[ \prod_{k=1}^{d} | x_k |^{ \nicefrac{1}{d} } \Biggr] - K, 0 \Biggr\}
.
\end{equation}
The random variable $ \mathcal{P} $ given in \eqref{apprprice}
provides approximations of the price
\begin{equation}
\label{eq:exact_GA_call1}
\sup\!\left\{ 
\E\bigl[ 
  g( \tau, X_{ \tau } )
\bigr]
\colon
\substack{
  \tau \colon \Omega \to [0,T] \text{ is an}
\\
  \mathbb{F}\text{-stopping time}
}
\right\}
.
\end{equation}
In Table~\ref{tab:example_GA_call1}, we show
approximations
of the mean of $ \mathcal{P} $,
of the standard deviation of $ \mathcal{P} $,
of the real number
\begin{equation}
\label{eq:european_GA_call1}
\E \bigl[
e^{ -r T }
\max \{ Y_T - K, 0 \}
\bigr],
\end{equation}
and of the relative $ L^1 $-approximation error associated with $ \mathcal{P} $,
the uncorrected sample standard deviation of the relative approximation error associated with $ \mathcal{P} $,
and
the average runtime in seconds needed for calculating one realisation of $ \mathcal{P} $
for $ d \in \{ 40, 80, 120, 160, 200, \allowbreak 400 \} $.
The approximative calculations of the mean of $ \mathcal{P} $, of the standard deviation of $ \mathcal{P} $,
and of the relative $ L^1 $-approximation error associated with $ \mathcal{P} $,
the computations of
the uncorrected sample standard deviation of the relative approximation error associated with $ \mathcal{P} $
as well as the computations of
the average runtime for calculating one realisation of $ \mathcal{P} $
in Table~\ref{tab:example_GA_call1}
each are based on $10$ independent realisations of $ \mathcal{P} $,
which were obtained from an implementation in {\sc Python}.
Moreover, in the approximative calculations of the relative approximation error associated with $ \mathcal{P} $,
the exact value of the price~\eqref{eq:exact_GA_call1}
was replaced by the real number
\begin{equation}
\label{eq:american_GA_call1}
  \sup\!\left\{ 
    \E\bigl[ 
      e^{ - r \tau }
      \max \{ Y_\tau - K, 0 \}
    \bigr]
    \colon
    \substack{
      \tau \colon \Omega \to [0,T] \text{ is an}
    \\
      \mathfrak{F}\text{-stopping time}
    }
  \right\}
\end{equation}
(cf.\ Proposition~\ref{prop:dim_reduction}).
It is well known (cf., e.g., Shreve~\cite[Corollary~8.5.3]{Shreve2004}) that
the number~\eqref{eq:american_GA_call1} is equal to
the number~\eqref{eq:european_GA_call1},
which was approximatively computed in {\sc Matlab} R2017b
using Lemma~\ref{lem:black_scholes} above.
Note that
\eqref{eq:american_GA_call1}
corresponds to the price of an American call option
on a single stock in the Black--Scholes model
with initial stock price $ \tilde{\xi} $,
interest rate $ r $,
volatility $ \tilde{\beta} $,
strike price $ K $,
and maturity $ T $,
while
\eqref{eq:european_GA_call1}
corresponds to the price of a European call option
on a single stock in the Black--Scholes model
with initial stock price $ \tilde{\xi} $,
interest rate $ r $,
volatility $ \tilde{\beta} $,
strike price $ K $,
and maturity $ T $.

\begin{table}
\centering
\begin{tabular}{|c|c|c|c|c|c|c|}
\hline
Dimen-&Mean&Standard&Price&Rel.\ $L^1$-&Standard&Average\\
sion $d$&of $ \mathcal{P} $&deviation&\eqref{eq:european_GA_call1}&approx.&deviation&runtime \\
&&of $ \mathcal{P} $&&error&of the rel.&in sec.\ for\\
&&&&&approx.&one realisa-\\
& & & & & error &tion of $ \mathcal{P} $\\
\hline
40 & 23.6878 & 0.0037 & 23.6883 & 0.00012 & 0.00009 & 70.8 \\
80 & 23.7221 & 0.0019 & 23.7235 & 0.00008 & 0.00006 & 112.5 \\
120 & 23.7361 & 0.0028 & 23.7357 & 0.00011 & 0.00004 & 147.9 \\
160 & 23.7414 & 0.0017 & 23.7419 & 0.00006 & 0.00004 & 188.3 \\
200 & 23.7454 & 0.0017 & 23.7456 & 0.00005 & 0.00004 & 224.5 \\
400 & 23.7523 & 0.0010 & 23.7531 & 0.00004 & 0.00003 & 1088.3 \\
\hline
\end{tabular}
\caption{\label{tab:example_GA_call1}%
Numerical simulations of the algorithm of Framework~\ref{algo:general}
for pricing the American geometric average call-type option
from the example in Subsection~\ref{sec:example_GA_call1}.
In the approximative calculations of the relative approximation errors,
the exact value of the price~\eqref{eq:exact_GA_call1}
was replaced by the number~\eqref{eq:european_GA_call1},
which was approximatively computed in {\sc Matlab}.}
\end{table}

\subsection{Examples without known one-dimensional representation}
\label{sec:high_dimensional}

In Subsection~\ref{sec:low_dimensional} above,
numerical results for examples with a one-dimensional representation can be found.
We test in this subsection different examples where such a representation is not known.

\subsubsection{Max-call options}
\label{sec:max_call}

\paragraph{A Bermudan max-call standard benchmark example}
\label{sec:max_std}

In this subsection, we test the algorithm of Framework~\ref{algo:general}
on the example of pricing
a Bermudan max-call option
on up to 500 stocks in the Black--Scholes model
(cf.~\cite[Subsection~4.1]{BeckerCheriditoJentzen2019}).
In the case of up to five underlying stocks, this example is a standard benchmark example in the literature
(cf., e.g.,
\cite[Subsection~5.4]{BroadieGlasserman2004},
\cite[Subsection~8.1]{LongstaffSchwartz2001},
\cite[Section~4]{AndersenBroadie2004},
\cite[Subsection~5.1]{HaughKogan2004},
\cite[Subsection~4.3]{Rogers2002},
\cite[Subsection~3.9]{Firth2005},
\cite[Subsection~4.2]{BenderKolodkoSchoenmakers2006},
\cite[Subsection~5.3]{BC08},
\cite[Subsection~6.1]{BelomestnyBendnerSchoenmakers2009},
\cite[Subsection~4.1]{JainOosterlee2012},
\cite[Subsection~7.2]{SchoenmakersZhangHuang2013},
\cite[Subsection~6.1]{BelomestnyLadkauSchoenmakers2015},
\cite[Subsection~5.2.1]{Lelong2018}).

Assume Framework~\ref{algo:examples-setting},
let
$ H \in \N_0 $,
$ r = 0.05 = 5 \% $,
$ \delta = 0.1 = 10\% $,
$ \beta = 0.2 = 20 \% $,
$ K = 100 $,
let $ \mathbb{F} = ( \mathbb{F}_t )_{ t \in [0,T] } $ be the
filtration generated by $ X $,
assume that each of the
employed neural networks
has $ H $ hidden layers,
and assume for all
$ m, j \in \N $,
$ n \in \{ 0, 1, \ldots, N \} $,
$ i \in \{ 1, \ldots, d \} $,
$ s \in [ 0, T ] $,
$ x = ( x_1, \ldots, x_d ) \in \R^d $
that
$ T = 3 $,
$ N = 9 $,
$ M = 3000 + d $,
$ J_0 =$ 4,096,000 
$ J_m = 8192 $,
$ \varepsilon = 0.1 $,
$ \gamma_m =
5 \, [ 10^{ -2 } \, \mathbbm{1}_{ [1,500 + \nicefrac{d}{5} ] }(m)
+ 10^{ -3 } \, \mathbbm{1}_{ ( 500 + \nicefrac{d}{5}, 1500 + \nicefrac{3d}{5} ] }(m)
+ 10^{ -4 } \, \mathbbm{1}_{ ( 1500 + \nicefrac{3d}{5}, \infty ) }(m) ] $,
$ \xi_i = \xi_1 $,
$ \mu( x ) = ( r - \delta ) \, x $,
$ \sigma(x) = \beta \operatorname{diag}( x_1, \dots, x_d ) $,
that
\begin{equation}
\mathcal{X}_{ n }^{ m-1, j, (i) }
=
\exp\bigl(
\bigl[ r - \delta - \tfrac{1}{2} \beta^2 \bigr] t_n
+
\beta \, W_{ t_n }^{ m-1, j, (i) }
\bigr)
\, \xi_1,
\end{equation}
and that
\begin{equation}
g(s,x)
=
e^{ - r s }
\max \bigl\{ \max\{ x_1, \ldots, x_d \} - K, 0 \bigr\}
.
\end{equation}
The random variable $ \mathcal{P} $ given in \eqref{apprprice}
provides approximations of the price
\begin{equation}
\label{eq:exact_max_std}
\sup\!\left\{ 
\E\bigl[ 
  g( \tau, X_{ \tau } )
\bigr]
\colon
\substack{
  \tau \colon \Omega \to \{ t_0, t_1, \ldots, t_N \} \text{ is an}
\\
  (\mathbb{F}_t )_{ t \in \{ t_0, t_1, \ldots, t_N \} }\text{-stopping time}
}
\right\}
.
\end{equation}
In Table~\ref{tab:ex_max_std1},
we show approximations
of the mean and of the standard deviation of $ \mathcal{P} $,
binomial approximations as well as
95\%~confidence intervals for the price~\eqref{eq:exact_max_std}
according to Andersen \& Broadie~\cite[Table~2 in Section~4]{AndersenBroadie2004} (where available),
95\%~confidence intervals for the price~\eqref{eq:exact_max_std}
according to Broadie \& Cao~\cite[Table~3 in Subsection~5.3]{BC08} (where available),
and the average runtime in seconds needed for calculating one realisation of $ \mathcal{P} $
for
$ ( d, \xi_1 ) \in \{ 2, 3, 5 \} \times \{ 90, 100, 110 \} $ and $ H = 2 $. The approximative calculations of 
the mean and of the standard deviation of $ \mathcal{P} $
as well as the computations of the average runtime for calculating one realisation of $ \mathcal{P} $
in Tables~\ref{tab:ex_max_std1}--\ref{tab:ex_max_layers}
each are based on $10$ independent realisations of $ \mathcal{P} $,
which were obtained from an implementation in {\sc Python}
(cf.\ {\sc Python} code~\ref{code:max_std} in Subsection~\ref{sec:code_max_std} below).

In Table~\ref{tab:ex_max_std2}, we list approximations
of the mean and of the standard deviation of $ \mathcal{P} $
and the average runtime in seconds needed for calculating one realisation of $ \mathcal{P} $
for
$ ( d, \xi_1 ) \in \{ 10, 20, 30, 50, 100, 200, 500 \} \times \{ 90, 100, 110 \} $ and $ H = 2 $.

\begin{table}
\centering
\small
\begin{tabular}{|c|c|c|c|c|c|c|c|}
\hline
Di-&Ini-&Mean&Standard&Bino-&95\% confidence&95\% confidence&Average\\
men-&tial&of $ \mathcal{P} $&deviation&mial&interval in~\cite{AndersenBroadie2004}&interval in~\cite{BC08}& runtime\\
sion&value&&of $ \mathcal{P} $&value&&& in sec.\ for \\
$d$&$\xi_1$&&&in~\cite{AndersenBroadie2004}&&& one realisa-\\
& & & & & & & tion of $ \mathcal{P} $ \\
\hline
2 & 90 & 8.068 & 0.006 & 8.075 & [8.053, 8.082] & -- & 20.9 \\
2 & 100 & 13.901 & 0.010 & 13.902 & [13.892, 13.934] & -- & 21.0 \\
2 & 110 & 21.341 & 0.008 & 21.345 & [21.316, 21.359] & -- & 21.0 \\
\hline
3 & 90 & 11.278 & 0.010 & 11.29 & [11.265, 11.308] & -- & 21.2 \\
3 & 100 & 18.685 & 0.006 & 18.69 & [18.661, 18.728] & -- & 21.3 \\
3 & 110 & 27.560 & 0.007 & 27.58 & [27.512, 27.663] & -- & 21.3 \\
\hline
5 & 90 & 16.631 & 0.009 & -- & [16.602, 16.655] & [16.620, 16.653] & 21.6 \\
5 & 100 & 26.147 & 0.009 & -- & [26.109, 26.292] & [26.115, 26.164] & 21.7 \\
5 & 110 & 36.774 & 0.009 & -- & [36.704, 36.832] & [36.710, 36.798] & 21.7 \\
\hline
\end{tabular}
\caption{\label{tab:ex_max_std1}%
Numerical simulations of the algorithm of Framework~\ref{algo:general}
for pricing the Bermudan max-call option
from the example in Subsection~\ref{sec:max_std}
for $ d \in  \{ 2, 3, 5 \} $ and $ H = 2 $
(cf.\ {\sc Python} code~\ref{code:max_std} in Subsection~\ref{sec:code_max_std} below).}
\end{table}

\begin{table}
\centering
\begin{tabular}{|c|c|c|c|c|}
\hline
Dimen-&Initial&Mean&Standard&Average runtime \\
sion $ d $ &value $\xi_1$&of $ \mathcal{P} $&deviation&in sec.\ for one\\
&&&of $ \mathcal{P} $&realisation of $ \mathcal{P} $\\
\hline
10 & 90 & 26.196 & 0.012 & 22.7 \\
10 & 100 & 38.272 & 0.008 & 22.9 \\
10 & 110 & 50.812 & 0.007 & 22.9 \\
\hline
20 & 90 & 37.692 & 0.011 & 26.9 \\
20 & 100 & 51.572 & 0.010 & 26.9 \\
20 & 110 & 65.510 & 0.010 & 27.5 \\
\hline
30 & 90 & 44.831 & 0.014 & 30.3 \\
30 & 100 & 59.516 & 0.010 & 30.0 \\
30 & 110 & 74.234 & 0.009 & 30.0 \\
\hline
50 & 90 & 53.897 & 0.016 & 36.8 \\
50 & 100 & 69.572 & 0.008 & 36.8 \\
50 & 110 & 85.262 & 0.011 & 36.8 \\
\hline
100 & 90 & 66.359 & 0.010 & 55.6 \\
100 & 100 & 83.390 & 0.014 & 55.6 \\
100 & 110 & 100.421 & 0.014 & 55.6 \\
\hline
200 & 90 & 79.005 & 0.009 & 94.3 \\
200 & 100 & 97.414 & 0.012 & 94.4 \\
200 & 110 & 115.832 & 0.012 & 94.4 \\
\hline
500 & 90 & 95.970 & 0.012 & 265.0 \\
500 & 100 & 116.254 & 0.010 & 261.4 \\
500 & 110 & 136.534 & 0.015 & 264.9 \\
\hline
\end{tabular}
\caption{\label{tab:ex_max_std2}%
Numerical simulations of the algorithm of Framework~\ref{algo:general}
for pricing the Bermudan max-call option
from the example in Subsection~\ref{sec:max_std}
for $ d \in  \{ 10, 20, 30, 50, \allowbreak 100, 200, 500 \} $ and $ H = 2 $
(cf.\ {\sc Python} code~\ref{code:max_std} in Subsection~\ref{sec:code_max_std} below).}
\end{table}

To see the impact of the number of hidden layers used in the neural networks, we additionally report in 
Table~\ref{tab:ex_max_layers} approximation results for $d=5$, $\xi_1 = 100$, and $H \in \{0,1,2,3,4,5\}$.
We used the same number of simulations as in Tables ~\ref{tab:ex_max_std1}--\ref{tab:ex_max_std2}, 
but due to the higher number of hidden layers, we chose
$ M = 5000 $ and
$\forall \, m \in \N \colon \gamma_m = 5[10^{-2} \mathbbm{1}_{[1,1000]}(m) + 10^{-3} \mathbbm{1}_{(1000, 3000]}(m)
+ 10^{-4} \mathbbm{1}_{(3000, \infty)}(m)]$. It can be seen that, in this example, two hidden layers
yield better results than zero or one hidden layer. But more than two hidden layers do not 
lead to an improvement.

\begin{table}[ht]
\centering
\begin{tabular}{|c|c|c|c|}
\hline
Number&Mean&Standard&Average runtime \\
of hidden&of $ \mathcal{P} $&deviation&in sec.\ for one\\
layers $ H $&&of $ \mathcal{P} $&realisation of $ \mathcal{P} $\\
\hline
0 & 25.843 & 0.010 & 11.4 \\ 
1 & 26.134 & 0.009 & 21.7 \\
2 & 26.147 & 0.007 & 37.7 \\ 
3 & 26.147 & 0.012 & 54.6 \\ 
4 & 26.146 & 0.006 & 70.2 \\ 
5 & 26.144 & 0.008 & 84.7 \\ 
\hline
\end{tabular}
\caption{\label{tab:ex_max_layers}%
Numerical results for the Bermudan max-call option
from the example in Subsection~\ref{sec:max_std}
for $ d = 5$ and $\xi_1 = 100$ obtained using networks with 
$H \in \{0,1,2,3,4,5\}$ hidden layers.}
\end{table}

\paragraph{A high-dimensional Bermudan max-call benchmark example}
\label{sec:max_big}

In this subsection, we test the algorithm of Framework~\ref{algo:general}
on the example of pricing
the Bermudan max-call option from the example in Subsection~\ref{sec:max_std}
in a case with 5000 underlying stocks.
All {\sc Python} source codes corresponding to this example
were run
in single precision (float32)
on a NVIDIA Tesla P100 GPU.

Assume Framework~\ref{algo:examples-setting},
let
$ r = 0.05 = 5 \% $,
$ \delta = 0.1 = 10\% $,
$ \beta = 0.2 = 20 \% $,
$ K = 100 $,
let $ \mathbb{F} = ( \mathbb{F}_t )_{ t \in [0,T] } $ be the
filtration generated by $ X $,
and assume for all
$ m, j \in \N $,
$ n \in \{ 0, 1, \ldots, N \} $,
$ i \in \{ 1, \ldots, d \} $,
$ s \in [ 0, T ] $,
$ x = ( x_1, \ldots, x_d ) \in \R^d $
that
$ T = 3 $,
$ d = 5000 $,
$ N = 9 $,
$ J_0 = 2^{20} $,
$ J_m = 1024 $,
$ \varepsilon = 10^{-8} $,
$ \gamma_m =
10^{ -2 } \, \mathbbm{1}_{ [ 1, 2000 ] }(m)
+ 10^{ -3 } \, \mathbbm{1}_{ ( 2000, 4000 ] }(m)
+ 10^{ -4 } \, \mathbbm{1}_{ ( 4000, \infty ) }(m) $,
$ \xi_i = 100 $,
$ \mu( x ) = ( r - \delta ) \, x $,
$ \sigma(x) = \beta \operatorname{diag}( x_1, \dots, x_d ) $,
that
\begin{equation}
\mathcal{X}_{ n }^{ m-1, j, (i) }
=
\exp\bigl(
\bigl[ r - \delta - \tfrac{1}{2} \beta^2 \bigr] t_n
+
\beta \, W_{ t_n }^{ m-1, j, (i) }
\bigr)
\, \xi_i,
\end{equation}
and that
\begin{equation}
g(s,x)
=
e^{ - r s }
\max \bigl\{ \max\{ x_1, \ldots, x_d \} - K, 0 \bigr\}
.
\end{equation}
For sufficiently large $ M \in \N $,
the random variable $ \mathcal{P} $
provides approximations of the price
\begin{equation}
\label{eq:exact_max_big}
\sup\!\left\{ 
\E\bigl[ 
  g( \tau, X_{ \tau } )
\bigr]
\colon
\substack{
  \tau \colon \Omega \to \{ t_0, t_1, \ldots, t_N \} \text{ is an}
\\
  (\mathbb{F}_t )_{ t \in \{ t_0, t_1, \ldots, t_N \} }\text{-stopping time}
}
\right\}
.
\end{equation}
In Table~\ref{tab:ex_max_big}, we show a realisation of $ \mathcal{P} $,
a 95\% confidence interval for the corresponding realisation of the random variable
\begin{equation}
\label{eq:RV_confidence}
\Omega \ni \mathbf{w}
\mapsto
\E \Bigl[
 g\Bigl( 
      \tau^{ 1, \Theta_M( \mathbf{w} ), \mathbb{S}_M( \mathbf{w} ) }
      ,
      \mathcal{X}^{ 0, 1 }_{ 
        \tau^{ 1, \Theta_M( \mathbf{w} ), \mathbb{S}_M( \mathbf{w} ) }
      }
    \Bigr)
\Bigr]
\in \R,
\end{equation}
the corresponding realisation of the relative approximation error associated with $ \mathcal{P} $,
and the runtime in seconds needed for calculating the realisation of $ \mathcal{P} $
for
$ M \in \{ 0, 250, 500, \ldots, \allowbreak 2000 \} \cup \{ 6000 \} $.
In addition,
Figure~\ref{fig:ex_max_big} depicts a realisation of the relative approximation error associated 
with $ \mathcal{P} $ against $ M \in \{ 0, 10, 20, \ldots, 2000 \} $.
For each case, the 95\% confidence interval for the realisation of the random variable~\eqref{eq:RV_confidence}
in Table~\ref{tab:ex_max_big}
was computed based on the corresponding realisation of $ \mathcal{P} $,
the corresponding sample standard deviation,
and the $ 0.975 $ quantile of the standard normal distribution
(cf., e.g., \cite[Subsection~3.3]{BeckerCheriditoJentzen2019}).
Moreover,
in the approximative calculations of the realisation of the relative approximation error associated with $ \mathcal{P} $,
in Table~\ref{tab:ex_max_big} and Figure~\ref{fig:ex_max_big}
the exact value of the price~\eqref{eq:exact_max_big}
was replaced by the value 165.430,
which corresponds to a realisation of $ \mathcal{P} $ with $ M = 6000 $
(cf.\ Table~\ref{tab:ex_max_big}).
All realisations of $ \mathcal{P} $ have been obtained from running a close variant of
the {\sc Python} code used for the example in Subsection~\ref{sec:max_std} with $ d = 5000 $
(cf.\ {\sc Python} code~\ref{code:max_std} in Subsection~\ref{sec:code_max_std} below).

\begin{table}[t]
\centering
\begin{tabular}{|c|c|c|c|c|}
\hline
Number of   & Realisation            & 95\% confidence & Rel.\ approx.& Runtime\\
steps $ M $ & of $ \mathcal{P} $ & interval               & error & in sec.\\
\hline
0 & 106.711 & [106.681, 106.741] & 0.35495 & 157.3 \\
250 & 132.261 & [132.170, 132.353] & 0.20050 & 271.7 \\
500 & 156.038 & [155.975, 156.101] & 0.05677 & 386.0 \\
750 & 103.764 & [103.648, 103.879] & 0.37276 & 500.4 \\
1000 & 161.128 & [161.065, 161.191] & 0.02601 & 614.3 \\
1250 & 162.756 & [162.696, 162.816] & 0.01616 & 728.8 \\
1500 & 164.498 & [164.444, 164.552] & 0.00563 & 842.8 \\
1750 & 163.858 & [163.803, 163.913] & 0.00950 & 957.3 \\
2000 & 165.452 & [165.400, 165.505] & 0.00014 & 1071.9 \\
\hline
6000 & 165.430 & [165.378, 165.483] & 0.00000 & 2899.5 \\
\hline
\end{tabular}
\caption{\label{tab:ex_max_big}%
Numerical simulations of the algorithm of Framework~\ref{algo:general}
for pricing the Bermudan max-call option on 5000 stocks
from the example in Subsection~\ref{sec:max_big}
(cf.\ {\sc Python} code~\ref{code:max_std} in Subsection~\ref{sec:code_max_std} below).
In the approximative calculations of the relative approximation error,
the exact value of the price~\eqref{eq:exact_max_big}
was replaced by the value 165.430,
which corresponds to a realisation of $ \mathcal{P} $ with $ M = 6000 $.}
\end{table}

\begin{figure}[!ht]
\centering
\includegraphics[width=0.75\textwidth]{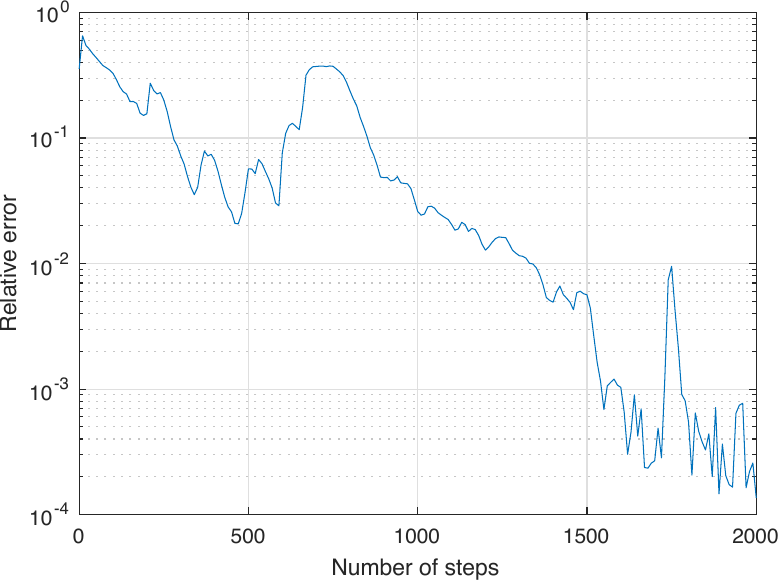}
\caption{\label{fig:ex_max_big}%
Plot of a realisation of the relative approximation error
$ \tfrac{|\mathcal{P}-165.430|}{165.430} $
against $ M \in \{ 0, 10, 20, \ldots, 2000 \} $
in the case of the Bermudan max-call option on 5000 stocks
from the example in Subsection~\ref{sec:max_big}.}
\end{figure}

\paragraph{Another Bermudan max-call example}
\label{sec:max-call_equity}

In this subsection, we test the algorithm of Framework~\ref{algo:general}
on the example of pricing
a Bermudan max-call option
for different maturities and strike prices
on up to 400 \emph{correlated} stocks, that do not pay dividends, in the Black--Scholes model.
This example is taken from Barraquand \& Martineau~\cite[Section~VII]{BarraquandMartineau1995}.

Assume Framework~\ref{algo:examples-setting},
let
$ \eta = \nicefrac{30}{365} $,
$ r = 0.05 \, \eta = 5 \% \cdot \eta $,
$ \beta = 0.4 \sqrt{\eta} = 40 \% \cdot \sqrt{\eta} $,
$ K \in \{ 35, 40, 45 \} $,
$ Q = ( Q_{ i, j } )_{ ( i, j ) \in \{ 1, \ldots, d \}^2 },
\mathfrak{S} = ( \varsigma_1, \ldots, \varsigma_d ) \in \R^{d \times d} $
satisfy for all $ i \in \{ 1, \ldots, d \} $ that
$ Q_{ i, i } = 1 $,
$ \forall \, j \in \{ 1, \ldots, d \} \setminus \{ i \} \colon Q_{ i, j } = 0.5 $,
and
$ \mathfrak{S}^* \, \mathfrak{S} = Q $,
let $ \mathbb{F} = ( \mathbb{F}_t )_{ t \in [0,T] } $ be the
filtration generated by $ X $,
and assume for all
$ m, j \in \N $,
$ n \in \{ 0, 1, \ldots, N \} $,
$ i \in \{ 1, \ldots, d \} $,
$ s \in [ 0, T ] $,
$ x = ( x_1, \ldots, x_d ) \in \R^d $
that
$ N = 10 $,
$ M = 1600 $,
$ J_0 =$ 4,096,000,
$ J_m = 8192 $,
$ \varepsilon = 0.001 $,
$ \gamma_m =
5 \, [ 10^{ -2 } \, \mathbbm{1}_{ [ 1, 400 ] }(m)
+ 10^{ -3 } \, \mathbbm{1}_{ ( 400, 800 ] }(m)
+ 10^{ -4 } \, \mathbbm{1}_{ ( 800, \infty ) }(m) ] $,
$ \xi_i = 40 $,
$ \mu( x ) = r \, x $,
$ \sigma(x) = \beta \operatorname{diag}( x_1, \dots, x_d ) \, \mathfrak{S}^* $,
that
\begin{equation}
\mathcal{X}_{ n }^{ m-1, j, (i) }
=
\exp\Bigl(
\bigl[ r - \tfrac{1}{2} \beta^2  \bigr] t_n
+
\beta \, \bigl\langle \varsigma_i, W_{ t_n }^{ m-1, j } \bigr\rangle_{ \R^d }
\Bigr)
\, \xi_i,
\end{equation}
and that
\begin{equation}
g(s,x)
=
e^{ - r s }
\max \bigl\{ \max\{ x_1, \ldots, x_d \} - K, 0 \bigr\}
.
\end{equation}
The random variable $ \mathcal{P} $ given in \eqref{apprprice}
provides approximations of the price
\begin{equation}
\label{eq:exact_max-call_equity}
\sup\!\left\{ 
\E\bigl[ 
  g( \tau, X_{ \tau } )
\bigr]
\colon
\substack{
  \tau \colon \Omega \to \{ t_0, t_1, \ldots, t_N \} \text{ is an}
\\
  (\mathbb{F}_t )_{ t \in \{ t_0, t_1, \ldots, t_N \} }\text{-stopping time}
}
\right\}
.
\end{equation}
In Table~\ref{tab:ex_max},
we show
approximations
of the mean and of the standard deviation of $ \mathcal{P} $,
Monte Carlo approximations of the
European max-call option price
\begin{equation}
\label{eq:European_price}
 \E\bigl[ g( T, X_T ) \bigr]
\end{equation}
corresponding to \eqref{eq:exact_max-call_equity},
approximations of the price~\eqref{eq:exact_max-call_equity}
according to \cite[Table~4 in Section~VII]{BarraquandMartineau1995}
(where available),
and
the average runtime in seconds needed for calculating one realisation of $ \mathcal{P} $
for
\begin{equation}
\begin{split}
( d, T, K ) \in
\left\{
\begin{array}{cccc}
( 10, 1, 35 ),
&( 10, 1, 40 ),
&( 10, 1, 45 ),
\\( 10, 4, 35 ),
&( 10, 4, 40 ),
&( 10, 4, 45 ),
\\( 10, 7, 35 ),
&( 10, 7, 40 ),
&( 10, 7, 45 ),
\\( 400, 12, 35 ),
&( 400, 12, 40 ),
&( 400, 12, 45 )
\end{array}
\right\}.
\end{split}
\end{equation}%
\begin{table}[!b]
\centering
\begin{tabular}{|c|c|c|c|c|c|c|c|}
\hline
Dimen-&Matu-&Strike&Mean&Standard&European&Price&Average runtime \\
sion $d$&rity $T$&price $K$&of $ \mathcal{P} $&deviation&price~\eqref{eq:European_price}&in~\cite{BarraquandMartineau1995}&in sec.\ for one \\
&&&&of $ \mathcal{P} $&&&realisation of $ \mathcal{P} $\\
\hline
10 & 1 & 35 & 10.364 & 0.001 & 10.365 & 10.36 & 15.8 \\
10 & 1 & 40 & 5.540 & 0.002 & 5.540 & 5.54 & 15.8 \\
10 & 1 & 45 & 1.895 & 0.001 & 1.896 & 1.90 & 15.8 \\
\hline
10 & 4 & 35 & 16.518 & 0.004 & 16.520 & 16.53 & 15.7 \\
10 & 4 & 40 & 11.867 & 0.004 & 11.870 & 11.87 & 15.6 \\
10 & 4 & 45 & 7.799 & 0.005 & 7.804 & 7.81 & 15.6 \\
\hline
10 & 7 & 35 & 20.914 & 0.006 & 20.916 & 20.92 & 15.6 \\
10 & 7 & 40 & 16.373 & 0.009 & 16.374 & 16.38 & 15.5 \\
10 & 7 & 45 & 12.270 & 0.006 & 12.277 & 12.28 & 15.5 \\
\hline
400 & 12 & 35 & 55.714 & 0.012 & 55.714 & -- & 126.7 \\
400 & 12 & 40 & 50.967 & 0.012 & 50.964 & -- & 127.6 \\
400 & 12 & 45 & 46.235 & 0.012 & 46.234 & -- & 126.4 \\
\hline
\end{tabular}
\caption{\label{tab:ex_max}%
Numerical simulations of the algorithm of Framework~\ref{algo:general}
for pricing the Bermudan max-call option from the example in Subsection~\ref{sec:max-call_equity}
(cf.\ {\sc Python} code~\ref{code:max-call_equity} in Subsection~\ref{sec:code_max-call_equity} below).}
\end{table}%
The approximative calculations of the mean and of the standard deviation of $ \mathcal{P} $
as well as the computations of the average runtime for calculating one realisation of $ \mathcal{P} $
in Table~\ref{tab:ex_max}
each are based on $10$ independent realisations of $ \mathcal{P} $,
which were obtained from an implementation in {\sc Python}
(cf.\ {\sc Python} code~\ref{code:max-call_equity} in Subsection~\ref{sec:code_max-call_equity} below).
Furthermore, the Monte Carlo approximations of the European price~\eqref{eq:European_price} in Table~\ref{tab:ex_max}
each were calculated in double precision (float64)
and are based on $ 2 \cdot 10^{10} $ independent realisations of the random variable
$ \Omega \ni \omega \mapsto g( T, X_T( \omega ) ) \in \R $.

\subsubsection{A strangle spread basket option}
\label{sec:strangle_spread}

In this subsection, we test the algorithm of Framework~\ref{algo:general}
on the example of pricing
an American strangle spread basket option
on five correlated stocks in the Black--Scholes model.
This example is taken from Kohler, Krzy\.zak, \& Todorovic~\cite[Section~4]{KohlerKrzyzakTodorovic2010}
(cf.~also Kohler~\cite[Section~3]{Kohler2008}
and Kohler, Krzy\.zak, \& Walk~\cite[Section~4]{KohlerKryzakWalk2008}).

Assume Framework~\ref{algo:examples-setting},
let
$ r = 0.05 = 5 \% $,
$ K_1 = 75 $,
$ K_2 = 90 $,
$ K_3 = 110 $,
$ K_4 = 125 $,
let
$ \mathfrak{S} = ( \varsigma_1, \ldots, \varsigma_5 ) \in \R^{5 \times 5} $
be given by
\begin{equation}
\mathfrak{S} =
\begin{pmatrix}
0.3024 & 0.1354 & 0.0722 & 0.1367 & 0.1641\\
0.1354 & 0.2270 & 0.0613 & 0.1264 & 0.1610\\
0.0722 & 0.0613 & 0.0717 & 0.0884 & 0.0699\\
0.1367 & 0.1264 & 0.0884 & 0.2937 & 0.1394\\
0.1641 & 0.1610 & 0.0699 & 0.1394 & 0.2535
\end{pmatrix},
\end{equation}
let $ \mathbb{F} = ( \mathbb{F}_t )_{ t \in [0,T] } $ be the
filtration generated by $ X $,
and assume for all
$ m, j \in \N $,
$ n \in \{ 0, 1, \ldots, N \} $,
$ i \in \{ 1, \ldots, d \} $,
$ s \in [ 0, T ] $,
$ x = ( x_1, \ldots, x_d ) \in \R^d $
that
$ T =  1 $,
$ d = 5 $,
$ N = 48 $,
$ M = 750 $,
$ J_0 =$ 4,096,000,
$ J_m = 8192 $,
$ \varepsilon = 10^{-8} $,
$ \gamma_m =
5 \, [ 10^{ -2 } \, \mathbbm{1}_{ [ 1, 250 ] }(m)
+ 10^{ -3 } \, \mathbbm{1}_{ ( 250, 500 ] }(m)
+ 10^{ -4 } \, \mathbbm{1}_{ ( 500, \infty ) }(m) ] $,
$ \xi_i = 100 $,
$ \mu( x ) = r \, x $,
$ \sigma(x) = \operatorname{diag}( x_1, \dots, x_d ) \, \mathfrak{S}^* $,
that
\begin{equation}
\mathcal{X}_{ n }^{ m-1, j, (i) }
=
\exp\Bigl(
\bigl[ r - \tfrac{1}{2} \| \varsigma_i \|_{ \R^d }^2  \bigr] t_n
+
\bigl\langle \varsigma_i, W_{ t_n }^{ m-1, j } \bigr\rangle_{ \R^d }
\Bigr)
\, \xi_i,
\end{equation}
and that
\begin{equation}
\begin{split}
g(s,x)
& =
-
e^{ - r s }
\max \Biggl\{
    K_1 - \frac{1}{d} \Biggl[ \sum_{ k = 1 }^{ d } x_k \Biggr],
    0
\Biggr\}
+
e^{ - r s }
\max \Biggl\{
    K_2 - \frac{1}{d} \Bigg[ \sum_{ k = 1 }^{ d } x_k \Biggr],
    0
\Biggr\}
\\ & \quad +
e^{ - r s }
\max \Biggl\{
    \frac{1}{d} \Biggl[ \sum_{ k = 1 }^{ d } x_k \Biggr] - K_3,
    0
\Biggr\}
-
e^{ - r s }
\max \Biggl\{
    \frac{1}{d} \Biggl[ \sum_{ k = 1 }^{ d } x_k \Biggr] - K_4,
    0
\Biggr\}
.
\end{split}
\end{equation}
The random variable $ \mathcal{P} $ given in \eqref{apprprice}
provides approximations of the price
\begin{equation}
\label{eq:exact_strangle_spread}
\sup\!\left\{ 
\E\bigl[ 
  g( \tau, X_{ \tau } )
\bigr]
\colon
\substack{
  \tau \colon \Omega \to [0,T] \text{ is an}
\\
  \mathbb{F}\text{-stopping time}
}
\right\}
.
\end{equation}%
\begin{table}[!b]
\centering
\begin{tabular}{|c|c|c|c|}
\hline
Mean&Standard&Lower&Average runtime \\
of $ \mathcal{P} $&deviation&bound&in sec.\ for one\\
&of $ \mathcal{P} $&in~\cite{KohlerKrzyzakTodorovic2010}&realisation of $ \mathcal{P} $\\
\hline
11.797 & 0.004 & 11.75 & 21.7 \\
\hline
\end{tabular}
\caption{\label{tab:strangle_spread}%
Numerical simulations of the algorithm of Framework~\ref{algo:general}
for pricing the American strangle spread basket option
from the example in Subsection~\ref{sec:strangle_spread}
(cf.\ {\sc Python} code~\ref{code:strangle_spread} in Subsection~\ref{sec:code_strangle_spread} below).}
\end{table}%
Table~\ref{tab:strangle_spread} shows approximations
of the mean and of the standard deviation of $ \mathcal{P} $,
a lower bound for the price~\eqref{eq:exact_strangle_spread}
according to Kohler, Krzy\.zak, \& Todorovic~\cite[Figure~4.5 in Section~4]{KohlerKrzyzakTodorovic2010}
(cf.\ also Kohler~\cite[Figure~2 in Section~3]{Kohler2008}
and, for an upper bound for the price~\eqref{eq:exact_strangle_spread},
Kohler, Krzy\.zak, \& Walk~\cite[Figure~4.2 in Section~4]{KohlerKryzakWalk2008}),
and
the average runtime in seconds needed for calculating one realisation of $ \mathcal{P} $.
Since the mean of $ \mathcal{P} $ is also a lower bound for the price~\eqref{eq:exact_strangle_spread},
a higher value indicates a better approximation of the price~\eqref{eq:exact_strangle_spread}
(cf.\ Table~\ref{tab:strangle_spread}).
The approximative calculations of the mean and of the standard deviation of $ \mathcal{P} $
as well as the computation of the average runtime for calculating one realisation of $ \mathcal{P} $
in Table~\ref{tab:strangle_spread}
each are based on $10$ independent realisations of $ \mathcal{P} $,
which were obtained from an implementation in {\sc Python}
(cf.\ {\sc Python} code~\ref{code:strangle_spread} in Subsection~\ref{sec:code_strangle_spread} below).

\subsubsection{A put basket option in Dupire's local volatility model}
\label{sec:example_Dupire}

In this subsection, we test the algorithm of Framework~\ref{algo:general}
on the example of pricing
an American put basket option
on five stocks in Dupire's local volatility model.
This example is taken from Labart \& Lelong~\cite[Subsection~6.3]{LabartLelong2011arXiv}
with the modification that
we also consider the case where the underlying stocks do not pay any dividends.

Assume Framework~\ref{algo:examples-setting},
let
$ L = 10 $,
$ r = 0.05 = 5 \% $,
$ \delta \in \{ 0 \%, 10 \% \} $,
$ K = 100 $,
assume for all
$ i \in \{ 1, \ldots, d \} $,
$ x \in \R^d $
that
$ \xi_i = 100 $
and
$ \mu( x ) = ( r - \delta ) \, x $,
let $ \beta \colon [ 0, T ] \times \R \to \R $
and $ \boldsymbol{\sigma} \colon [ 0, T ] \times \R^d \to \R^{ d \times d } $
be the functions which satisfy
for all $ t \in [ 0, T ] $, $ x = ( x_1, \ldots, x_d ) \in \R^d $
that
\begin{equation}
\beta( t, x_1 ) =
0.6
\,
e^{ - 0.05 \sqrt{t} }
\bigl(
    1.2
    - e^{
        - 0.1 \, t - 0.001 ( e^{ r t } x_1 - \xi_1 )^2
    }
\bigr) \,
x_1
\end{equation}
and
$ \boldsymbol{\sigma}( t, x ) = \operatorname{diag}( \beta( t, x_1 ), \beta( t, x_2 ), \dots, \beta( t, x_d ) ) $,
let $ S = ( S^{(1)}, \ldots, S^{(d)} ) \colon [0,T] \times \Omega \to \R^d $ 
be an $ \mathscr{F} $-adapted stochastic process with continuous sample paths
which satisfies
that for all $ t \in [ 0, T ] $
it holds $ \P $-a.s.\ that
\begin{equation}
S_t
 = \xi +
\int_{0}^{t} \mu( S_s ) \, ds +
\int_{0}^{t} \boldsymbol{\sigma}( s, S_s ) \, dW_s^{0,1}
,
\end{equation}
let
$
  \mathcal{Y}^{m,j} = ( \mathcal{Y}^{m,j,(1)}, \ldots, \mathcal{Y}^{m,j,(d)} )\colon
  [ 0, T ] \times \Omega \to \R^d
$,
$ j \in \N $,
$ m \in \N_0 $,
be the stochastic processes which satisfy for all
$ m \in \N_0 $,
$ j \in \N $,
$ \ell \in \{ 0, 1, \dots, L - 1 \} $,
$ t \in \bigl[ \frac{\ell T}{L}, \frac{(\ell + 1) T}{L} \bigr] $,
$ i \in \{ 1, \ldots, d \} $
that
$ \mathcal{Y}_0^{m,j,(i)} = \log( \xi_i ) $
and
\begin{equation}
\begin{split}
\mathcal{Y}^{m,j,(i)}_t
& =
\mathcal{Y}^{m,j(i)}_{ \nicefrac{\ell T}{L} }
+
\bigl( t - \tfrac{\ell T}{L} \bigr)
\Bigl( r - \delta - \tfrac{1}{2} \bigl[ \beta\bigl( \tfrac{\ell T}{L}, \exp\bigl( \mathcal{Y}^{m,j,(i)}_{ \nicefrac{\ell T}{L} } \bigr) \bigr) \bigr]^2 \Bigr)
\\ & \quad +
\bigl( \tfrac{t L}{T} - \ell \bigr)
\beta\bigl( \tfrac{\ell T}{L}, \exp\bigl( \mathcal{Y}^{m,j,(i)}_{ \nicefrac{\ell T}{L} } \bigr) \bigr)
\bigl(
W_{ \nicefrac{(\ell + 1) T}{L} }^{m,j,(i)}
-
W_{ \nicefrac{\ell T}{L} }^{m,j,(i)}
\bigr)
,
\end{split}
\end{equation}
let $ \mathbb{F} = ( \mathbb{F}_t )_{ t \in [0,T] } $ be the
filtration generated by $ S $,
let $ \mathfrak{F} = ( \mathfrak{F}_t )_{ t \in [0,T] } $ be the
filtration generated by $ \mathcal{Y}^{0,1} $,
and assume for all
$ m, j \in \N $,
$ n \in \{ 0, 1, \ldots, N \} $,
$ s \in [ 0, T ] $,
$ x = ( x_1, \ldots, x_d ) \in \R^d $
that
$ T =  1 $,
$ d = 5 $,
$ M = 1200 $,
$ \mathcal{X}_{ n }^{m-1,j} = \mathcal{Y}_{ t_n }^{m-1,j} $,
$ J_0 =$ 4,096,000,
$ J_m = 8192 $,
$ \varepsilon = 10^{-8} $,
$ \gamma_m =
5 \, [ 10^{ -2 } \, \mathbbm{1}_{ [ 1, 400 ] }(m)
+ 10^{ -3 } \, \mathbbm{1}_{ ( 400, 800 ] }(m)
+ 10^{ -4 } \, \mathbbm{1}_{ ( 800, \infty ) }(m) ] $,
and 
\begin{equation}
g(s,x)
=
e^{ - r s }
\max \Biggl\{
    K - \frac{1}{d} \Biggl[ \sum_{ i = 1 }^{ d } \exp( x_i ) \Biggr],
    0
\Biggr\}
.
\end{equation}
The random variable $ \mathcal{P} $ given in \eqref{apprprice}
provides approximations of the price
\begin{equation}
\label{eq:example_Dupire}
\sup\!\left\{ 
\E\bigl[ 
  g( \tau, \mathcal{Y}_\tau^{0,1} )
\bigr]
\colon
\substack{
  \tau \colon \Omega \to [0,T] \text{ is an}
\\
  \mathfrak{F}\text{-stopping time}
}
\right\}
,
\end{equation}
which, in turn, is an approximation of the price
\begin{equation}
\sup\!\left\{ 
\E\Biggl[ 
  e^{ - r \tau }
  \max \Biggl\{
      K - \frac{1}{d} \Biggl[ \sum_{ i = 1 }^{ d } S_\tau^{(i)} \Biggr],
      0
  \Biggr\}
\Biggr]
\colon
\substack{
  \tau \colon \Omega \to [0,T] \text{ is an}
\\
  \mathbb{F}\text{-stopping time}
}
\right\}
.
\end{equation}
In Table~\ref{tab:example_Dupire},
we show approximations
of the mean and of the standard deviation of $ \mathcal{P} $
and
the average runtime in seconds needed for calculating one realisation of $ \mathcal{P} $
for $ ( \delta, N ) \in \{ 0 \%, 10 \% \} \times \{ 5, 10, 50, 100 \} $.
For each case, the calculations of the results in Table~\ref{tab:example_Dupire}
are based on $10$ independent realisations of $ \mathcal{P} $,
which were obtained from an implementation in {\sc Python}
(cf.\ {\sc Python} code~\ref{code:example_Dupire} in Subsection~\ref{sec:code_example_Dupire} below).
According to~\cite[Subsection~6.3]{LabartLelong2011arXiv}, the value $ 6.30 $
is an approximation of the price~\eqref{eq:example_Dupire} for $ \delta = 10 \% $.
Furthermore, the European put basket option price 
$ \E\bigl[ g( T, \mathcal{Y}_T^{0,1} ) \bigr] $ corresponding to \eqref{eq:example_Dupire}
was approximatively calculated using
a Monte Carlo approximation based on $ 10^{10} $ realisations of the random variable
$ \Omega \ni \omega \mapsto g( T, \mathcal{Y}_T^{0,1}( \omega ) ) \in \R $
(cf.\ {\sc Python} code~\ref{code:example_Dupire} in Subsection~\ref{sec:code_example_Dupire} below),
which resulted
in the value $ 1.741 $ in the case $ \delta = 0 \% $ and
in the value $ 6.304 $ in the case $ \delta = 10 \% $.

\begin{table}[!t]
\centering
\begin{tabular}{|c|c|c|c|c|}
\hline
Divi-&Time discreti-&Mean&Standard&Average runtime\\
dend&sation para-&of $ \mathcal{P} $&deviation&in sec.\ for one\\
yield $\delta$&meter $ N $&&of $ \mathcal{P} $&realisation of $ \mathcal{P} $\\
\hline
0\% & 5 & 1.934 & 0.001 & 8.7 \\
0\% & 10 & 1.977 & 0.001 & 12.1 \\
0\% & 50 & 1.976 & 0.001 & 34.2 \\
0\% & 100 & 1.972 & 0.002 & 67.0 \\
\hline
10\% & 5 & 6.301 & 0.003 & 8.2 \\
10\% & 10 & 6.303 & 0.004 & 12.0 \\
10\% & 50 & 6.305 & 0.002 & 34.2 \\
10\% & 100 & 6.303 & 0.003 & 67.0 \\
\hline
\end{tabular}
\caption{\label{tab:example_Dupire}%
Numerical simulations of the algorithm of Framework~\ref{algo:general}
for pricing the American put basket option in Dupire's local volatility model 
from the example in Subsection~\ref{sec:example_Dupire}
(cf.\ {\sc Python} code~\ref{code:example_Dupire} in Subsection~\ref{sec:code_example_Dupire} below).
The corresponding European put basket option price is approximately equal to the value $ 1.741 $ in the case $ \delta = 0 \% $ and
to the value $ 6.304 $ in the case $ \delta = 10 \% $.}
\end{table}

\subsubsection{A path-dependent financial derivative}
\label{sec:example_TsitsiklisVanRoy1999}

In this subsection, we test the algorithm of Framework~\ref{algo:general}
on the example of pricing
a specific path-dependent financial derivative
contingent on prices of a single underlying stock in the Black--Scholes model,
which is formulated as a 100-dimensional optimal stopping problem.
This example is taken from Tsitsiklis \& Van Roy~\cite[Section IV]{TsitsiklisVanRoy1999}
with the modification that we consider a finite instead of an infinite time horizon.

Assume Framework~\ref{algo:examples-setting},
let
$ r = 0.0004 = 0.04 \% $,
$ \beta = 0.02 = 2 \% $,
let $ \mathcal{W}^{m,j} \colon [0,\infty) \times \Omega \to \R $,
$ j \in \N $,
$ m \in \N_0 $,
be independent $ \P $-standard Brownian motions
with continuous sample paths,
let $ S^{m,j} \colon [-100,\infty) \times \Omega \to \R $,
$ j \in \N $,
$ m \in \N_0 $,
and
$ \mathcal{Y}^{m,j} \colon \N_0 \times \Omega \to \R^{100} $,
$ j \in \N $,
$ m \in \N_0 $,
be the stochastic processes which satisfy
for all $ m, n \in \N_0 $, $ j \in \N $, $ t \in [ -100, \infty) $ that
$
S_t^{m,j}
=
\exp\bigl(
    \bigl[ r - \tfrac{1}{2}\beta^2 \bigr] ( t + 100 ) + \beta \, \mathcal{W}_{ t + 100 }^{m,j}
\bigr) \, \xi_1
$
and
\begin{equation}
\begin{split}
& \mathcal{Y}_n^{m,j}
=
\Bigl(
\tfrac{S_{ n - 99 }^{m,j}}{S_{ n - 100 }^{m,j}},
\tfrac{S_{ n - 98 }^{m,j}}{S_{ n - 100 }^{m,j}},
\ldots,
\tfrac{S_{ n }^{m,j}}{S_{ n - 100 }^{m,j}}
\Bigr)
\\ &
=
\bigl(
\exp\bigl(
    \bigl[ r - \tfrac{1}{2}\beta^2 \bigr] + \beta \, [ \mathcal{W}_{ n + 1 }^{m,j} - \mathcal{W}_{ n }^{m,j} ]
\bigr),
\exp\bigl(
    2 \, \bigl[ r - \tfrac{1}{2}\beta^2 \bigr] + \beta \, [ \mathcal{W}_{ n + 2 }^{m,j} - \mathcal{W}_{ n }^{m,j} ]
\bigr),
\\ & \quad\ \
\ldots,
\exp\bigl(
    100 \, \bigl[ r - \tfrac{1}{2}\beta^2 \bigr] + \beta \, [ \mathcal{W}_{ n + 100 }^{m,j} - \mathcal{W}_{ n }^{m,j} ]
\bigr)
\bigr)
,
\end{split}
\end{equation}
let $ \mathbb{F} = ( \mathbb{F}_n )_{ n \in \N_0 } $ be the
filtration generated by $ \mathcal{Y}^{ 0, 1 } $,
and assume
for all
$ m, j \in \N $,
$ n \in \{ 0, 1, \ldots, N \} $,
$ s \in [ 0, T ] $,
$ x = ( x_1, \ldots, x_d ) \in \R^d $
that
$ T \in \N $,
$ d = 100 $,
$ N = T $,
$ M =
1200 \, \mathbbm{1}_{ [ 1, 150 ] }(T) 
+ 1500 \, \mathbbm{1}_{ ( 150, 250 ] }(T)
+ 3000 \, \mathbbm{1}_{ ( 250, \infty ) }(T)
$,
$ \mathcal{X}_n^{ m - 1, j } = \mathcal{Y}_n^{ m - 1, j } $,
$ J_0 =$ 4,096,000,
$ J_m =
8192 \, \mathbbm{1}_{ [ 1, 150 ] }(T) 
+ 4096 \, \mathbbm{1}_{ ( 150, 250 ] }(T)
+ 512 \, \mathbbm{1}_{ ( 250, \infty ) }(T)
$,
$ \varepsilon = 10^{-8} $,
$ \gamma_m =
5 \, [ 10^{ -2 } \, \mathbbm{1}_{ [ 1, \nicefrac{M}{3} ] }(m)
+ 10^{ -3 } \, \mathbbm{1}_{ ( \nicefrac{M}{3}, \nicefrac{2M}{3} ] }(m)
+ 10^{ -4 } \, \mathbbm{1}_{ ( \nicefrac{2M}{3}, \infty ) }(m) ] $,
and
$ g(s,x) = e^{ - r s } \, x_{100} $.
The random variable $ \mathcal{P} $
provides approximations of the real number
\begin{equation}
\label{eq:example_TsitsiklisVanRoy1999}
\sup\!\left\{ 
\E\Bigl[
  e^{ - r \tau } 
  \tfrac{S_{ \tau }^{0,1}}{S_{ \tau - 100 }^{0,1}}
\Bigr]
\colon
\substack{
  \tau \colon \Omega \to \{ 0, 1, \ldots, T \} \text{ is an}
\\
  ( \mathbb{F}_n )_{ n \in \{ 0, 1, \ldots, T \} }\text{-stopping time}
}
\right\}
.
\end{equation}
In Table~\ref{tab:example_TsitsiklisVanRoy1999},
we show approximations
of the mean and of the standard deviation of $ \mathcal{P} $
and
the average runtime in seconds needed for calculating one realisation of $ \mathcal{P} $
for $ T \in \{ 100, 150, 200, 250, 1000 \} $.
For each case, the calculations of the results in Table~\ref{tab:example_TsitsiklisVanRoy1999}
are based on $10$ independent realisations of $ \mathcal{P} $,
which were obtained from an implementation in {\sc Python}
(cf.\ {\sc Python} code~\ref{code:example_TsitsiklisVanRoy1999} in Subsection~\ref{sec:code_example_TsitsiklisVanRoy1999} below).
Note that in this example time is measured in days
and that, roughly speaking,~\eqref{eq:example_TsitsiklisVanRoy1999}
corresponds to the price of a financial derivative
which, if the holder decides to exercise,
pays off the amount given by the ratio between
the current underlying stock price
and the underlying stock price 100 days ago
(cf.\ \cite[Section~IV]{TsitsiklisVanRoy1999} for more details).
According to~\cite[Subsection IV.D]{TsitsiklisVanRoy1999},
the value $ 1.282 $
is a lower bound for the price
\begin{equation}
\label{eq:example_TsitsiklisVanRoy1999_inf}
\sup\!\left\{ 
\E\Bigl[
  e^{ - r \tau } 
  \tfrac{S_{ \tau }^{0,1}}{S_{ \tau - 100 }^{0,1}}
\Bigr]
\colon
\substack{
  \tau \colon \Omega \to \N_0 \text{ is an}
\\
  \mathbb{F}\text{-stopping time}
}
\right\},
\end{equation}
which corresponds to the price~\eqref{eq:example_TsitsiklisVanRoy1999}
in the case of an infinite time horizon.
Since the mean of $ \mathcal{P} $ is a lower bound for the price~\eqref{eq:example_TsitsiklisVanRoy1999},
which, in turn, is a lower bound for the price~\eqref{eq:example_TsitsiklisVanRoy1999_inf},
a higher value indicates a better approximation of the price~\eqref{eq:example_TsitsiklisVanRoy1999_inf}.
In addition, observe that the price~\eqref{eq:example_TsitsiklisVanRoy1999}
is non-decreasing in $ T $.
While in our numerical simulations the approximate value of the mean of $ \mathcal{P} $
is less or equal than $ 1.282 $ for comparatively small time horizons, i.e., for $ T \leq 150 $,
it is already higher for slightly larger time horizons, i.e., for $ T \geq 200 $
(cf.~Table~\ref{tab:example_TsitsiklisVanRoy1999}).

\begin{table}[!hb]
\centering
\begin{tabular}{|c|c|c|c|}
\hline
Time&Mean&Standard&Average runtime\\
horizon $T$&of $ \mathcal{P} $&deviation&in sec. for one\\
&&of $ \mathcal{P} $&realisation of $ \mathcal{P} $\\
\hline
100 & 1.2721 & 0.0001 & 212.2 \\
150 & 1.2820 & 0.0002 & 683.7 \\
200 & 1.2895 & 0.0002 & 293.0 \\
250 & 1.2959 & 0.0001 & 408.4 \\
1000 & 1.2998 & 0.0007 & 864.5 \\
\hline
\end{tabular}
\caption{\label{tab:example_TsitsiklisVanRoy1999}%
Numerical simulations of the algorithm of Framework~\ref{algo:general}
for pricing the path-dependent financial derivative
from the example in Subsection~\ref{sec:example_TsitsiklisVanRoy1999}
(cf.\ {\sc Python} code~\ref{code:example_TsitsiklisVanRoy1999} in Subsection~\ref{sec:code_example_TsitsiklisVanRoy1999} below).
According to~\cite[Subsection IV.D]{TsitsiklisVanRoy1999},
the value $ 1.282 $
is a lower bound for the price~\eqref{eq:example_TsitsiklisVanRoy1999_inf}.}
\end{table}

\section{{\sc Python} source codes}
\label{sec:source_codes}

In Subsections~\ref{sec:code_max_std}--\ref{sec:code_example_TsitsiklisVanRoy1999} below we present {\sc Python} source codes associated to the numerical simulations in Subsections~\ref{sec:max_std}, \ref{sec:max-call_equity}, \ref{sec:strangle_spread}, \ref{sec:example_Dupire}, and~\ref{sec:example_TsitsiklisVanRoy1999} above. The following {\sc Python} source code, {\sc Python} code~\ref{code:common} in Subsection~\ref{sec:code_common} below,
contains the main part of the implementation of the
algorithm in Framework~\ref{algo:general}
and Framework~\ref{algo:examples-setting}
and
is employed in each of the {\sc Python} source codes in Subsections~\ref{sec:code_max_std}--\ref{sec:code_example_TsitsiklisVanRoy1999}.

\subsection{A {\sc Python} source code for the algorithm}
\label{sec:code_common}
\lstset{caption={\it common.py}}

\begin{lstlisting}[label={code:common}]
import tensorflow as tf
from tensorflow.python.training.moving_averages import assign_moving_average


def neural_net(x, neurons, is_training, dtype=tf.float32, decay=0.9):
    def batch_normalization(y):
        shape = y.get_shape().as_list()
        y = tf.reshape(y, [-1, shape[1] * shape[2]])
        beta = tf.compat.v1.get_variable(
            name='beta', shape=[shape[1] * shape[2]],
            dtype=dtype, initializer=tf.zeros_initializer())
        gamma = tf.compat.v1.get_variable(
            name='gamma', shape=[shape[1] * shape[2]],
            dtype=dtype, initializer=tf.ones_initializer())
        mv_mean = tf.compat.v1.get_variable(
            'mv_mean', [shape[1] * shape[2]],
            dtype=dtype, initializer=tf.zeros_initializer(),
            trainable=False)
        mv_var = tf.compat.v1.get_variable(
            'mv_var', [shape[1] * shape[2]],
            dtype=dtype, initializer=tf.ones_initializer(),
            trainable=False)
        mean, variance = tf.nn.moments(y, [0], name='moments')
        tf.compat.v1.add_to_collection(
            tf.compat.v1.GraphKeys.UPDATE_OPS,
            assign_moving_average(mv_mean, mean, decay,
                                  zero_debias=True))
        tf.compat.v1.add_to_collection(
            tf.compat.v1.GraphKeys.UPDATE_OPS,
            assign_moving_average(mv_var, variance, decay,
                                  zero_debias=False))
        mean, variance = tf.cond(is_training, lambda: (mean, variance),
                                 lambda: (mv_mean, mv_var))
        y = tf.nn.batch_normalization(y, mean, variance, beta, gamma, 1e-6)
        return tf.reshape(y, [-1, shape[1], shape[2]])

    def fc_layer(y, out_size, activation, is_single):
        shape = y.get_shape().as_list()
        w = tf.compat.v1.get_variable(
            name='weights',
            shape=[shape[2], shape[1], out_size],
            dtype=dtype,
            initializer=tf.initializers.glorot_uniform())
        y = tf.transpose(tf.matmul(tf.transpose(y, [2, 0, 1]), w),
                         [1, 2, 0])
        if is_single:
            b = tf.compat.v1.get_variable(
                name='bias',
                shape=[out_size, shape[2]],
                dtype=dtype,
                initializer=tf.zeros_initializer())
            return activation(y + b)
        return activation(batch_normalization(y))

    x = batch_normalization(x)
    for i in range(len(neurons)):
        with tf.compat.v1.variable_scope('layer_' + str(i)):
            x = fc_layer(x, neurons[i],
                         tf.nn.relu if i < len(neurons) - 1
                         else tf.nn.sigmoid, False)
    return x


def deep_optimal_stopping(x, t, n, g, neurons, batch_size, train_steps,
                          mc_runs, lr_boundaries, lr_values, beta1=0.9,
                          beta2=0.999, epsilon=1e-8, decay=0.9):
    is_training = tf.compat.v1.placeholder(tf.bool, [])
    p = g(t, x)
    nets = neural_net(tf.concat([x[:, :, :-1], p[:, :, :-1]], axis=1),
                      neurons, is_training, decay=decay)

    u_list = [nets[:, :, 0]]
    u_sum = u_list[-1]
    for k in range(1, n - 1):
        u_list.append(nets[:, :, k] * (1. - u_sum))
        u_sum += u_list[-1]

    u_list.append(1. - u_sum)
    u_stack = tf.concat(u_list, axis=1)
    p = tf.squeeze(p, axis=1)
    loss = tf.reduce_mean(tf.reduce_sum(-u_stack * p, axis=1))
    idx = tf.argmax(tf.cast(tf.cumsum(u_stack, axis=1) + u_stack >= 1,
                            dtype=tf.uint8),
                    axis=1, output_type=tf.int32)
    stopped_payoffs = tf.reduce_mean(
        tf.gather_nd(p, tf.stack([tf.range(0, batch_size, dtype=tf.int32),
                                  idx], axis=1)))

    global_step = tf.Variable(0)
    learning_rate = tf.compat.v1.train.piecewise_constant(global_step,
                                                          lr_boundaries,
                                                          lr_values)
    optimizer = tf.compat.v1.train.AdamOptimizer(learning_rate,
                                                 beta1=beta1,
                                                 beta2=beta2,
                                                 epsilon=epsilon)
    update_ops = tf.compat.v1.get_collection(
        tf.compat.v1.GraphKeys.UPDATE_OPS)
    with tf.control_dependencies(update_ops):
        train_op = optimizer.minimize(loss, global_step=global_step)

    with tf.compat.v1.Session() as sess:

        sess.run(tf.compat.v1.global_variables_initializer())

        for _ in range(train_steps):
            sess.run(train_op, feed_dict={is_training: True})

        px_mean = 0.
        for _ in range(mc_runs):
            px_mean += sess.run(stopped_payoffs,
                                feed_dict={is_training: False})

    return px_mean / mc_runs
\end{lstlisting}

\subsection[A {\sc Python} source code associated to Subsection~\ref{sec:max_std}]{A {\sc Python} source code associated to the numerical simulations in Subsection~\ref{sec:max_std}}
\label{sec:code_max_std}
\lstset{caption={\it example\_4\_4\_1\_1.py}}

\begin{lstlisting}[label={code:max_std}]
import tensorflow as tf
import numpy as np
import time
from common import deep_optimal_stopping
from tensorflow.python.framework.ops import disable_eager_execution

disable_eager_execution()

T, N, K = 3., 9, 100.
r, delta, beta = 0.05, 0.1, 0.2
batch_size = 8192
lr_values = [0.05, 0.005, 0.0005]
mc_runs = 500


def g(s, x):
    return tf.exp(-r * s) \
           * tf.maximum(tf.reduce_max(x, axis=1, keepdims=True) - K, 0.)


_file = open('example_4_4_1_1.csv', 'w')
_file.write('dim, run, mean, time\n')
for d in [2, 3, 5, 10, 20, 30, 50, 100, 200, 500]:
    for s_0 in [90., 100., 110.]:
        for run in range(10):
            tf.compat.v1.reset_default_graph()
            t0 = time.time()
            neurons = [d + 50, d + 50, 1]
            train_steps = 3000 + d
            lr_boundaries = [int(500 + d / 5), int(1500 + 3 * d / 5)]
            W = tf.cumsum(tf.compat.v1.random_normal(
                shape=[batch_size, d, N],
                stddev=np.sqrt(T / N)), axis=2)
            t = tf.constant(np.linspace(start=T / N, stop=T, num=N,
                                        endpoint=True, dtype=np.float32))
            X = tf.exp((r - delta - beta ** 2 / 2.) * t + beta * W) * s_0
            px_mean = deep_optimal_stopping(
                X, t, N, g, neurons, batch_size,
                train_steps, mc_runs,
                lr_boundaries, lr_values, epsilon=0.1)
            t1 = time.time()
            _file.write('%i, %i, %f, %f\n' % (d, run, px_mean, t1 - t0))
_file.close()
\end{lstlisting}

\pagebreak

\subsection[A {\sc Python} source code associated to Subsection~\ref{sec:max-call_equity}]{A {\sc Python} source code associated to the numerical simulations in Subsection~\ref{sec:max-call_equity}}
\label{sec:code_max-call_equity}
\lstset{caption={\it example\_4\_4\_1\_3.py}}

\begin{lstlisting}[label={code:max-call_equity}]
import tensorflow as tf
import numpy as np
import time
from common import deep_optimal_stopping
from tensorflow.python.framework.ops import disable_eager_execution

disable_eager_execution()

tau, N = 30. / 365., 10
r, beta = 0.05 * tau, 0.4 * np.sqrt(tau)
batch_size = 8192
train_steps = 1600
lr_boundaries = [400, 800]
lr_values = [0.05, 0.005, 0.0005]
mc_runs = 500


def g(s, x, k):
    return tf.exp(-r * s) \
           * tf.maximum(tf.reduce_max(x, axis=1, keepdims=True) - k, 0.)


_file = open('example_4_4_1_3.csv', 'w')
_file.write('dim, T, K, run, mean, time\n')
for d in [10, 400]:
    for T in [1., 4., 7.] if d == 10 else [12.]:
        for K in [35., 40., 45.]:
            for run in range(10):
                tf.compat.v1.reset_default_graph()
                t0 = time.time()
                Q = np.ones([d, d], dtype=np.float32) * 0.5
                np.fill_diagonal(Q, 1.)
                L = tf.constant(np.linalg.cholesky(Q).transpose())
                neurons = [d + 50, d + 50, 1]
                W = tf.matmul(tf.compat.v1.random_normal(
                    shape=[batch_size * N, d],
                    stddev=np.sqrt(T / N)), L)
                W = tf.cumsum(tf.transpose(tf.reshape(
                    W, [batch_size, N, d]),
                    [0, 2, 1]), axis=2)
                t = tf.constant(np.linspace(
                    start=T / N, stop=T, num=N,
                    endpoint=True, dtype=np.float32))
                X = tf.exp((r - beta ** 2 / 2.) * t + beta * W) * 40.
                px_mean = deep_optimal_stopping(
                    X, t, N, lambda s, x: g(s, x, K),
                    neurons, batch_size,
                    train_steps, mc_runs,
                    lr_boundaries, lr_values, epsilon=1e-3)
                t1 = time.time()
                _file.write('%i, %f, %f, %i, %f, %f\n'
                            % (d, T, K, run, px_mean, t1 - t0))
_file.close()
\end{lstlisting}

\subsection[A {\sc Python} source code associated to Subsection~\ref{sec:strangle_spread}]{A {\sc Python} source code associated to the numerical simulations in Subsection~\ref{sec:strangle_spread}}
\label{sec:code_strangle_spread}
\lstset{caption={\it example\_4\_4\_2.py}}

\begin{lstlisting}[label={code:strangle_spread}]
import tensorflow as tf
import numpy as np
import time
from common import deep_optimal_stopping
from tensorflow.python.framework.ops import disable_eager_execution

disable_eager_execution()

T, N = 1., 48
d, r = 5, 0.05
batch_size = 8192
neurons = [50, 50, 1]
train_steps = 750
lr_boundaries = [250, 500]
lr_values = [0.05, 0.005, 0.0005]
mc_runs = 500
Q = np.array([[0.3024, 0.1354, 0.0722, 0.1367, 0.1641],
              [0.1354, 0.2270, 0.0613, 0.1264, 0.1610],
              [0.0722, 0.0613, 0.0717, 0.0884, 0.0699],
              [0.1367, 0.1264, 0.0884, 0.2937, 0.1394],
              [0.1641, 0.1610, 0.0699, 0.1394, 0.2535]], dtype=np.float32)


def g(s, x):
    x = tf.reduce_mean(x, axis=1, keepdims=True)
    return tf.exp(-r * s) * (-tf.maximum(75. - x, 0.)
                             + tf.maximum(90. - x, 0.)
                             + tf.maximum(x - 110., 0.)
                             - tf.maximum(x - 125., 0.))


_file = open('example_4_4_2.csv', 'w')
_file.write('run, mean, time\n')
for run in range(10):
    tf.compat.v1.reset_default_graph()
    t0 = time.time()
    sigma = tf.expand_dims(tf.constant(np.linalg.norm(Q, axis=1)), axis=1)
    W = tf.matmul(tf.compat.v1.random_normal(shape=[batch_size * N, d],
                                   stddev=np.sqrt(T / N)), Q)
    W = tf.cumsum(tf.transpose(tf.reshape(W, [batch_size, N, d]),
                               [0, 2, 1]), axis=2)
    t = tf.constant(np.linspace(start=T / N, stop=T, num=N,
                                endpoint=True, dtype=np.float32))
    X = tf.exp((r - sigma ** 2 / 2.) * t + W) * 100.
    px_mean = deep_optimal_stopping(X, t, N, g,
                                    neurons, batch_size,
                                    train_steps, mc_runs,
                                    lr_boundaries, lr_values)
    t1 = time.time()
    _file.write('%i, %f, %f\n' % (run, px_mean, t1 - t0))
_file.close()
\end{lstlisting}

\subsection[A {\sc Python} source code associated to Subsection~\ref{sec:example_Dupire}]{A {\sc Python} source code associated to the numerical simulations in Subsection~\ref{sec:example_Dupire}}
\label{sec:code_example_Dupire}
\lstset{caption={\it example\_4\_4\_3.py}}

\begin{lstlisting}[label={code:example_Dupire}]
import tensorflow as tf
import numpy as np
import time
from common import deep_optimal_stopping
from tensorflow.python.framework.ops import disable_eager_execution

disable_eager_execution()


T, L = 1., 10
d, r = 5, 0.05
delta, K = 0.1, 100.
neurons = [50, 50, 1]
train_steps = 1200
lr_boundaries = [400, 800]
lr_values = [0.05, 0.005, 0.0005]


def g(s, x):
    return tf.exp(-r * s) * tf.maximum(
        K - tf.reduce_mean(tf.exp(x), axis=1, keepdims=True), 0.)


_file = open('example_4_4_3.csv', 'w')
_file.write('L, N, delta, type, run, mean, time\n')
for N in [5, 10, 50, 100]:
    for delta in [0., 0.1]:
        for _type in ['A', 'E'] if N == 10 else ['A']:
            for run in range(10):
                tf.compat.v1.reset_default_graph()
                t0 = time.time()
                if _type == 'A':
                    batch_size, mc_runs = 8192, 500
                else:
                    batch_size, mc_runs = 100000, 10000

                W = tf.compat.v1.random_normal(shape=[batch_size, d, L],
                                     stddev=np.sqrt(T / L))
                X = []
                Y = tf.constant(np.log(100.), tf.float32)

                for l in range(L):
                    beta = 0.6 * np.exp(
                        -0.05 * np.sqrt(l * T / L)) \
                           * (1.2 - tf.exp(-0.1 * l * T / L - 0.001
                                * (np.exp(r * l * T / L)
                                    * tf.exp(Y) - 100.) ** 2))
                    _b = (r - delta - beta ** 2 / 2.)
                    _w = beta * W[:, :, l]
                    if N >= L:
                        X.append(Y)
                        for n in range(int(N / L - 1)):
                            _x = Y + (n + 1) * T / N * _b \
                                 + (n + 1.) * L / N * _w
                            X.append(_x)
                    elif l % (L / N) == 0:
                        X.append(Y)
                    Y += T / L * _b + _w

                if _type == 'A':
                    X.append(Y)
                    X = tf.stack(X[1:], axis=2)
                    t = tf.constant(np.linspace(
                        start=T / N, stop=T, num=N,
                        endpoint=True, dtype=np.float32))

                    px_mean = deep_optimal_stopping(X, t, N, g,
                                                    neurons,
                                                    batch_size,
                                                    train_steps,
                                                    mc_runs,
                                                    lr_boundaries,
                                                    lr_values)
                else:
                    t = tf.constant(T)
                    px = tf.reduce_mean(tf.maximum(g(t, Y), 0.))

                    with tf.compat.v1.Session() as sess:
                        px_mean = 0.
                        for k in range(mc_runs):
                            px_mean += sess.run(px)
                        px_mean /= mc_runs

                t1 = time.time()
                _file.write('%i, %i, %f, %s, %i, %f, %f\n'
                            % (L, N, delta, _type, run,
                               px_mean, t1 - t0))
                _file.flush()
_file.close()
\end{lstlisting}

\subsection[A {\sc Python} source code associated to Subsection~\ref{sec:example_TsitsiklisVanRoy1999}]{A {\sc Python} source code associated to the numerical simulations in Subsection~\ref{sec:example_TsitsiklisVanRoy1999}}
\label{sec:code_example_TsitsiklisVanRoy1999}
\lstset{caption={\it example\_4\_4\_4.py}}

\begin{lstlisting}[label={code:example_TsitsiklisVanRoy1999}]
import tensorflow as tf
import numpy as np
import time
from common import deep_optimal_stopping
from tensorflow.python.framework.ops import disable_eager_execution

disable_eager_execution()

d, r = 100, 0.0004
sigma = 0.02
batch_size = int(8192)
neurons = [150, 150, 1]
train_steps = 1200
lr_boundaries = [400, 800]
lr_values = [0.05, 0.005, 0.0005]
mc_runs = int(500)


def g(s, x):
    return tf.exp(-r * s) * tf.expand_dims(x[:, 99, :], axis=1)


_file = open('example_4_4_4.csv', 'w')
_file.write('N, run, mean, time\n')
for N in [100, 150, 250, 1000, 2000]:
    for run in range(10):
        tf.compat.v1.reset_default_graph()
        t0 = time.time()
        T = np.float32(N)
        W = tf.compat.v1.random_normal(
            shape=[batch_size, N + 100],
            stddev=np.sqrt((N + 100.) / N))
        W = tf.cumsum(tf.concat([tf.zeros([batch_size, 1]),
                                 W], axis=1), axis=1)
        _W = tf.reshape(W[:, 1:], [batch_size, N + 100, 1, 1])
        patches = tf.compat.v1.extract_image_patches(
            _W, ksizes=[1, d, 1, 1],
            strides=[1, 1, 1, 1],
            rates=[1, 1, 1, 1],
            padding='VALID')
        W = tf.squeeze(patches) - tf.expand_dims(W[:, :N + 1], axis=2)
        W = tf.transpose(W, [0, 2, 1])
        n = tf.expand_dims(tf.constant(np.linspace(
            start=1, stop=d, num=d,
            endpoint=True,
            dtype=np.float32)), axis=1)
        X = tf.exp((r - sigma ** 2 / 2.) * n + sigma * W)
        t = tf.constant(np.linspace(start=0, stop=T, num=N + 1,
                                    endpoint=True, dtype=np.float32))
        px_mean = deep_optimal_stopping(X, t, N + 1, g,
                                        neurons, batch_size,
                                        train_steps, mc_runs,
                                        lr_boundaries, lr_values)
        t1 = time.time()
        _file.write('%i, %i, %f, %f\n' % (N, run, px_mean, t1 - t0))
        _file.flush()
_file.close()
\end{lstlisting}

\section*{Acknowledgements}

This project has been partially supported through the ETH Research Grant \mbox{ETH-47 15-2}
`Mild stochastic calculus and numerical approximations for nonlinear stochastic evolution equations with L\'evy noise',
by the project
`Construction of New Smoothness Spaces on Domains'
(project number I 3403)
funded by the Austrian Science Fund (FWF),
and
by the project
`Deep artificial neural network approximations for stochastic partial differential equations: Algorithms and convergence proofs'
(project number 184220)
funded by the Swiss National Science Foundation (SNSF).

\bibliographystyle{acm}
\bibliography{bibfile}

\end{document}